\documentclass [14pt]{amsart}
\usepackage{amsmath}
\usepackage{amssymb}
\usepackage{amscd}
\usepackage{epsfig}
\usepackage[T1]{fontenc}
\usepackage{ae,aecompl}

\usepackage{mathabx}

 \theoremstyle{bold}
\newtheorem{theorem}{Theorem}[section]
\newtheorem{proposition}{Proposition}[section]
\newtheorem{lemma}{Lemma}[section]
\newtheorem{corollary}{Corollary}[section]
\newtheorem{conjecture}{Conjecture}[section]

\theoremstyle{shy}
\newtheorem{definition}{Definition}[section]
\newtheorem{remark}{Remark}[section]
\newtheorem{example}{Example}[section]

\DeclareMathOperator{\per}{per}

\DeclareMathOperator{\dens}{dens}

\DeclareMathOperator{\supp}{supp}

\DeclareMathOperator{\Real}{Re}
\DeclareMathOperator{\Imag}{Im}

\DeclareMathOperator{\ord}{ord}

\newcommand{\vL}{\varLambda}
\newcommand{\vG}{\varGamma}

\newcommand{\cL}{\mathcal{L}}

\newcommand{\cU}{\mathcal{U}}

\newcommand{\AAA}{\mathbb{A}}

\newcommand{\ZZ}{\mathbb{Z}}
\newcommand{\QQ}{\mathbb{Q}\ts}
\newcommand{\RR}{\mathbb{R}}
\newcommand{\CC}{\mathbb{C}}

\newcommand{\TT}{\mathbb{T}}

\newcommand{\XX}{\mathbb{X}}

\newcommand{\oplam}{\mbox{\Large $\curlywedge$}}

\newcommand{\dd}{\,\mathrm{d}}
\newcommand{\ee}{\ts\mathrm{e}\ts}
\newcommand{\ii}{\ts\mathrm{i}\ts}

\newcommand{\ts}{\hspace{0.5pt}}
\newcommand{\nts}{\hspace{-0.5pt}}

\newcommand{\myfrac}[2]{\frac{\raisebox{-2pt}{$#1$}}{\raisebox{0.5pt}{$#2$}}}

\newcommand{\exend}{\hfill$\Diamond$}

\begin{document}

\title{Almost Periodic Measures and
Meyer Sets}

\author{Nicolae Strungaru}
\address{Department of Mathematical Sciences, MacEwan University
\\
10700 “ 104 Avenue, Edmonton, AB, T5J 4S2;\\
and \\
Institute of Mathematics ``Simon Stoilow'' \\
Bucharest, Romania}
\email{strungarun@macewan.ca}
\urladdr{http://academic.macewan.ca/strungarun/}

\begin{abstract}
In the first part, we construct a cut and project scheme from a family
$\{P_\varepsilon\}$ of sets verifying four conditions. We use this
construction to characterise weighted Dirac combs defined by cut and project
schemes and by continuous functions on the internal groups in terms of
almost periodicity. We are also able to characterise those weighted Dirac
combs for which the internal function is compactly supported. Lastly,
using the same cut and project construction for $\varepsilon$-dual
sets, we are able to characterise Meyer sets in $\sigma$-compact
locally compact Abelian groups.
\end{abstract}

\maketitle

\section{Introduction}

In 1984, Shechtman et al.\ announced the discovery of a solid with an
unusual diffraction pattern. While the diffraction was similar to the
one of a periodic crystal, exhibiting only bright spots (called Bragg
peaks) and no diffuse background, it also had fivefold symmetry which
is impossible in a perfect (periodic) crystal in $3$-space; see
\cite[Cor.~3.1]{S-TAO}. Solids with such unusual symmetries are called
quasicrystals. A quasicrystal cannot repeat periodically in all
directions, yet in order to produce many Bragg peaks, a large number
of its elements need to repeat in a highly ordered way.

Twelve years earlier, Y.\ Meyer introduced the concept of a harmonious
set in \cite{S-MEY}. These sets exhibit long-range order and are
usually aperiodic. Meyer also introduced cut and project schemes as a
simple way of generating such sets and studied the relation between
harmonious sets and model sets. In the context of icosahedral
structures, the projection method was independently discovered by
Kramer and coworkers \cite{S-Kra82,S-KN84}. The relevance of Meyer's
work to quasicrystals was realised only in the 1990s; see
\cite{S-MOO,S-LAG}.

A large class of sets produced by the cut and project scheme, called
regular model sets, produce a pure point diffraction pattern
\cite{S-Martin1}. Conversely, under suitable additional assumptions,
any set with a pure point diffraction pattern can be obtained from a
cut and project scheme \cite{S-BLM}.

A harmonious set is nowadays also called a Meyer set. The full
characterisation of Meyer sets by Moody \cite{S-MOO}, completing the
previous work of Meyer and Lagarias, is surprising, Meyer sets being
characterised in different ways: via duality, discrete geometry, cut
and project schemes and the theory of almost lattices. All these
different characterisations emphasise that Meyer sets have very strong
long-range order and, as it was recently shown, this order shows in
the diffraction pattern as a relatively dense set of Bragg peaks
\cite{S-NS1} which is highly ordered \cite{S-NS2}.

The connection between cut and project schemes and almost periodic
measures has became apparent over the last few years. Pure point
diffraction is equivalent to the strongly almost periodicity of the
autocorrelation measure \cite{S-ARMA}, see also \cite{S-MoSt} in this volume, but in general strongly almost
periodicity is a concept which is hard to check. Baake and Moody
\cite{S-BM} introduced a stronger concept of almost periodicity,
namely norm-almost periodicity. For a regular model set, their work
shows that norm-almost periodicity of the autocorrelation follows from
the fact that the covariogram of the window is continuous and
compactly supported in the internal space. A nice yet intriguing
consequence is that the pure point nature of the diffraction spectra
of a regular model set is a consequence of the continuity of the
covariogram. They also showed that, if the autocorrelation of a Meyer
set is norm-almost periodic, there is a natural way of constructing a
cut and project scheme from the sets of almost periods of the
autocorrelation.

Motivated by some random tiling examples, Richard \cite{S-CR} and
Lenz--Richard \cite{S-LR} studied a class of (possibly dense) weighted
combs coming from a cut and project scheme. More exactly, given a cut
and project scheme $(G \times H, \cL)$ and a function $g$ on
$H$, we can define a weighted comb $\omega_g$ on $G$ by
\[
   \omega_g \, := \sum_{(x, x^{\star\nts}) \in \cL} g(x^{\star\nts})
    \ts \delta_x \ts .
\]
We will often refer to such a measure as a weighted model comb.

Since we need the measure $\omega_g$ to be translation bounded, $g$
has to be decaying fast \cite{S-CR}. Note that, if $g =1_W$ is the
characteristic function of a window $W$, then $\omega_g$ is exactly
the Dirac comb of the model set $\oplam(W)$; while if $g$ is the
covariogram of the window, we get exactly the construction of Baake
and Moody, and in this case $\omega_g$ is the autocorrelation of the
model set.

When one studies regular model sets, most of the issues appear around
the boundary of the window. These issues seem to come from the fact
that $\partial W$ is exactly the set of discontinuities of $1_W$. All
these considerations motivated Richard and Lenz--Richard to look at
the combs $\omega_g$ for continuous functions with fast decay. They
showed that, in this case, $\omega_g$ is strongly almost periodic
\cite{S-LR}.

A last form of almost periodicity is the sup-almost periodicity which
appeared in a natural way in the study of the pure point spectrum of a
Meyer set \cite{S-NS2}. While this concept is more recent than
\cite{S-BM,S-LR}, looking back to those papers it also appears hidden
in the proof of almost periodicity of the weighted model combs as a
consequence of continuity of the corresponding function on the
internal group.

Norm-almost periodicity is the strongest form implying both strongly and
sup-almost periodicity. In general, strongly and sup-almost periodicity
are weaker and not related to each other, but for a measure with
sparse support, and in particular when the support of the measure is a
Meyer set, we will see that the three are equivalent. Because of its
connection to pure point diffraction, strongly almost periodicity is
the most natural to study, but we will see in this paper that sup-
and norm-almost periodicity are directly related to the cut and
project formalism.

The reason why we need sup- or norm-almost periodicity to get the cut
and project scheme, and why strongly almost periodicity does not seem to
work, is interesting and hidden in the construction. Any group $G$
acts by translations on the space of translation bounded measures
${\mathcal M}^\infty(G)$. The group action is continuous in the strong
topology, but not in the norm- nor sup-topology. It is exactly this
discontinuity which allows us to construct the scheme: we will build
the cut and project scheme from the set of almost periods of an almost
periodic measure using the axioms (A1)-(A4) from
Section~\ref{construction cp}. And while in general (A1) and (A3) hold
for any (pseudo)-metric and (A4) is equivalent to almost periodicity,
we cannot get (A2) unless the group action on ${\mathcal M}^\infty(G)$
is discontinuous.

This chapter is organised as follows.  In Section~\ref{construction
  cp}, we verify that the Baake--Moody construction of a cut and project scheme of
\cite{S-BM} works under much weaker assumptions. Given a family $\{
P_\varepsilon \}^{} _{ 0< \varepsilon <C}$ which satisfies axioms
(A1)--(A4) listed at the beginning of the section, and a group which
contains all these sets, we construct a cut and project scheme in
which all these $\{ P_\varepsilon \}$ are model sets. In particular,
the Baake--Moody construction works for any discrete sup-almost periodic measure,
and moreover we show that the support function of any such measure
defines an uniformly continuous function on the internal groups $H$.

This will allow us the fully characterise all weighted model combs
$\omega_g$ given by uniformly continuous functions $g \in C_{0}(H)$ in
terms of sup-almost periodicity as follows.

\setcounter{theorem}{1}
\begin{theorem} Let\/ $\omega := \sum_{x \in \vG} \omega(\{x\})\,
  \delta_x$ be a regular measure on\/ $G$. Then, the following
  statements are equivalent.
  \begin{itemize}\itemsep=1pt
  \item[(i)] $\omega$ is sup-almost periodic.
  \item[(ii)] There exists a cut and project scheme\/ $(G\times H,
    \cL)$, and a uniformly continuous function\/ $g\in
    C_0(H)$ vanishing at\/ $\infty$, so that $\omega = \omega_g$.
  \end{itemize}
\end{theorem}
\setcounter{theorem}{0}

Under the conditions of Theorem~\ref{T1}, compactness in the internal
group can be characterised in terms of the sparseness of the support
of the measure, and this allows us to also characterise the weighted
model combs with $g \in C_{\mathsf{c}}(H)$ in terms of all three types
of almost periodicity in Theorem~\ref{T2}.

The Baake--Moody construction from Section~\ref{construction cp} also works for the
sets of $\varepsilon$-dual characters, if these sets are relatively
dense, and this allows us to give a characterisation of Meyer sets in
any $\sigma$-compact LCAG (locally compact Abelian group) similar to the one
in \cite{S-MOO}.  While many of the proofs in \cite{S-MOO} are based
on the geometry of $\RR^d$, and cannot be extended to arbitrary groups
$G$, the strong connection between sup-almost periodicity and cut and
project schemes from Theorem~\ref{T1} allows us find new proofs for
some of the equivalences. We are able to prove all the equivalences
except the one introduced by Lagarias \cite{S-LAG}. Moreover, if $G$
is also compactly generated, we also get this last equivalence. The
main result of Section~\ref{CHARMS} can be summarised as follows.

\setcounter{theorem}{7}
\begin{theorem}
  Let\/ $\vL \subset G$ be relatively dense, where\/ $G$ is a\/
  $\sigma$-compact LCAG. Then, the following properties are
  equivalent.
\begin{itemize}\itemsep=1pt
\item[(i)] $\vL$ is a subset of a model set.
\item[(ii)] $\vL$ is harmonious.
\item[(iii)] $\vL^\varepsilon$ is relatively dense for all\/ $0< \varepsilon <2$.
\item[(iv)] $\vL^\varepsilon$ is a model set for all\/
           $0< \varepsilon < \frac{\sqrt{3}}{2}$.
\item[(v)] $\vL^\varepsilon$ is a model set for some\/
           $0< \varepsilon < \frac{\sqrt{3}}{2}$.
\item[(vi)] $\vL^\varepsilon$ is relatively dense for some
           $0< \varepsilon < \frac{\sqrt{3}}{6}$.
\item[(vii)] $\vL$ is an almost lattice.
\item[(viii)] $\vL- \vL- \vL$ is locally finite.
\end{itemize}
\end{theorem}
\setcounter{theorem}{0}

If the group $G$ is compactly generated, it follows from \cite{S-BLM} that property\/ {\rm (viii)} can be
replaced by
\begin{itemize}
\item[(viii')] $\vL- \vL$ is uniformly discrete.
\end{itemize}

We summarise the equivalence in the case of compactly generated LCAG in Theorem~\ref{MeyCharCompGen}.

Using this characterisation, we are able to extend all results about
the diffraction of Meyer sets from $\RR^d$ to an arbitrary
$\sigma$-compact LCAG $G$ in Section~\ref{diffMS}, and obtain some new
characterisations for Meyer sets in terms of almost periods for
almost periodic measures.

As already mentioned, Baake and Moody \cite{S-BM} showed that, for a
regular model set, both the autocorrelation measure $\gamma$ and the
diffraction measure $\widehat{\gamma}$ can be obtained as weighted
model combs. For Meyer sets, we show in Section~\ref{sectfewnotes}
that the same holds for the strongly almost periodic component
$\gamma^{}_{\mathsf{s}}$ of the autocorrelation and for the pure point diffraction
spectrum $\widehat{\gamma}_{\mathsf{pp}}=\widehat{\gamma^{}_{\mathsf{s}}}$.

At the end of this chapter, we also take a closer look at uniform
density in regular model sets.

\section{The Baake--Moody construction of a cut and project
scheme}\label{construction cp}

Throughout the entire chapter, $G$ is assumed to be a $\sigma$-compact
LCAG.

The goal of this section is to construct a cut and project scheme from
a family of subsets $\{ P_\varepsilon \}$ and a subgroup $L$ of $G$ which
satisfy the conditions (A1)--(A5) listed below. It is easy to see that
the set of almost periods of an almost periodic element verifies
conditions (A1), (A3) and (A4), and if the group action is not continuous
in the topology we consider, condition (A2) usually holds for
$\varepsilon$ small enough. The (A5) requirement is the only restriction
we have on $L$, and we can easily satisfy it by picking any group $L$
so that $\langle \cup P_\varepsilon\rangle \leq L \leq G$. Choosing $L=G$ always
works, but sometimes it leads to too big schemes. We will see that in
fact (A5) is not a restriction, but a degree of freedom. We will provide
some examples which emphasise this point at the end of this section.

Our construction is a variation of the one from
\cite{S-BM}. While the construction in \cite{S-BM} is done for the sets of
almost periods of a norm-almost periodic autocorrelation, they only
used certain properties of those $P_\varepsilon$. In particular, this
construction holds for a larger class of almost periodic measures.
We will also see in Section~\ref{CHARMS} a completely different
setting in which this basic construction works.

Recall that a set $A \subset G$ is called \emph{locally finite} if, for
every compact set $K$, $A \cap K$ is a finite set. This is equivalent
to saying that $A$ is closed and discrete in $G$.

Given some $C>0$, a family $\{ P_\varepsilon \}^{} _{ 0< \varepsilon
  <C}$ and a subgroup $L \leq G$, our basic requirements are:
\begin{itemize}\label{A1-A5}
\item[(A1)] $0 \in P_\varepsilon=- P_\varepsilon$ for all $0 <
  \varepsilon < C$.
\item[(A2)] For each $0 < \varepsilon <C$, the set $P_\varepsilon$ is
  locally finite.
\item[(A3)] For each $\varepsilon, \varepsilon'$ with $\varepsilon
  +\varepsilon' <C$, we have
\[
   P_\varepsilon + P_{\varepsilon'} \,\subset\, P_{\varepsilon +\varepsilon'} \ts .
\]
\item[(A4)] For all $0< \varepsilon < C$, the set $P_\varepsilon$ is
  relatively dense.
\item[(A5)] For all $0< \varepsilon < C$, we have $P_\varepsilon \subset L$.
\end{itemize}

Before proceeding to the construction, let us try to understand these requirements.
Consider a cut and project scheme $(G \times H, \cL)$, with $H$ a metrisable group. Define $L:= \pi_1(\cL)$ to be the projection of the lattice $\cL$ on $G$.
This will be exactly our group $L$.

The star map $\star: L \to H$ has dense image. We can use this map to induce the topology of $H$ on $L$, and then $H$ is isomorphic to the completion of $L$ with respect to the induced topology.
Therefore, if we know $L \subset G$ and the topology induced on $L$ by $H$, we can recover $H$. We should also emphasize here that $L$ with the induced topology give us not only $H$ but also the $\star$ mapping. Indeed, the completion of $L$ with respect to this topology becomes is a LCAG which is isomorphic to $H$, and the isomorphism takes the inclusion into the $\star$ map.

Moreover, as
\[
\cL \, := \,
    \bigl\{ (t, t^\star) \mid t \in L \bigr\}  \ts ,
\]
as soon as we have $H$ and the $\star$ map, we get the lattice for free.

Next, let us look at the induced topology on $L$. As the topology on $H$ is given by a metric, we can define the balls $B_\varepsilon$ of radius $\varepsilon$ centered at $0 \in H$. The locally compactness property of $H$  implies that there exists some $C>0$ such that $B_\varepsilon$ is precompact in $H$ for all $0 < \varepsilon <C$.

Since $\{ B_\varepsilon \}$ is a fundamental basis of open sets at $0$ for the topology of $H$, we can define a basis of open sets at $0$ for the induced topology on $L$ by setting
\[
P_\varepsilon \, := \,
  \star^{-1}(  B_\varepsilon \cap L )  \ts ,
\]
Then, $\{ P_\varepsilon \}$ is a fundamental basis of open sets at $0$ for the induced topology on $L$. The five conditions (A1)--(A5) are immediate consequences of the choice of  $\{ B_\varepsilon \}^{} _{ 0< \varepsilon <C}$ and the fact that $(G \times H, \cL)$ is a cut and project scheme.

It is more interesting, and probably surprising, that the five conditions are sufficient to produce a cut and project schemes. The axioms (A1), (A3) and (A5) are both necessary and sufficient to define a topology on $L$. Once we take the completion $H$ of $L$ in this topology, we also get a subgroup $\cL \subset G \times H$. But in general the pair $(G \times H, \cL)$ is not a cut and project scheme. The axioms (A1), (A3) and (A5) are not enough to guarantee that $H$ is a LCAG and, more importantly, that $\cL$ is a lattice in $G \times H$.

In the construction above, it is easy to check that the discreteness of $\cL$ is equivalent to (A2) which makes this condition both natural and necessary. Moreover, once (A2) is assumed, the relatively denseness of $\cL$ becomes equivalent to (A4). This makes both conditions (A2) and (A4) necessary for the construction.

We now proceed to the construction. Starting with the condition (A1)--(A5), we will define in order an uniform topology on $L$, the group $H$ and the $\star$ map, a lattice $\cL$ and finally we prove that we get a cut and project scheme $(G \times H, \cL)$.

We start by showing that the sets $P_\varepsilon$ are decreasing with $\varepsilon$. This is an immediate consequence of (A3) and the fact that $0 \in P_\varepsilon$.

\begin{lemma}\label{L1111}
  If\/ {\rm (A1)} and\/ {\rm (A3)} hold and\/ $0< \varepsilon' <
  \varepsilon< C$, then
\[
    P_{\varepsilon'} \, \subset\,  P_{\varepsilon} \ts .
\]
\end{lemma}
\begin{proof}
Let $\varepsilon''=\varepsilon - \varepsilon' >0$.
By (A1), we have $0 \in P_{\varepsilon''}$, and by (A3), we have
$P_{\varepsilon'} + P_{\varepsilon''} \subset P_{\varepsilon' +\varepsilon''}=
P_{\varepsilon}$. Thus,
\[
   P_{\varepsilon'}\, =\, P_{\varepsilon'}+0 \, \subset\,
   P_{\varepsilon'} + P_{\varepsilon''}
   \,\subset\, P_{\varepsilon}\ts ,
\]
which completes the proof.
\end{proof}

The set $P_\varepsilon$ are locally finite by (A2), but this is usually a very weak requirement. There are many pointsets which are locally finite, but which are far from being uniformly discrete. One such simple example is
\[
\vL \, := \,
    \bigl\{ n+\myfrac{k}{n+1} \mid n \in \ZZ_+ \quad\text{and}\quad 0 \leq k \leq n-1 \bigr\}  \ts .
\]

In the following Lemma we show that, for $\varepsilon$ small, the sets $P_\varepsilon$ are actually uniformly discrete.

\begin{lemma}\label{L2222}
  If\/ {\rm (A1)}, {\rm (A2)} and\/ {\rm (A3)} hold, then
\begin{itemize}
\item[(i)] For all\/ $0 < \varepsilon < \frac{C}{2}$, the set\/
  $P_{\varepsilon}$ has finite local complexity. In particular,
  $P_{\varepsilon}$ is uniformly discrete.
\item[(ii)] For all\/ $0 < \varepsilon < \frac{C}{4}$, the set\/
  $P_{\varepsilon}-P_{\varepsilon}$ is uniformly discrete.
\end{itemize}
\end{lemma}

\begin{proof}
  Combining (A1) and (A3), we get
\[
    P_{\varepsilon}-P_{\varepsilon} \, =\,
     P_{\varepsilon}+P_{\varepsilon} \,\subset\,
    P_{2\varepsilon} \ts,
\]
whenever $2 \varepsilon <C$.

Claim (i). Since $2\varepsilon < C$, the set
$P_{2\varepsilon}$ is locally finite by (A2), thus
$P_{\varepsilon}-P_{\varepsilon}$ is locally finite.

Claim (ii). $2\varepsilon < \frac{C}{2}$ and thus, by (i),
$P_{2\varepsilon}$ is uniformly discrete. Hence
$P_{\varepsilon}-P_{\varepsilon}$ is also uniformly discrete.
\end{proof}

If $G = \RR^d$ for some $d$ and if the collection $\{ P_\varepsilon
\}$ satisfies requirements (A1)--(A4), combining Lemma~\ref{L2222}
with (A4) shows that the sets $\{P_{\varepsilon} \}^{}_{0 <
  \varepsilon < \frac{C}{4}}$ are Meyer sets. Thus, each of them is a
subset of a model set, and embedding some $P_\varepsilon$ in a model
set, all $\{ P_{\varepsilon'} \}^{}_{0 < \varepsilon' < \varepsilon}$
become, by Lemma~\ref{L1111}, subsets of the same model set. Thus, for
some $0< a <\frac{C}{4}$, all $\{ P_{\varepsilon} \}^{}_{0 <
  \varepsilon < a}$ are coming from the same cut and project scheme.

We will see later that the same is true under the more general setting of
$G$ being a $\sigma$-compact LCAG. Under the assumptions (A1)--(A4),
picking any $L$ which satisfies (A5), we will prove that all $\{ P_\varepsilon \}^{}_{ 0<
  \varepsilon <C}$ are subsets of model sets in the cut and project scheme we are constructing.

Next, we want to introduce a topology on $L$ for which the family $\{
P_\varepsilon \}^{}_{0 < \varepsilon < C}$ is a neighbourhood basis of pre-compact
sets. We first prove the following Lemma, which will imply later in the Section the pre-compactness of the
$P_\varepsilon$.\bigskip

\parbox{0.9\textwidth}{ \textbf{Note.}  \emph{For the remainder of the
    section, we assume\/ {\rm (A1)--(A5)} to hold and we construct the
  cut and project scheme.}\bigskip}

\begin{lemma}\label{L3333}
  Let\/ $0 < \varepsilon' < \varepsilon < C$. Then, there exists a
  finite set\/ $F$ so that
\[
   P_{\varepsilon} \,\subset\, P_{\varepsilon'} +F \ts .
\]
\end{lemma}
\begin{proof}
  Since $\varepsilon < C$, there exists $0 < \varepsilon''<
  \varepsilon' $ so that $\varepsilon + \varepsilon'' < C$.  Then, as
  $\varepsilon + \varepsilon'' < C$, by (A2) we have
\[
     P_{\varepsilon}+P_{\varepsilon''} \,\subset\,
     P_{\varepsilon+\varepsilon''} \ts .
\]
Now, we use the relatively denseness of $P_{\varepsilon''}$ and the
discreteness of $P_{\varepsilon+\varepsilon''}$ to find the finite
set.  Let $K \subset G$ be compact so that $P_{\varepsilon''}+K = G$.
Let $F := K \cap P_{\varepsilon+\varepsilon''}$. Then by
(A2) $F$ is finite. We show that $P_{\varepsilon} \subset P_{\varepsilon'} +F$.

For all $z \in P_\varepsilon$, there exists $y \in P_{\varepsilon''}$ and
$k \in K$ so that $z =y+k$. Then we get
\[
   k \,=\, z-y \,\in\, P_{\varepsilon}-P_{\varepsilon''}
   \,\subset\, P_{\varepsilon+\varepsilon''} \ts .
\]
Thus $k \in F$, and since $y  \in P_{\varepsilon''} \subset
P_{\varepsilon'}$, we are done.
\end{proof}

We are ready now to define the topology on $L$. As our goal is to make $\{ P_\varepsilon \}^{}_{0 <
    \varepsilon <C}$ a neighbourhood basis at $0$, the definition of the uniformity ${\mathcal U}$ is straightforward, we only need to check that indeed it is an uniformity.

\begin{proposition}\label{prop1.1} The family\/ $\{ P_\varepsilon \}^{}_{0 <
    \varepsilon <C}$ is a neighbourhood basis at\/ $0$ for an uniform
  topology on\/ $L$.
\end{proposition}

\begin{proof}
Let
\[
   U_\varepsilon \,:=\, \bigl\{ (x,y) \in L \times L \mid
   x-y \in P_\varepsilon \bigr\} .
\]
We show that the set ${\mathcal U} := \{ U \subset L \times L \mid
\exists\, 0 < \varepsilon <C \text{ so that } U_\varepsilon \subset U
\}$ is a uniformity on $L$. For this, we have to show that $\cU$
verifies the following five conditions.
\begin{itemize}
\item[(U1)] $U \in \cU$ , $U \subset V$ $\,\Longrightarrow\,$ $V \in \cU$.\\
  This is clear from the definition of $\cU$.
\item[(U2)] $V_1,\dots,V_n \in \cU$ $\,\Longrightarrow\,$ $V_1 \cap
  V_2 \cap
  \dots \cap V_n \in \cU$.\\
  For each $1 \leq i \leq n$, there exists an $0 < \varepsilon_i < C$
  so that $U_{\varepsilon_i} \subset V_i$. Let $\varepsilon := \min_{1
    \leq i \leq n} \{ \varepsilon_i \}$.  Then, $\varepsilon >0$; $U
  _\varepsilon \subset U_{\varepsilon_i}$ for all $1 \leq i \leq n$
  and hence $U _\varepsilon \subset V_1 \cap V_2 \cap \dots \cap V_n$.
  Thus, by the definition of $\cU$, we have $V_1 \cap V_2 \cap \dots
  \cap V_n \in \cU$.
\item[(U3)] $U \in \cU$ $\,\Longrightarrow\,$  $\Delta \subset U$.\\
  This follows from (A1).
\item[(U4)] $U \in \cU$ $\,\Longrightarrow\,$ $U^{-1} \in \cU$.\\
  Since $U \in \cU$, there exists $0 < \varepsilon < C$ so that
  $U_{\varepsilon} \subset U$. Then, by (A1), we have $U_{\varepsilon}
  \subset U^{-1}$.
\item[(U5)] $U \in \cU$ $\,\Longrightarrow\,$ $\exists\, V \in \cU $ so
  that $V_\circ V \subset U$.\\
  This follows from (A3).
\end{itemize}
This completes the proof.
\end{proof}

\begin{remark}\label{R1} It is easy to see that the family $\{ U_\varepsilon
  \}^{}_{\varepsilon}$ is a fundamental system of entourages for this
  topology, and one could use this to define the uniformity on
  $L$. This shows that $\{ P_\varepsilon \}^{}_{0 < \varepsilon <C}$
  is a neighbourhood basis at $0$ for the topology on $L$.

  Moreover, since the family $\{ P_\varepsilon \}^{}_{0 < \varepsilon
    <C}$ is a neighbourhood basis at $0$ for the topology on $L$,
  it follows that this topology is metrisable (since $\{ P_\frac{1}{n}
  \}^{}_{0 < \frac{1}{n} <C}$ is also a neighbourhood basis at $0$ for the topology on $L$). This tells us that the topology on $L$ can
  also be given by some pseudo-metric. In this case we can easily
  define a pseudo-metric, which gives the topology
\[
   d(x,y) \, := \,
   \begin{cases}
   \inf\{ \varepsilon | x-y \in P_\varepsilon \} \ts , &\text{if }x-y \in
   \bigcup_{0 < \varepsilon <C} P_\varepsilon\ts , \\ C\ts , & \text{otherwise.}
   \end{cases}
\]
Using (A1) and (A3), it is easy to prove that $d$ is a pseudo-metric
on $L$, which yields exactly the same topology as the uniformity from Proposition~\ref{prop1.1}.  \exend
\end{remark}

Let $H$ be the completion of $L$ in this topology. Then, there exists
a uniformly continuous map $\phi\! :\, L \longrightarrow H$ with the
following properties.
\begin{itemize}
\item[(C1)] $\phi(L)$ is dense in $H$.
\item[(C2)] The mapping $\phi$ is an open mapping from $L$ onto
  $\phi(L)$, the latter with the induced topology.
\item[(C3)] $\ker(\phi)$ is the closure of $0$ in this topology.
\end{itemize}
These three properties characterise the pair $(H, \phi)$ up to
isomorphism. Also note that
\[
\ker(\phi) \,=\, \bigcap_{\varepsilon > 0} P_\varepsilon \ts .
\]
Finally, for each $0 < \varepsilon < C$, there exists a neighbourhood
$V_\varepsilon \subset H$ of $0$ so that
\[
V_\varepsilon \cap \phi(L)\, =\,  \phi(P_\varepsilon) \ts .
\]
We will frequently make use of this fact.

We constructed the internal group $H$ of our cut and project scheme. In order to get a cut and project scheme, this must be a LCAG, and this is the next result we prove. Our proof below is very similar to the one in \cite[Prop.~1]{S-BM}.

\begin{proposition} The completion\/ $H$ is a LCAG.
\end{proposition}

\begin{proof}
  It is clear from the completion process that $H$ is an Abelian
  group.

  We now show that locally compactness follows from Lemma~\ref{L3333}. Indeed, for each $0 < \varepsilon < C$,
  $P_\varepsilon$ is a precompact set and hence, by the denseness of
  $\phi(L)$, the closure $\overline{\phi(P_\varepsilon)}=\overline{V_\varepsilon}$ is a
  compact neighbourhood of $0 \in H$.
\end{proof}

Next, we construct a lattice $\cL$ in $G \times H$, and this way we will get a cut and project scheme. As we said before, under this construction $\phi$ will be the $\star$ map, which makes the definition of $\cL$ in Lemma~\ref{latticecp} natural.

\begin{lemma}\label{latticecp} The set
\[
   \cL \, := \,
    \bigl\{ (t, \phi(t)) \mid t \in L \bigr\}
\]
is a lattice in $G \times H$.
\end{lemma}

\begin{proof}
  Since $\cL$ is a subgroup of $G \times H$, to prove that
  it is a lattice we have to show that it is discrete and relatively
  dense.

  \emph{Step 1: $\cL$ is discrete.}

  The idea for this step is simple: To prove that $\cL$ is discrete, we must find an open neighbourhood of $0$ in $G \times H$ which intersects $\cL$ only at $0$. As the topology on $G \times H$ is the product topology, and $\{ V_\varepsilon \}$ is a neighbourhood basis at $0$ in $H$, it suffices to find an open set $U \in G$ and some $\varepsilon >0$ such that $\cL \cap ( U \times V_\varepsilon) = \{ (0,0) \}$. Any point in this intersection has, the form $(t , \phi(t))$ with $t \in U \cap \left(P_\varepsilon + \ker(\phi)\right)$. Thus, if we can find $U$ and $\varepsilon$ such that $U \cap \left( P_\varepsilon + \ker(\phi) \right)= \{ 0 \}$, the discreteness of $\cL$ follows. As for each $\varepsilon' >0$ we have $P_\varepsilon + \ker(\phi) \subset P^{}_{\varepsilon'}$, the existence of such pair follows immediately from (A2).

  This reasoning tells us exactly how to prove this result. We start by fixing some $0 < \varepsilon < C$. Next lets pick some open set $V \subset H$ so that
  $ V \cap \phi(L) \, \subset \, \phi(P_\varepsilon)$.

  Let $\varepsilon < \varepsilon' < C$. Since $P_{\varepsilon'}$ is locally finite in
  $G$, there exists a neighbourhood $U$ of $0$ so that
  $P_{\varepsilon'} \cap U =\{ 0\}$.

  We show that
\[
    \cL \cap ( U \times V) = \{ (0,0) \} \ts .
\]
Let $(t ,\phi(t)) \in \cL \cap ( U \times V)$. Then $t \in
L$ and $\phi(t) \in V$, thus $\phi(t) \in V \cap \phi(L) \, \subset \,
\phi(P_\varepsilon)$.

Therefore, there exists $s \in P_\varepsilon$ so
that $\phi(s) = \phi(t)$, which implies that $t-s \in \ker(\phi) \subset P_{\varepsilon' -\varepsilon} $.

This shows that $t=s+ (t-s) \in P_\varepsilon + P_{\varepsilon' -\varepsilon} \subset P_{\varepsilon'}$. As $t$ is also in $U$, it follows that
\[
t \in P_{\varepsilon'} \cap U \, = \, \{ 0 \} \ts ,
\]
which completes the proof of this step.

\emph{Step 2: $\cL$ is relatively dense.}

Let us pick some $0 < \varepsilon < C$. We know that the set $W :=
\overline{\phi(P_\varepsilon)}$ is compact in $H$ and has non-empty
interior. As $P_\varepsilon$ is relatively dense in $G$, there
exists a compact set $K \subset G$ so that $G= P_\varepsilon + K$.

We prove that
\[
    \cL + K \times (W-W) =G \times H \ts .
\]

Let us pick some arbitrary $(x,y) \in G \times H$. Since $\phi(L)$ is dense in $H$, there
exists some $z \in L$ so that $y - \phi(z) \in W$.

As $G= P_\varepsilon + K$, we can write the difference $x-z$ in the form $x-z = t +k$
with $t \in P_\varepsilon, k \in K$. Thus
\[
\begin{split}
  (x,y)\, &= \,(z, \phi(z))+ (x-z, y-\phi(z)) \\
  & = \, (z, \phi(z))+ (t+k, y-\phi(z))\\
  & =  \, (z, \phi(z))+ (t , \phi(t)) +
       \bigl(k , [y -\phi(z)] - \phi(t)\bigr)\\
  & \in \, \cL + \cL + K \times (W-W) \, =\,
  \cL + K \times (W-W)
\end{split}
\]
Thus $(x,y) \in  \cL + K \times (W-W)$. As the element $(x,y) \in G \times H$ was chosen arbitrary, it follows that
\[
    G \times H \,\subset\, \cL + K \times (W-W) \ts ,
\]
which completes the proof.
\end{proof}

Combining all the results in this section together we obtain the
following result.

\begin{theorem}\label{capscheme}
  Let\/ $\bigl(\{P_\varepsilon\}^{}_{0 < \varepsilon <C}, L\bigr)$ be
  a pair which satisfies\/ {\rm (A1)--(A5)}. Let\/ $H$ be the
  completion of\/ $L$ in the uniformity for which\/
  $\{ P_\varepsilon \}^{}_{0 < \varepsilon <C}$ is a neighbourhood basis at\/$0$, and let $\phi\! :\, L \longrightarrow H$ be
  the completion map. Define\/ $\cL := \bigl\{ (t, \phi(t))\mid t
  \in L \bigr\}$.  Then,
\begin{equation}\label{S1cpScheme}
   \begin{array}{ccccc}
    G & \xleftarrow{\;\pi^{}_{1}\;} & G\times H &
    \xrightarrow{\;\pi^{}_{2}\;} & H  \\
    && \cup \\
    && \cL
   \end{array}
\end{equation}
is a cut and project scheme.
\end{theorem}

We will refer to the construction of the cut and project scheme from Theorem \ref{capscheme} as the \emph{Baake--Moody construction}.

\begin{remark}
  In Theorem~\ref{capscheme}, $L = \pi^{}_1(\cL)$ and $\phi$
  is the $\star$-map. \exend
\end{remark}

\begin{corollary}\label{capschemecor}
  Let\/ $\{ P_\varepsilon \}^{}_{0 < \varepsilon <C}$ be such that they satisfy\/ {\rm
    (A1)--(A4)}. Then, for all $0 < \varepsilon < C$, the set
   \[
P_\varepsilon+  \bigcap_{0< \varepsilon' <C} P_{\varepsilon'} \ts .
   \]
is a model set.

In particular, if\/ $\bigcap_{0< \varepsilon' <C} P_{\varepsilon'} = \{ 0 \}$, then  for all $0 < \varepsilon < C$, the set\/ $P_\varepsilon$ is a model set.
\end{corollary}

\begin{proof}
  Pick
\[
    L \, :=\,  \bigcup_{0< \varepsilon <C} P_\varepsilon  \,.
\]
Then, in the Baake--Moody construction, $(G \times H,
\cL)$ is a cut and project scheme. Also, there exists some precompact set $V_\varepsilon \subset H$ with non-empty interior so that
$\phi(L) \cap V_\varepsilon = \phi(P_\varepsilon)$.

We claim that
\[
   P_\varepsilon + \bigcap_{0< \varepsilon' <C} P_{\varepsilon'} \,=\, \oplam (V_\varepsilon) \ts .
\]

The inclusion $P_\varepsilon \subset \oplam (V_\varepsilon)$ follows immediately from $\phi(P_\varepsilon) \subset V_\varepsilon$. Moreover, as
\[
\bigcap_{0< \varepsilon' <C} P_{\varepsilon'} \, = \, \ker(\phi) \ts,
\]
we get $P_\varepsilon +\bigcap_{0< \varepsilon' <C} P_{\varepsilon'} \subset \oplam (V_\varepsilon)$

To complete the proof we need to show that $\oplam (V_\varepsilon) \subset P_\varepsilon + \bigcap_{0< \varepsilon' <C} P_{\varepsilon'}$.

Let $x \in \oplam (V_\varepsilon)$. Then $x \in L$ and $\phi(x) \in V_\varepsilon$.
This shows that $\phi(x) \in \phi(L) \cap V_\varepsilon = \phi(P_\varepsilon)$, thus $\phi(x)=\phi(y)$ for some $y \in P_\varepsilon$.

Then
\[
 x -y \in \ker(\phi) \, = \, \bigcap_{0< \varepsilon' <C} P_{\varepsilon'}  \ts .
\]
Thus
\[
 x  \, = \, y +(x-y) \in P_\varepsilon + \bigcap_{0< \varepsilon' <C} P_{\varepsilon'}  \ts .
\]

\end{proof}

Recall that for a topological group $G$, we denote by
$C_{\mathsf{u}}(G)$ the space of uniformly continuous bounded
functions from $G$ to $\CC$.

The last question we will address in this section is the following: given a function $f : L \to \CC$, when can we lift it through $\phi$ to some function $g \in C_{\mathsf{u}}(H)$?

 As $(H, \phi)$ is the completion of $L$, a function $f$ on $L$ can be extended to some $g \in C_{\mathsf{u}}(H)$ if and only if $f$ is bounded and uniformly continuous in the uniformity on $L$ \cite{S-RUD}.

Therefore we get the following simple Lemma.

\begin{lemma}\label{capschemelemma}
  Let\/ $f\! :\, L \longrightarrow \CC$ be a bounded function with the property that for each\/
  $\varepsilon > 0$, there exists an\/ $\delta >0$ such that, for
  all\/ $t \in P_{\delta}$ and all\/ $x \in L$, we have\/
  $\left| f(t+x) -f(x) \right| < \varepsilon$. Then  there
  exists a function\/ $g \in C_{\mathsf{u}}(H)$ so that\/
  $f(x)=g(x^\star)$ for all\/ $x \in L$.
\end{lemma}

We complete the section by looking at some explicit examples of the Baake--Moody construction.

We start with a very simple example, which yields $\ZZ$ as a regular model set. The cut and project scheme we get is very natural, and this example can easily be generalized to any lattice in an arbitrary LCAG $G$.

\begin{example}
Let $G= \RR, L= \ZZ$ and let
\[
P_\varepsilon \, = \, \ZZ \ts .
\]
In this case we get $H=0$ and $\cL=\{ (n,0) \mid n \in \ZZ \}$.

This is the standard cut and project scheme which produces $\ZZ$ as a model set.
\exend
\end{example}

With a very small modification, we can get a new cut and project scheme which yields $\ZZ$ with internal group $H=\dfrac{\ZZ}{n \ZZ}$.

\begin{example}\label{s1.2Ex2} Again, set $G= \RR, L= \ZZ$ and let $n$ be a fixed positive integer. Define
\[
P_\varepsilon \, = \, n \ZZ \ts .
\]
In this case we get $H=\dfrac{\ZZ}{n \ZZ}$ and
\[ \cL\, = \, \{ (m, m \bmod{n}) \mid m \in \ZZ \} \, \subset \, \RR \times \dfrac{\ZZ}{n \ZZ} \ts .
\]
\exend
\end{example}

The next example shows what happens when we take advantage of the freedom of (A5).

\begin{example}\label{s1.2Ex3} Set $G= \RR$ and define again
\[
P_\varepsilon \, = \, \ZZ \ts .
\]
Pick any additive subgroup $L \leq \RR$ which contains $\ZZ$. Then, it is easy to see that $H$ is simply the factor group $\frac{L}{\ZZ}$ equipped with the discrete topology, and

\[ \cL\, = \, \{ (\alpha, \alpha +\ZZ) \mid \alpha \in L \} \, \subset \, \RR \times \dfrac{L}{ \ZZ} \ts .
\]

In the particular case when $L=\frac{1}{n} \ZZ$, we recover, up to multiplication by a constant, the Example~\ref{s1.2Ex2}.
\exend
\end{example}

Let us look more closely at a very interesting particular case of Example~\ref{s1.2Ex3}. As we mentioned before, (A5) is not critical for our construction, and can always be fulfilled by setting $G=L$. But, as we will see in this example, this choice makes typically the cut and project scheme very big, and it contains much unneeded information.

\begin{example}\label{s1.2Ex4} Set $G= \RR, L=\RR$ and define again
\[
P_\varepsilon \, = \, \ZZ \ts .
\]
Then in the Baake--Moody construction, the group $H$ is just the torus
\[
H \, = \, \frac{\RR}{\ZZ} \,=:\, \TT \ts ,
\]
equipped with the discrete topology, and the lattice $\cL$ is
\[ \cL\, = \, \{ (\alpha, \alpha\bmod{1}) \mid \alpha \in \RR \} \, \subset \, \RR \times \TT \ts .
\]
We should point that in this example, even if our physical group $G$ is simply $\RR$, the lattice of the cut and project scheme is not countable. Moreover, we have
\[ \pi_1(\cL)\, = \, \RR \ts ,
\]
which looks surprising. The fact that $\cL$ is very big does not create any issues, as it is compensated by the spareness of compact sets in $\TT$. Indeed, compact sets in $\TT$ are simply the finite sets, any model set in this scheme is simply a set of the form $\ZZ+F$, where $F$ is a finite set.

It is easy to check that $\cL$ is a lattice
in $\RR \times \TT$. As the example is a little counterintuitive, we do this now.

\emph{Discreteness}: Let $ \lim_n (t_n , t_n+\ZZ)= (0 , 0+Z)$. Then,
$ \lim_n t_n+\ZZ= 0+\ZZ$ in the discrete topology on $\TT$, and hence there
exists an $N$ so that $t_n +\ZZ = 0 +\ZZ$ for all $n>N$. Thus, $t_n \in \ZZ \, \forall n >N$.

Now, since $t_n \in \ZZ$ and $ \lim_n t_n =0$, there exists some $M$ such that $t_n =0$
for all $n >M$.

\emph{Relative denseness}: Let $(x, y+\ZZ) \in \RR \times \TT$. Let
\[
     z \, := \, \lfloor x \rfloor + \{ y \} \ts ,
\]
where $\lfloor \, \rfloor$ and $\{ \, \}$ denote the integer respectively fractional part of a real number.

Then, $z+\ZZ =y+\ZZ $ and $\lfloor x \rfloor=\lfloor z \rfloor$, hence $x -z \in [-1, 1]$. Thus,
\[
   (x, y+\ZZ) \, =\,  (x, z+\ZZ) \, =\,
   ( z, z+ \ZZ) + (x-z , 0+\ZZ) \, \in \,
   \cL + [-1,1] \times \{ 0+ \ZZ \} \ts ,
\]
and hence
\[
  \RR \times \TT \,\subset\,  \cL + [-1,1] \times \{ 0+ \ZZ \}
 \ts .
\]
\exend
\end{example}

In the next example we see that we can use the Baake--Moody construction to create a cut and project scheme for $\ZZ$ with the $p$-adic group as the internal space.

\begin{example}
Let $ p \in \ZZ$ be a prime number. Set $G= \RR, L= \ZZ$. We next define $P_\varepsilon$ in such a way that $P_{\frac{1}{n}}=p^n \ZZ$, and the sets $P_\varepsilon=P_{\frac{1}{n}}$ for all $\frac{1}{n+1}< \varepsilon \leq \frac{1}{n}$.
\[
P_\varepsilon \, := \, p^{\lfloor \frac{1}{\varepsilon} \rfloor} \ZZ \ts .
\]
The conditions (A1), (A2), (A4) and (A5) are straightforward to check.

We see next that (A4) also holds. Indeed, it is easy to check that for all $0 < \varepsilon < \varepsilon'$ we have
\[
P_\varepsilon \, \subset P_{\varepsilon'} \, \ts .
\]
Moreover, as $p^n \ZZ$ is a subgroup of $\ZZ$ we have $P_\varepsilon+P_\varepsilon=P_\varepsilon$. Thus, we actually get a stronger version of (A4):
\[
P_\varepsilon \,+\, P_{\varepsilon'} \, \subset \, P_{ \max\{\varepsilon , \varepsilon' \} } \, \ts .
\]
It is easy to see that the group $H$ produced by the construction is exactly the group $\ZZ_p$ of p-adic integers, and the lattice $\cL$ is just the diagonal in $\ZZ \times \ZZ_p \subset \RR \times \ZZ_p$.
\exend
\end{example}

Finally, we will rediscover by the Baake--Moody construction a very well known cut and project scheme.

\begin{example} Let $\alpha$ be an irrational number. Set $G=\RR^2, L=\{ a+b \alpha \mid a, b \in \ZZ \}$. Then, the sets
\[
P_\varepsilon \, := \, \{ a+ b \alpha \mid  a- b \alpha \in (-\varepsilon , \varepsilon) \} \ts ,
\]
satisfy the conditions (A1)-(A5). It is easy to see that in this case the completion of $L$ in this topology is up to isomorphism $(\RR, \phi)$, where $\phi: L \to \RR$ is defined by
\[
\phi( a+ b \alpha) \,=\,  a- b \alpha  \ts .
\]

The cut and project scheme we get is the one we should expect $(\RR \times \RR, \cL)$ where
\[
\cL \, = \, \{ \left( a+ b \alpha, a- b \alpha \right) \mid a,b \in \ZZ \} \ts .
\]

If $\alpha = \sqrt{2}$, the cut and project scheme we get is exactly the one from \cite[Cor.~7.1]{S-TAO}.

\exend
\end{example}

\section{Almost periodic measures}

The connection between strongly almost periodicity and the discreteness of the Fourier Transform, a connection which makes the class of strongly almost periodic measures important to the theory of mathematical diffraction, was established in \cite{S-ARMA}. For a review of this connection we refer the reader to \cite{S-MoSt} in this volume.

There are two other notions of almost periodicity which appear in a natural way in the study of Meyer sets: norm-almost periodicity \cite{S-BM} and sup-almost periodicity \cite{S-NS2}.
In general it is much easier to deal with these two concepts than strongly almost periodicity, but it is the former which is important for mathematical diffraction. In this section we introduce these three concepts and see their basic properties as well as the connections among them.

On the space of translation bounded measures $\mathcal{M}^\infty(G)$,
we introduce the norm of \cite{S-BL}. Let $K \subset G$ be a fixed
compact set with nonempty interior. For a measure $\omega \in
\mathcal{M}^\infty(G)$ we define
\[
   {\| \omega \|}^{}_{K} \, := \, \sup_{t \in G}
   \left| \omega \right|( t+K) \ts ,
\]
where $\left| \omega \right|$ denotes the total variation measure \cite{S-HeRo};
compare \cite[Sec.~8.5.1]{S-TAO} and references therein.

\begin{definition} A measure $\omega$ is called
\emph{norm-almost periodic} if, for all $\varepsilon >0$, the set
\[
    P^K_\varepsilon(\omega) \, :=\,
    \{ t \in G | \| \omega - T_t \omega \|_K < \varepsilon \}
\]
is relatively dense. It is called \emph{strongly almost periodic} if, for
all $f \in C_{\mathsf{c}}(G)$, the function $f * \omega$ is almost periodic.
\end{definition}

\begin{remark}
If $K' \subset G$ is another compact set with non-empty interior, then
there exist two constants $0 <c \leq C < \infty$, so that
\[
    c \,{\| \cdot \|}^{}_{K'} \,\leq\,
     {\| \cdot \|}^{}_{K}  \,\leq\,
    C \,{\| \cdot \|}^{}_{K'} \ts .
\]
Thus, while different choices of $K$ define different norms, the
topology and the concept of almost periodicity are independent of this
choice.
\exend
\end{remark}

As we will see in Proposition~\ref{BMSNAP} and in Example~\ref{SSAM}
below, norm-almost periodicity is, in general, a stronger requirement
than strongly almost periodicity.

\begin{proposition}\cite[Lemma~7 and Prop.~8]{S-BM}\label{BMSNAP}
  Let\/ $\omega$ be a translation
  bounded measure on\/ $G$.
\begin{itemize}
\item[(i)] If\/ $\omega$ is norm-almost periodic, then\/ $\omega$ is
  strongly almost periodic.
\item[(ii)] If\/ $\supp(\omega)- \supp(\omega)$ is uniformly discrete, then\/ $\omega$ is
  norm-almost periodic if and only if\/ $\omega$ is strongly almost
  periodic.
\end{itemize}
\end{proposition}

We will see at the end of this section that the condition $\supp(\omega)- \supp(\omega)$ is uniformly discrete cannot be dropped from Proposition~\ref{BMSNAP} (ii).

Next we define a new norm on the space $\mathcal{M}^\infty_d(G)$ of discrete translation bounded measures by
\[
    \| \omega \|_\infty \, := \,
    \sup_{x \in G} \bigl\{ \left| \omega (\{ x \}) \right| \bigr\} \ts .
\]
It is easy to see that $\| \cdot \|_\infty$ can be defined on the
space of translation bounded measures. On this space, $\| \cdot \|_\infty$
is a seminorm, and $\| \mu \|_\infty =0$ if and only if $\mu$ is a
continuous measure.\bigskip

\parbox{0.9\textwidth}{ \textbf{Note.}  \emph{For the remainder of this
    article, unless stated otherwise, all measures are assumed to be
    discrete.}}\bigskip

\begin{definition}
  A discrete measure $\omega$ is called \emph{sup-almost periodic} if,
  for all $\varepsilon >0$, the set
\[
\begin{split}
     P^\infty_\varepsilon(\omega) \, &:=\,
    \Bigl\{ t \in G \;\big|\;
     \| \omega - T_t \omega \|_\infty < \varepsilon \Bigr\} \\
    & \hphantom{:}= \, \Bigl\{ t \in G \;\big|\; \sup_{x\in G}
    \bigl\{ \left| \omega( \{ t+x \}) -\omega(\{ x \}) \right| \bigr\}
     < \varepsilon \Bigr\}
\end{split}
\]
is a relatively dense subset of $G$.
\end{definition}

The class of sup-almost periodic measures contains the fully periodic discrete measures, as well as the pure point diffraction measures of Meyer sets in $\RR^d$ \cite{S-NS2}. We will see later in this article that there is a very strong unexpected connection between sup almost periodic measures and cut and project schemes.

Next we study the connection between norm and sup almost periodicity. We show in the next two results that norm almost periodicity implies sup almost periodicity, and that, for measures with weakly uniformly discrete support, the two concepts are equivalent.

\begin{lemma}\label{LMSNAP}
  Let\/ $\omega := \sum_{x \in \vG} \omega(x)\ts \delta_x$ be a translation
  bounded measure on\/ $G$.  Then
\[
     {\| \omega \|}^{}_\infty \,\leq\, {\| \omega \|}^{}_{K}   \ts .
\]
\end{lemma}

\begin{proof}
  Pick some $k \in K^\circ$. Since $x =x-k+k \in x-k+K$ we get
\[
    \left| \omega (\{ x \})\right|  \,\leq\,
     \left| \omega \right| (x-k+K) \ts .
\]
Thus, the assertion
\[
    {\| \omega \|}^{}_{\infty} \,\leq\, {\| \omega \|}^{}_{K}
\]
holds.
\end{proof}

Let us point out that the condition $\mu$ is translation bounded is not needed. Indeed, if the measure $\omega$ is not translation bounded, we have ${ \|
  \omega \|}^{}_{K} =\infty$, and hence the inequality of Lemma~\ref{LMSNAP} is trivially true.

\begin{lemma}\label{MSNAP}
  Let\/ $\omega := \sum_{x \in \vG} \omega(x)\ts \delta_x$ be a translation
  bounded measure on\/ $G$.
\begin{itemize}
\item[(i)] If\/ $\omega$ is norm-almost periodic, then\/ $\omega$ is
  sup-almost periodic.
\item[(ii)] If\/ $\vG$ is weakly uniformly
  discrete, then\/ $\omega$ is norm-almost periodic if and only if\/
  $\omega$ is sup-almost periodic.
\end{itemize}
\end{lemma}

\begin{proof}
Claim (i) follows immediately from Lemma \ref{LMSNAP}.

For claim (ii), we know that there exists an $N>0$ so that, for
all $x \in G$, we have $\# \left( \vG \cap (x+K) \right) \leq N$, where
$\#(\vL)$ denotes the number of elements in a finite point set
$\vL$. Let $t \in G$. Then
\[
   \begin{split}
    \left| T_t\omega - \omega \right| (x+K)\, &=
     \sum_{y \in [ \vG \cup (-t+\vG)] \cap (x+K)  }
     \left| T_t\omega - \omega \right| (\{y\})\\
     &\leq \sum_{y \in [ \vG \cup (-t+\vG)] \cap (x+K)  }
      {\| T_t\omega - \omega \|}_\infty  \\[1mm]
      &=\left[ \# \bigl( [ \vG \cup (-t+\vG)] \cap (x+K) \bigr)\right]
      {\| T_t\omega - \omega \|}^{}_{\infty} \\[1mm]
      &\leq\, 2N\ts {\| T_t\omega - \omega \|}_\infty\ts .
\end{split}
\]
Hence
\[
    {\| T_t\omega - \omega \|}^{}_{K}  \,\leq\,
     2N\ts {\| T_t\omega - \omega \|}^{}_{\infty} \ts ,
\]
and, in particular,
\[
  P^{\infty}_{\!\!\frac{\varepsilon}{2N}}(\omega) \,\subset\,
  P^{K}_{\!\nts\varepsilon}(\omega) \ts ,
\]
which completes the proof.
\end{proof}

We finish the section by studying the norm and sup almost periodicity of a Dirac comb
$\delta^{}_{\!\vL}$. We show that each of norm-almost periodicity respectively sup-almost periodicity
is equivalent to $\vL$ being a fully periodic crystal.

\begin{lemma}\label{L11112}
  Let\/ $\vG$ be a weakly uniformly discrete set, $0 < \varepsilon <1$ and\/ $t \in G$. Then, the following
  statements are equivalent.
\begin{itemize}\itemsep=1pt
\item[(i)] ${\| T_t(\delta_{\vG})-\delta_{\vG} \|}^{}_{K} \leq \varepsilon$.
\item[(ii)] ${\| T_t(\delta_{\vG})-\delta_{\vG} \|}^{}_{\infty} \leq \varepsilon$.
\item[(iii)] $t+ \vG =\vG$.
\end{itemize}
\end{lemma}
\begin{proof}
  The implications $\text{(iii)} \Longrightarrow \text{(i)}
  \Longrightarrow \text{(ii)}$ are clear. We now prove that
  $\text{(ii)} \Longrightarrow \text{(iii)}$.

  Let $x \in \vG$. Then, we have $\left| \delta_{\vG}(\{-t+x\}) - \delta_{\vG}(\{x \})
  \right| < \varepsilon$. Thus, we find that $\left| \delta_{\vG}(\{-t+x \}) -1
  \right| < 1 \Longrightarrow \delta_{\vG}(\{-t+x \}) \neq 0 \Longrightarrow x \in
  t+ \vG$.  This proves that $\vG \subset t+ \vG$.

  The other inclusion is analogous.  Let $x \in t+\vG$. Then $x=t+y$
  for some $y \in \vG$ and, since
\[
    \left| \delta_{\vG}(\{y\}) - \delta_{\vG}(\{t+y \}) \right| < \varepsilon \ts ,
\]
we have $\left| 1 -\delta_{\vG}(\{x \}) \right| < 1 \Longrightarrow
\delta_{\vG}(\{x \}) \neq 0 \Longrightarrow x \in \vG$. Hence $t+\vG \subset
\vG$.
\end{proof}

An immediate consequence of Lemma~\ref{L11112} is the following Corollary.

\begin{corollary}\label{CL1111}
  Let\/ $\vG$ be a weakly uniformly discrete set. Then, the following statements
  are equivalent.
\begin{itemize}\itemsep=1pt
\item[(i)] The measure\/ $\delta_{\vG}$ is norm-almost periodic.
\item[(ii)] The measure\/ $\delta_{\vG}$ is sup-almost periodic.
\item[(iii)] $P^{\infty}_{\!\varepsilon}(\delta_{\vG})$ is relatively dense for
  some\/ $0 < \varepsilon <1$.
\item[(iv)] $P^{K}_{\!\varepsilon}(\delta_{\vG})$ is relatively dense for
  some\/ $0 < \varepsilon <1$.
\item[(v)] The set\/ $\per(\vG):= \{ t \in G \mid t+ \vG = \vG
  \}$ is a lattice.
\item[(vi)] There exists a lattice\/ $L$ and a finite set\/ $F$ so that\/
  $\vG=L+F$.
\end{itemize}
\end{corollary}

In Lemma~\ref{L11112}, weakly uniformly discreteness is only needed to make
sure that norm-almost periodicity makes sense. If $\vG$ is not
weakly uniformly discrete, the equivalence $\text{(ii)} \Longleftrightarrow
\text{(iii})$ still holds, and we can prove the following result.

\begin{corollary}\label{CL1112}
  Let\/ $\vG$ be a locally finite pointset. Then, the following
  statements are equivalent.
\begin{itemize}\itemsep=1pt
\item[(i)] The measure\/ $\delta_{\vG}$ is sup-almost periodic.
\item[(ii)] $P^{\infty}_{\!\varepsilon}(\delta_{\vG})$ is relatively dense
  for some\/ $0 < \varepsilon <1$.
\item[(iii)] The set\/ $\per(\vG):= \{ t \in G \mid t+ \vG = \vG \}$ is
  a lattice.
\item[(iv)] There exists a lattice\/ $L$ and a set\/ $J$ so that\/ $\vG=L+J$.
\end{itemize}
\end{corollary}

Note that, for $\vG$ as in Corollary~\ref{CL1112}, we have the
following equivalence: $\delta_{\vG}$ is regular if and only if $J$ is
finite, if and only if $\vG$ is weakly uniformly discrete, if and only if $\delta_{\vG}$ is translation bounded.

We complete the section by providing a simple example of a measure which is strongly almost periodic, but not norm- nor sup-almost periodic:

\begin{example}\label{SSAM}
Consider
\[
    \omega \, :=\,
     \delta^{}_{\ZZ^2}+  \delta^{}_{(\ZZ \times \pi\ZZ)+ (\frac{1}{2},0)} \ts .
\]
Both measures $\delta^{}_{\ZZ^2}$ and $\delta^{}_{(\ZZ \times \pi\ZZ)+
  (\frac{1}{2},0)}$ individually are fully periodic, thus strongly,
norm- and sup-almost periodic.

Since the sum of two strongly almost periodic measures is strongly
almost periodic \cite{S-ARMA}, $\omega$ is strongly almost
periodic. By Corollary~\ref{CL1111}, however, $\omega$ is neither sup-
nor norm-almost periodic.  \exend
\end{example}

As $\supp(\omega)$ is uniformly discrete, the condition $\supp(\omega)- \supp(\omega)$ is uniformly discrete in Proposition~\ref{BMSNAP} cannot be weaken to $\supp(\omega)$ is uniformly discrete. Later, after the proof of Theorem~\ref{T2}, it will natural to ask whether the above condition can be weakened to finite local complexity of $\supp(\omega)$. The author does not know the answer to this question.

\section{Dense weighted model combs}\label{DWMC}

In this section we introduce the class of measures defined by cut and project schemes and continuous functions on the internal space, and we study the connection between these measures and sup almost periodicity.

\begin{definition}
 Let $( G \times H , \cL )$ be a cut and project scheme. We say that $g :H \to \CC$ is a \emph{window function} if the expression
\[
    \omega_g \,:= \sum_{\ell\in\cL}
    g(\pi^{}_{2}(\ell))\, \delta_{\pi^{}_{1}(\ell)}
    \,= \sum_{ x \in \pi^{}_{1}(\cL)} \!g(x ^\star)\, \delta_x \ts .
\]
is a regular measure on \/$G$. In this case we will call the measure $\omega_g$ a \emph{weighted model
  comb}, and we will refer to\/ $g$ as the \emph{associated function}.

If moreover $g \in C_{\mathsf{u}}(G)$, we will refer to $\omega_g$ as a \emph{continuous weighted model comb}.
\end{definition}

We will not address the question of which functions $g$ are window functions, this will always be part of our assumption.

We now show that any sup-almost periodic discrete regular measure
comes from a cut and project scheme. The key for the next theorem is
the fact that, given such a measure $\omega$, the family $\{
P^\infty_\varepsilon(\omega) \}$ satisfies the requirements of
Section~\ref{construction cp}.

\begin{theorem}\label{T1}
  Let\/ $\omega := \sum_{x \in \vG} \omega(x)\ts \delta_x$ be a regular
  measure on\/ $G$. Then, the following statements are equivalent.
\begin{itemize}\itemsep=1pt
\item[(i)] The measure\/ $\omega$ is sup-almost periodic.
\item[(ii)] There exists a cut and project scheme\/ $(G \times H,
  \cL)$ and a uniformly continuous window function\/ $g \in C_0(H)$ that
  vanishes at\/ $\infty$, so that\/ $\omega = \omega_g$.
\end{itemize}
\end{theorem}

\begin{proof}
  To show that $\text{(ii)} \Longrightarrow \text{(i)}$, note that the
  same argument as in \cite{S-BLRS} holds.  The main idea is that a
  small translation in $H$ does not change $g$ too much. Therefore, as every open set $U \subset H$ yields a relatively dense set $\oplam(U)$, it follows that
  the sets $P_\varepsilon$ are relatively dense. We now make this a formal proof.

  Let $\varepsilon > 0$ be fixed but arbitrary. Since $g$ is uniformly
  continuous, there exists a precompact open set $U \subset H$ so
  that, for all $t \in U$, we have
\[
    \| g - T_t g \|_\infty < \varepsilon \ts .
\]
Then, for all $x \in \pi^{}_{1}(\cL)$ and $t \in \oplam(U)$, we have
\[
   \left| \omega_g( x+t ) - \omega_g (x) \right| \,=\, \left| g(
  x^\star +t^\star ) -g( x^\star ) \right| \,<\, \varepsilon \ts .
\]
Therefore, $\oplam (U) \subset P^{\infty}_{\!\varepsilon}(\omega_g) $. Since
$\oplam(U)$ is relatively dense \cite{S-MOO}, so is
$P^{\infty}_{\!\varepsilon}(\omega_g)$.

To prove the implication $\text{(i)} \Longrightarrow \text{(ii)}$, we
first see that $\| \omega \|_\infty < \infty$. Let $K'$ be so that
$P^{\infty}_{\! 1}(\omega) +K' = G$. Then, it is easy to see that, for
all $x \in G$, we have
\[
   \left| \omega \right| (x) \,\leq\,
    1 + \left| \omega \right| (K') \, <\,  \infty \ts .
\]

We now claim that the pair $\bigl(\{ P_{\varepsilon}(\omega) \}^{}_{ 0 <
  \varepsilon < \| \omega \|^{}_\infty } , L:= \langle \vG\rangle \bigr)$
satisfies conditions (A1)--(A5), thus the Baake--Moody construction yields our cut and project scheme.

Condition (A1) is clear.

Condition (A2) is proved by contradiction. Let $0 < \varepsilon < \|
\omega \|_\infty$ and assume that $P_\varepsilon$ is not locally
finite.  Then, there exists a compact set $K'$ so that $\# (P_\varepsilon
\cap K') = \infty$. Pick a $y \in \vG$ so that $ \left| \omega(y)
\right| > \varepsilon$. Then,
\[
    \left| \omega( y+t ) \right| \, >\, \left| \omega(y) \right|
    - \varepsilon \, >\, 0
\]
holds for all $t \in P_\varepsilon \cap K'$. Hence,
\[
   \left| \omega \right| (y+K') \,\geq
   \sum_{t \in P_\varepsilon \cap K} \left| \omega \right| (y+t)
   \,\geq\, \bigl(\left| \omega(y) \right| - \varepsilon\bigr)\,
    \# (P_\varepsilon \cap K') \,=\, \infty
  \ts .
\]
But this contradicts the regularity of $\omega$.

Condition (A3) follows from the triangle inequality.

Condition (A4) is equivalent to $\omega$ being sup-almost periodic.

Finally, consider condition (A5). Let $0 < \varepsilon < \| \omega
\|_\infty$. Let $x \in \vG$ be so that $\varepsilon < \left| \omega(x)
\right|$. Then, for each $t \in P_\varepsilon$ we have
\[
   \left| \omega(x+t) \right| \,\geq\, \left| \omega(x) \right| -
   \varepsilon \,>\, 0\ts .
\]
Thus, $t+x \in \vG$, and hence $t \in -x +\vG$. We conclude
that
\[
   P_\varepsilon \,\subset\, -x + \vG \,\subset\, \vG - \vG \,\subset\, L \ts .
\]
Hence the pair $\bigl(\{ P_\varepsilon(\omega) \}_{ 0 < \varepsilon <
  \| \omega \|_\infty } , L:= \langle \vG\rangle> \bigr)$ satisfies
conditions (A1)--(A5). Thus, by Theorem~\ref{capscheme}, we get a cut
and project scheme $(G \times H, \cL)$.

Let $f\! :\, L \rightarrow \CC$ be defined by $f(x) = \omega(\{x\})$. Then
$f(x)=\omega(x)$ for all $x \in \vG$, and $f(x)=0$ for $x$ outside
$\vG$. Let $t \in P_\varepsilon$ and $x \in L$. Since $t + x \in L$ we get
\[
    \left| f(t+x) - f(x) \right| \,=\,
    \left| \omega(\{ t+x \}) - \omega(\{ x \})
   \right| \, <\,  \varepsilon \ts .
\]
This shows that $f$ is uniformly continuous in the topology induced on $L$.
Hence, there exists a function $g \in C_{\mathsf{u}}(H)$ so that
\[
    f \, =\, g \circ \phi \ts .
\]
Then
\[
   \omega  \, = \sum_{x \in L}
  \omega(\{ x \})\, \delta_x
   \, = \sum _{x \in L} f(x)\, \delta_x
   \, =\sum _{x \in L} g(\phi(x))\, \delta_x
    \,= \sum _{x \in L} g(x^\star)\, \delta_x =\omega_g \ts .
\]

The only thing left to prove is that $g$ vanishes at $\infty$. We will
see that this is a consequence of the regularity of the measure.

Let observe first that, since $\omega$ is regular measure, the set
\[
   S_\varepsilon \,:=\,
    \{ x \in G \mid \left| \omega(x) \right| > \varepsilon \}
\]
is locally finite for all $ 0< \varepsilon$.

Fix some $\varepsilon >0$, and let $K'$ be some compact set such that
\[P_{\varepsilon/2} +K' =G \ts .\]
Define $F:= S_{\varepsilon/2} \cap K'$. Then $F$ is a finite subset
of $L$.

Let $W := \{ y \in H \mid \left| g(y) \right| > \varepsilon
\}$. To complete the proof we show that $W$ is a subset of the compact set $\overline{ \phi(P^{}_{\!\varepsilon/2})+
     \phi(F)\vphantom{\cL}}$.

 Let $y \in W \cap \pi^{}_{2}(\cL)$. Then, there exists
$x \in S_\varepsilon$ so that $y = \phi(x)$.  We can write $x = z + k$
with $z \in P^{}_{\varepsilon/2}$ and $k \in K'$.  Then,
\[
   \left| \omega (k) \right| \, =\,
   \left| \omega (x-z) \right| \, \geq\,
  \left| \omega (x) \right|- \myfrac{\varepsilon}{2}
   \, >\, \myfrac{\varepsilon}{2} \ts .
\]
Thus, $k \in S_{\varepsilon/2} \cap K' = F$.  So $x \in
P_{\varepsilon/2} +F$, and hence $y \in \phi(P_{\varepsilon/2})+ \phi(
F)$.

We showed that
\begin{equation}\label{EQ111}
    W \cap \pi^{}_{2}(\cL) \,\subset\,
   \phi(P^{}_{\!\varepsilon/2})+ \phi( F) \ts .
\end{equation}
Since $\pi^{}_{2}(\cL)$ is dense in $H$ and since $W$ is
open, taking the closure on both sides of Eq.~\eqref{EQ111} gives
\[
    W \,\subset\, \overline{ W \cap \pi^{}_{2}(\cL)}
      \,\subset\, \overline{ \phi(P^{}_{\!\varepsilon/2})+
     \phi(F)\vphantom{\cL}} \ts .
\]
Thus, we conclude that $\left| g(y) \right| \leq \varepsilon $
outside the compact set $\overline{ \phi(P_{\varepsilon/2})+ \phi(
  F)}$, which completes the proof.
\end{proof}

\begin{remark}\label{R1.4}\mbox{}
\begin{itemize}
\item[(i)] It is easy to check that under the conditions of Theorem~\ref{T1} we have
\[
\bigcap_{\varepsilon' >0} P^{\infty}_{\!\varepsilon'}(\omega) \, = \, \{ t \in G | T_t \omega = \omega \}  \, = \, \mbox{Per}(\omega) \ts .
\]
Therefore, for all $\varepsilon >0$ we have
\[
P^{\infty}_{\!\varepsilon}(\omega) + \bigcap_{\varepsilon' >0} P^{\infty}_{\!\varepsilon'}(\omega) \,=\, P^{\infty}_{\!\varepsilon}(\omega) \ts .
\]
\item[(ii)] The fact that $\omega$ is a regular measure is essential for the
  discreteness of the $P_\varepsilon$ and thus the construction of the
  cut and project scheme. Note that the formal sum $\delta^{}_{\QQ}$ is
  sup-almost periodic in $\RR$, but not a measure, and the set of almost periods
  does not produce a cut and project scheme.
\item[(iii)] The condition $g \in C_0(H)$ is necessary but not sufficient condition for $\omega$ to be a regular measure. \exend
\end{itemize}
\end{remark}

Combining Theorem~\ref{T1}, Corollary~\ref{capschemecor} and Remark~\ref{R1.4} (i), we obtain
the following result.

\begin{corollary}\label{123123123}
  Let\/ $\omega$ be a discrete regular measure on\/ $G$. If\/ $\omega$ is
  sup-almost periodic, then, for all\/ $0 < \varepsilon < \|
  \omega\|_\infty$, the set\/ $P^{\infty}_{\!\varepsilon}(\omega)$ is a model
  set.\qed
\end{corollary}

In the same way we can prove the following result.

\begin{lemma}\label{NAPLJ}
  Let\/ $\omega$ be a discrete translation bounded measure on\/
  $G$. If\/ $\omega$ is norm-almost periodic, then there exists a\/ $C>
  0$ so that, for all\/ $0 < \varepsilon < C$, the set\/
  $P^{K}_{\!\varepsilon}(\omega)$ is a model set.
\end{lemma}

\begin{proof}
  We show that $\{ P^{K}_{\!\varepsilon}(\omega) \}^{}_{0 < \varepsilon
    < \| \omega\|^{}_\infty}$ satisfies (A1)--(A4), and thus, by
  Corollary~\ref{capschemecor}, the sets
  $P^{K}_{\!\varepsilon}(\omega)$ are all model sets.

  Conditions (A1) and (A3) are obvious, while (A4) is equivalent to
  norm-almost periodicity.

  Since $\omega$ is norm-almost periodic, it is sup-almost periodic,
  and we showed in Theorem~\ref{T1} that, for all $0 < \varepsilon <
  \| \omega \|_\infty$, the set $P^{\infty}_{\!\varepsilon}(\omega)$
  is locally finite. Since $P^{K}_{\!\varepsilon}(\omega) \subset
  P^{\infty}_{\!\varepsilon}(\omega)$, condition (A2) also holds.
\end{proof}

Note that Lemma \ref{NAPLJ} is not necessarily true for non-discrete
measures. The following Lemma shows that if $f \in C_{\mathsf{u}}(G)$ is a almost periodic
function, then the measure $\mu= f \dd \theta_G$ is a norm-almost
periodic measure. Note that in this case, but the set of norm-almost periods of $\mu$ is not
locally finite.

\begin{lemma}
  Let\/ $f \in C_{\mathsf{u}}(G)$ and let\/ $\omega=f d \theta_G$. Then\/ $\omega$ is norm almost periodic if and only if\/ $f$ is an almost periodic function.
\end{lemma}

\begin{proof}
$\Longrightarrow$ is immediate: since $\omega$ is a norm almost periodic measure, it is a strongly almost periodic measure. Therefore, $f \theta_G$ is a strongly almost periodic measure. As $f \in C_{\mathsf{u}}(G)$, it follows \cite{S-MoSt}, that $f$ is an almost periodic function.

$\Longleftarrow$ Let $K$ be a compact set with non-empty interior. Now, if $t$ is an $\varepsilon$ period of $f$, we have
\[
\|f -T^tf \|_\infty \, < \, \varepsilon \ts .
\]
Therefore
\[
\begin{split}
\| \omega -T^t \omega \|_K \, & = \, \sup_{x \in G}  \left| \omega -T^t \omega \right| (x+K) \\
&= \,   \sup_{x \in G} \int_{x+K}  \left| f(y) -f(-t+y) \right| dy  \\
&\leq \,   \sup_{x \in G} \int_{x+K} \varepsilon dy  =\varepsilon \theta_G(K) \ts .
\end{split}
\]

Therefore, every $\varepsilon$ of $f$ is an $\varepsilon \theta_G(K)$-norm almost period of $\omega$. As $\theta_G(K)$ is a fixed constant, which doesn't depend on $\varepsilon$, the conclusion follows.
\end{proof}

\section[Continuous weighted model combs]{The
characterisation of continuous weighted
model combs with compactly supported associated function}\label{Continuous weighted model combs}

In this section, we characterise those weighted model combs which have a compactly supported continuous associated function $g \in C_{\mathsf{c}}(H)$.

The following lemma will be an important tool for many of the results
we will prove next. In the particular case $G=\RR^d$, this result was proven in \cite[Lemma 3.10]{S-NS1}, and the same proof works in the general case of a LCAG $G$. We include the proof below for completeness.

\begin{lemma}\label{NS1x}
  Let\/ $B \subset A$. If\/ $B$ is relatively dense and\/ $A$ has finite
  local complexity, then there exists a finite set\/ $F$ so that
\[
    A \,\subset\, B+F \ts.
\]
\end{lemma}

\begin{proof}
  Let $K$ be so that $B+K=G$. Let $F:= K \cap A-A$. We show that $A
  \subset B +F$.

  Let $x \in A$. Then $x \in B+K$, thus $x=y+k$ with $y \in B$, $k \in
  K$. Then, $k=x-y \in (A-B) \cap K \subset F$.

  Since $x=y+k$ with $y \in B$ and $k \in F$, we are done.
\end{proof}

We already proved in the previous section that any sup-almost periodic
discrete measure comes from a cut and project scheme and a uniformly
continuous function. We can now use Lemma~\ref{NS1x} to show that, if
the measure is supported on a set with finite local complexity, then
the associated function $g$ on $H$ is compactly supported.
Thus, we obtain the following result.

\begin{theorem}\label{T2}
  Let\/ $\omega = \sum_{x \in \vG} \omega(x)\, \delta_x$ be a
  translation bounded regular measure, with\/ $\omega(x) \neq 0 \, \forall x \in \vG$.  Then, the following
  statements are equivalent.
\begin{itemize}\itemsep=1pt
\item[(i)] $\omega$ is norm-almost periodic and\/ $\vG$ has finite local
  complexity.
\item[(ii)] $\omega$ is norm-almost periodic and\/ $\vG-\vG$ is uniformly
  discrete.
\item[(iii)] $\omega$ is strongly almost periodic and\/ $\vG-\vG$ is
  uniformly discrete.
\item[(iv)]$\omega$ is sup-almost periodic and\/ $\vG-\vG$ is uniformly
  discrete.
\item[(v)]$\omega$ is sup-almost periodic and\/ $\vG$ has finite local
  complexity.
\item[(vi)]There exists a cut and project scheme\/ $(G \times H,
  \cL)$ and a function\/ $f \in C_{\mathsf{c}}(H)$ so that
\[
    \omega \, =\, \omega_f\, = \!
     \sum_{(x,x^\star) \in \cL}\! f(x ^\star)\,\delta_x \ts .
\]
\end{itemize}
\end{theorem}

\begin{proof}
  The equivalence $\text{(ii)} \Longleftrightarrow \text{(iii)}$
  follows from Proposition~\ref{BMSNAP}, while $\text{(ii)}
  \Longleftrightarrow \text{(iv)}$ follows from Lemma~\ref{MSNAP}.

  The implication $\text{(ii)} \Longrightarrow \text{(i)}$ is obvious,
  while $\text{(i)} \Longleftrightarrow \text{(v)}$ follows also from
  Lemma~\ref{MSNAP}. To complete the proof, we prove $\text{(v)}
  \Longrightarrow \text{(vi)} \Longrightarrow \text{(iv)}$.

  The implication $\text{(vi)} \Longrightarrow \text{(iv)}$ is
immediate. The measure $\omega$ is translation bounded, thus regular,
and applying Theorem~\ref{T1} we get that $\omega=\omega_f$ is
sup-almost periodic.

Let $W$ be a compact set in $H$ so that $\supp(f) \subset W$. Then, as $\omega(x) \neq 0 \, \forall x \in \vG$ we have
$\vG \subset \oplam(W)$ and $\vG- \vG \subset \oplam(W-W)$. Since $W-W$ is
compact, it follows that $\vG- \vG$ is uniformly discrete.

We now prove that $\text{(v)} \Longrightarrow \text{(vi)}$.  By
Theorem~\ref{T1}, there exists a cut and project scheme $(G \times H ,
\cL)$ and a function $f \in C_0(H)$ so that $\omega=
\omega_f$. We will show that $f$ is compactly supported.

If $f \equiv 0$, we are done. Otherwise, pick some $0 < \varepsilon <
\|f\|_\infty$. Since $f \in C_0(H)$, the set
\[
   U \,:=\, \bigl\{ y \in H \mid \left| f(y) \right| >
            \varepsilon \bigr\}
\]
is open, non-empty and precompact. Thus the set $\oplam(U)$ is relatively
dense, and it is contained in $\vG$.

Since $\vG$ has finite local complexity, by Lemma~\ref{NS1x}, there
exists a finite set $F$ so that
\[
    \vG \,\subset\, \oplam(U) +F \ts .
\]
Moreover, since both $\oplam(U)$ and $\vG$ are contained in the group
$\pi^{}_{1}(\cL)$, we can choose $F$ so that $F \subset
\pi^{}_{1}(\cL)$. Let $F^\star= \{x^\star \mid (x,x^\star)
\in \cL \text{ and } x \in F \}$.

Define
\[
   V\, := \, \{ y \in H \mid f(y) \neq 0 \} \ts .
\]
We need to show that $V$ is precompact. Since
$\pi^{}_{2}(\cL)$ is dense in $H$ and $V$ is open, it is
sufficient to show that $\overline{\pi^{}_{2}(\cL) \cap V}$ is
compact.

Let $y^\star \in \pi^{}_{2}(\cL) \cap V$. Then, there exists
an $y \in G$ so that $(y,y^\star) \in \cL$.
By the definition of $\omega$, we have $\omega(y)=f(y^\star) \neq
0$. Hence $y \in \vG \subset \oplam(U) +F$, and therefore $y^\star \in U +
F^\star$.

This shows that $\pi^{}_{2}(\cL) \cap V \subset U + F^{\star}$, and
hence
\[
  \supp(f) \,= \,\overline{\pi^{}_{2}(\cL) \cap V}
           \,\subset\, \overline{U}+ F^{\star} \ts ,
\]
which is compact since $U$ is precompact and $F^\star$ is finite.
\end{proof}

We complete the section by looking at one example, which emphasizes why we need the freedom of choice of $L$ in (A5).
\begin{example}
Let
\[
    \omega \, := \, \delta^{}_\ZZ - \myfrac{1}{2}\, \delta^{}_{\pi + \ZZ} \ts .
\]
In this case, $C(\omega) =1$ and
\[
  P_{\!\varepsilon} (\omega) \, =\,  \ZZ
\]
for all $0 < \varepsilon < 1$.

\emph{First Choice}. Let us choose $L = \langle \bigcup_{0 <
  \varepsilon < C} P_{\!\varepsilon} \rangle$, the smallest
possibility for $L$.  Then, our cut and project scheme is simply $(\RR
\times \{ 0 \} , \ZZ \times \{ 0 \} )$.

While we get a cut and project scheme, $\omega$ does not come from a
continuous function $g \in C_{\mathsf{u}}(H)$, it is only a linear
combination of translates of such weighted combs (one can show that
this is always the case for this choice).

\emph{Second Choice}. Let us choose $L = G$, the largest possibility
for $L$.  Then $H = \RR/\ZZ$ equipped with the discrete topology, and
\[
   \cL \,=\, \{ (t, t+ \ZZ) \mid t \in \RR \} \ts .
\]
Let $g\! :\, \RR/\ZZ \longrightarrow \CC$ be defined by
\[
   g(0+ \ZZ) \,=\, 1 \ts ,\quad
   g(\pi + \ZZ) \,=\, - \myfrac{1}{2} \quad\text{and}\quad
   g(x+ \ZZ) = 0 \, \text{ otherwise.}
\]
Then, $g \in C_{\mathsf{c}}(H)$ and $\omega =\omega^{}_{g}$.

While this choice of $G$ always works, let us observe that our cut and
project scheme is too big: $H$ is an uncountable discrete group, and
we only use two of its elements. It should be sufficient to replace
$H$ by the subgroup generated by these two elements.

\emph{Third Choice}. Let us choose $L= \langle \ZZ \cup (\pi
+Z)\rangle$, the group generated by the support of $\omega$. Then $L=
\ZZ + \pi\ZZ$ and $H = \pi\ZZ$, the free Abelian group with one
generator. Again $H$ is equipped with the discrete topology, and
\[
  \cL \, =\,
  \{ (a + b\pi, b \pi) \mid a, b  \in \ZZ \} \ts .
\]
Let $g\! :\, \pi \ZZ \longrightarrow \CC$ be defined by
\[
   g(0) \, =\,  1 \ts , \quad
   g(\pi) \, =\,  - \myfrac{1}{2}\quad\text{and}\quad
   g(x) \,=\, 0 \, \text{ otherwise.}
\]
Then, $g \in C_{\mathsf{c}}(H)$ and $\omega =\omega^{}_{g}$.
\exend
\end{example}

\section{$\varepsilon$-dual characters}\label{CHARMS}

Let $\vL \subset G$ be a relatively dense set. As usual, let
\[
    \vL^\varepsilon \, :=\,
     \bigl\{ \chi \in \widehat{G} \mid
     \left| \chi(x)-1 \right| \leq \varepsilon
     \text{ for all } x \in \vL \bigr\} .
\]
$\vL^\varepsilon$ is the set of $\varepsilon$-dual characters of
$\vL$.  In this section, we will study the $\varepsilon$-dual sets of
a Delone set $\vL \subset G$, where $G$ is an arbitrary
$\sigma$-compact LCAG. We will show that, if all the
$\varepsilon$-dual sets are relatively dense for $0 < \varepsilon <2$, then the axioms
(A1)--(A4) hold, and we thus the Baake--Moody construction yields a cut and project scheme in
which all $\vL^\varepsilon$ are model sets. This way, we will be
able to prove that most of the equivalent definitions of Meyer sets in
$\RR^d$ are still equivalent in the more general context of
$\sigma$-compact LCAGs.

It looks surprising that the theory we built at the beginning of the Chapter, with
the model of almost periods of sup-almost periodic measures in mind,
can also be applied to the $\varepsilon$-dual sets, but there is a
simple reason why this happens.

For a point set $\vL$ we denote $\Delta:= \vL - \vL$. It is easy to
see (compare for example Lemma~\ref{DMS3}) that
\[
    \vL^\varepsilon \, \subset\,   \Delta^{(2\varepsilon)}  .
\]
Moreover, as we will see in Section \ref{diffMS}, there exists a
constant $C>0$ so that $\Delta^\varepsilon$ are $C\varepsilon$-almost
periods of the pure point diffraction measure
$(\widehat{\gamma})_{\mathsf{pp}}$ of $\vL$. Thus, if
$\vL^\varepsilon$ is relatively dense for each $\varepsilon >0$, then
$(\widehat{\gamma})_{\mathsf{pp}}$ is sup-almost periodic, and the
sets $\vL^{\varepsilon}$ are subsets of the model sets given by the
$\varepsilon$ sup-almost periods of the measure.

In the first part of the section, we will show that, for a relative
dense set $\vL \subset G$, the family $\{ \vL^\varepsilon \}^{}_{0 <
  \varepsilon < \frac{\sqrt{3}}{2}}$ satisfies the axioms
(A1)--(A3). Both (A1) and (A3) are immediate consequences of the
definition of $\vL^\varepsilon$, and even hold if $\vL$ is not
relatively dense, while (A2) follows from the relative denseness of
$\vL$.

\begin{lemma}\label{a1a1a3}
  Let\/ $\vL \subset G$. Then, the family\/ $\{ \vL^\varepsilon
  \}^{}_{ 0 < \varepsilon <2 }$ satisfies\/ {\rm (A1)} and\/ {\rm
    (A3)} in\/ $\widehat{G}$.
\end{lemma}

\begin{proof}
Property (A1) is obvious.

To verify (A3), let $\varepsilon, \varepsilon'>0$ and let $\chi \in
\vL^\varepsilon$, $\psi \in \vL^{\varepsilon'}$. Then, for all $x \in
\vL$, we have
\[
\begin{split}
   \left| \chi(x)\psi(x) -1 \right|
    \,& =\,  \left| \chi(x)\psi(x) -\chi(x)+\chi(x)-1 \right| \\
    &\leq\, \left| \chi(x)\psi(x) -\chi(x) \right|+
      \left| \chi(x)-1 \right|
    \,\leq\, \varepsilon+\varepsilon' \ts ,
\end{split}
\]
completing the proof.
\end{proof}

To prove (A2), we first need to prove a simple lemma, which is a
pretty standard result in character theory.

\begin{lemma}\label{a2a2a3}
  Let\/ $H$ be a subgroup of\/ $U(1) \subset \CC^*$. If\/ $H \neq \{ 1
  \}$, then there exists some \/ $z \in H$ with $\left| z -1 \right| \geq
  \sqrt{3}$.
\end{lemma}

\begin{proof}
  Since $H \neq \{ 1 \}$, we can find an $\zeta \in H$ so that
  $\zeta \neq 1$. We have to consider three cases.

  \emph{First case:} $\ord(\zeta) = \infty$, which means that $\zeta$ is not
  a root of unity. Then, by Dirichlet's theorem, $\zeta$ generates a
  subgroup which is dense in $U(1)$.

  \emph{Second case:} $\ord(\zeta) = 2n$ for some $n \geq 1$. Then, $-1 =
  \zeta^{n} \in H$.

  \emph{Third case:} $\ord(\zeta) = 2n+1$ for some $n \geq 1$. Then, $\zeta$
  is a primitive root of unity, and thus $\ee^{\frac{2 \pi \ii}{2n+1}} \in
  H \Longrightarrow \ee^{\frac{2 n \pi \ii}{2n+1}} \in H$. Then,
  $\frac{2n\pi}{2n+1}=1- \frac{\pi}{2n+1} \geq 1-
  \frac{\pi}{3}=\frac{2\pi}{3}$. Using the fact that $-\cos(x)$ is
  increasing on the interval $[ -\frac{\pi}{2} , \pi ]$, we get
\[
\begin{split}
    \Bigl| \ee^{\frac{2 n \pi \ii}{2n+1}} -1 \Bigr|^2 \,&=\,
    \Bigl[\cos \Bigl(\myfrac{2 n \pi }{2n+1}\Bigr) -1\Bigr]^2 +
    \Bigl[\sin \Bigl(\myfrac{2 n \pi }{2n+1}\Bigr)\Bigr]^2\\[1mm]
     &=\, 2-2 \cos \Bigl(\myfrac{2 n \pi }{2n+1}\Bigr)
     \,\geq\, 2-2 \cos \Bigl(\myfrac{2  \pi }{3}\Bigr) \, =\, 3 \ts .
\end{split}
\]
\end{proof}
Let us also note that unless $H$ consists exactly of the third roots of unity, we can find some $z \in H$ such that the inequality is strict.

We now show that, for all $0< \varepsilon < \frac{\sqrt{3}}{2}$, if one
of $\vL$ or $\vL^\varepsilon$ is relatively dense, the other one is
uniformly discrete. The proof is very similar to the one in
\cite{S-MOO}.

\begin{proposition}\label{EDProp}
  Let\/ $\vL \subset G$.
\begin{itemize}\itemsep=1pt
\item[(i)] If\/ $\vL$ is relatively dense, $\vL^\varepsilon$ is
  uniformly discrete for all\/ $0< \varepsilon < \frac{\sqrt{3}}{2}$.
\item[(ii)] If\/ $\vL^\varepsilon$ is relatively dense for some\/ $0<
  \varepsilon < \frac{\sqrt{3}}{2}$, the set\/ $\vL$ is uniformly discrete.
\end{itemize}
\end{proposition}

\begin{proof}
  Claim (i). Let $K$ be such that $\vL+K = G$. Pick $\varepsilon'$ be
  so that $2\varepsilon+\varepsilon'< \sqrt{3}$.

  Let $N(K, \varepsilon'):= \bigl\{ \chi \in \widehat{G}\mid
  |\chi(x)-1| < \varepsilon'\text{ for all } x \in K \bigr\}$. Then,
  $N(K, \varepsilon')$ is an open neighbourhood of $1 \in \widehat{G}$.

  We prove that $(\vL^\varepsilon-\vL^\varepsilon) \cap N(K,
  \varepsilon') = \{ 1 \}$.

  Let $\chi \in (\vL^\varepsilon-\vL^\varepsilon) \cap N(K,
  \varepsilon')$ and let $x \in G$. Then, by Lemma~\ref{a1a1a3}, $\chi
  \in \vL^{2\varepsilon}$, and thus $x=y+k$ for some $y \in \vL$ and
  $k \in K$. Hence
\[
    \left| \chi(x) -1 \right| \,=\,
    \left| \chi(y)\chi(k) -1 \right| \,\leq\,
    \left| \chi(y) -1 \right| + \left| \chi(k) -1 \right|
    \,\leq\, 2\varepsilon+\varepsilon' \,<\, \sqrt{3} \ts .
\]
Thus, $\Imag(\chi)$ is a subgroup of $U(1)$ so that for all $z \in
\Imag(\chi)$ we have
\[
\left| z -1 \right| \, <\,  \sqrt{3} \ts .
\]
Then, by Lemma~\ref{a2a2a3}, we find that $\Imag(\chi)= \{ 1 \}$,
which shows that $\chi =1$, proving claim (i).

Claim (ii) follows immediately from (i) and the observation that $\vL
\subset \vL^{\varepsilon \varepsilon}$.
\end{proof}

\begin{remark}
  In $\RR^d$, the result in
  Proposition~\ref{EDProp} holds for all $0 < \varepsilon <1$ \cite[Prop.~6.5]{S-MOO}.

  The reason why we get the stronger result is simple: If we pick any $0< \varepsilon <1$ with $\varepsilon'
  <2-2\varepsilon$, and choose a $\chi \in
  (\vL^\varepsilon-\vL^\varepsilon) \cap N(K, \varepsilon')$, exactly as in Proposition~\ref{EDProp}, we get
\begin{equation}\label{EQ55}
   \left| \chi(x) -1 \right| \,\leq\, 2\varepsilon+\varepsilon'
\end{equation}
for all $x \in G$.
This shows that $-1 \notin \Imag (\chi)$.

If $G = \RR^d$, then $\Imag(\chi)$ is a
\emph{connected} subgroup of $U(1)$, thus is either $\{ 1 \}$ or
the entire $U(1)$. Therefore, $\Imag(\chi)=\{ 1 \}$.

In general, connectedness cannot be used. However, since
$2\varepsilon+\varepsilon' <1$, using Eq.~\eqref{EQ55} it is easy to
prove that $\Imag(\chi)$ is finite (since $-1$ is not in its
closure), and to show that there is an upper bound (which only
depends on the sum $2\varepsilon+\varepsilon'$) on the number of
elements of $\Imag(\chi)$.

In particular, it follows that, for every $\chi \in
(\vL^\varepsilon-\vL^\varepsilon) \cap N(K, \varepsilon')$, the
subgroup $\ker(\chi)$ is open and closed in $G$ and the factor group
$G / \ker(\chi)$ is finite of bounded order. This observation is
probably sufficient to prove that $\vL^\varepsilon$ is uniformly
discrete also for $ \frac{\sqrt{3}}{2} \leq \varepsilon <1$.  Anyhow,
Proposition~\ref{EDProp} is all we need in this section.
\exend
\end{remark}

\begin{corollary}\label{DMS}
  Let\/ $\vL \subset G$ be relatively dense.
\begin{itemize}\itemsep=1pt
\item[(i)] For all\/ $0< \varepsilon < \frac{\sqrt{3}}{4}$, the set\/
  $\vL^\varepsilon-\vL^\varepsilon$ is uniformly discrete.
\item[(ii)] For all\/ $0< \varepsilon < \frac{\sqrt{3}}{6}$, the set\/
  $\vL^\varepsilon-\vL^\varepsilon-\vL^\varepsilon$ is uniformly
  discrete.
\end{itemize}
\end{corollary}

\begin{proof}
  We know that $\vL^\varepsilon-\vL^\varepsilon \subset \vL^{2
    \varepsilon}$ and $\vL^\varepsilon-\vL^\varepsilon-\vL^\varepsilon
  \subset \vL^{3 \varepsilon}$. The proof then follows immediately from
  Proposition~\ref{EDProp}.
\end{proof}

\begin{lemma}\label{DMS3}
  Let\/ $\vL \subset G$, let\/ $\Delta:= \vL- \vL$ and let\/
  $\varepsilon >0$.  Then, one has\/ $\Delta \subset \vL^{\varepsilon
    (2\varepsilon)}$ and\/ $\vL^\varepsilon \subset \Delta^{2
    \varepsilon}$.
\end{lemma}

\begin{proof}
  The first claim is a consequence of $\Delta \subset \vL^{\varepsilon
    \varepsilon}-\vL^{\varepsilon \varepsilon} \subset
  \vL^{\varepsilon (2\varepsilon)}$, while the second follows from
  $\vL^\varepsilon \subset
  (\vL^{\varepsilon})^{(2\varepsilon)(2\varepsilon)}
  =(\vL^{\varepsilon (2 \varepsilon)})^{2\varepsilon} \subset
  \Delta^{2\varepsilon}$, where the final inclusion uses the first
  claim again.
\end{proof}

\begin{corollary}\label{DMS2}
  Let\/ $\vL \subset G$. If, for some\/ $0< \varepsilon <
  \frac{\sqrt{3}}{4}$, the set\/ $\vL^\varepsilon$ is relatively dense,
  then\/ $\vL-\vL$ is uniformly discrete.
\end{corollary}

\begin{proof}
  Since $\vL^{\varepsilon}$ is relatively dense, and $2\varepsilon <
  \frac{\sqrt{3}}{2}$, it follows by Proposition~\ref{EDProp}(i)
  that $\vL^{\varepsilon (2\varepsilon)}$ is uniformly discrete. By
  Lemma~\ref{DMS3}, $\Delta := \vL - \vL$ is also uniformly discrete.
\end{proof}

\begin{theorem}\label{T66t66}
  Let\/ $\vL \subset G$ be relatively dense. If\/ $\vL^\varepsilon$ is
  relatively dense for all\/ $\varepsilon >0$, then\/ $\vL^\varepsilon$ is
  a model set for all\/ $0 < \varepsilon < \frac{\sqrt{3}}{2}$.
\end{theorem}

\begin{proof}
  By Lemma~\ref{a1a1a3}, the family $\{ \vL^\varepsilon \}^{}_{ 0 <
    \varepsilon <2 }$ satisfies (A1) and (A3) in $\widehat{G}$. Since
  $\vL$ is relatively dense, by Proposition~\ref{EDProp},
  $\vL^\varepsilon$ is also uniformly discrete for all $0 <
  \varepsilon < \frac{\sqrt{3}}{2}$. Thus $\{ \vL^\varepsilon \}^{}_{ 0
    < \varepsilon < \frac{\sqrt{3}}{2}}$ satisfies (A1)--(A4).

    Moreover,
    \[ \bigcap_{\varepsilon' >0} \vL^{\varepsilon'} \, =\,  \bigl\{ \chi \in \widehat{G} \mid
     \chi(x)=1
     \text{ for all } x \in \vL \bigr\} \ts .\]
  Therefore for all $\varepsilon >0$ we have
  \[ \vL^{\varepsilon} + \bigcap_{\varepsilon' >0} \vL^{\varepsilon'} \,=\, \vL^{\varepsilon} \ts .\]

  The claim follows now from Corollary~\ref{capschemecor}.
\end{proof}

Given a cut and project scheme $(\RR^d \times H, \cL)$,
Meyer \cite{S-MEY} and Moody \cite{S-MOO}  introduced an important scheme in the dual spaces,
called the \emph{dual cut and project scheme}, and observed its
importance in the study of the $\varepsilon$-dual sets of a Meyer
set. Theorem~\ref{dualcps} below extends the result to arbitrary LCAG $G$.

\begin{theorem}\cite[Prop.~1]{S-JPS}\label{dualcps}
  Let\/ $(G \times H, \cL)$ be a cut and project
  scheme. Let\/ $Z$ be the annihilator of\/ $\cL$ in\/
  $\widehat{G} \times \widehat{H} \cong \widehat{G \times H}$. Then,
  $(\widehat{G} \times \widehat{H}, Z)$ is a cut and project scheme.
\end{theorem}

\begin{definition}
  Given a cut and project scheme $(G \times H, \cL)$, we
  call the scheme  $(\widehat{G} \times \widehat{H}, Z)$ from
  Theorem~\ref{dualcps} the \emph{dual scheme} to $(G \times H,
  \cL)$.
\end{definition}

It is easy to see that the $\varepsilon$-dual sets of model sets are relatively dense.

\begin{proposition}\cite[Lemma~2.6]{S-NS9}\label{T55t55}
  Let\/ $\vL \subset G$ be a subset of a model set. Then,
  $\vL^\varepsilon$ is relatively dense for all\/ $0 < \varepsilon <2$.
\end{proposition}

Combining the results of this section we get a (partial) characterisation of the
subsets of model sets in $\sigma$-compact LCAGs, similar to
\cite[Thm.~9.1]{S-MOO}. We will complete this later in Theorem~\ref{MEYchar}.

\begin{theorem}\label{T1x}
  Let\/ $\vL \subset G$ be relatively dense. Then, the following
  statements are equivalent.
\begin{itemize}\itemsep=1pt
\item[(i)] $\vL$ is a subset of a model set.
\item[(ii)] $\vL^\varepsilon$ is relatively dense for all\/ $0<
  \varepsilon <2$.
\item[(iii)] $\vL^\varepsilon$ is a model set for all\/ $0< \varepsilon <
  \frac{\sqrt{3}}{2}$.
\item[(iv)] $\vL^\varepsilon$ is a model set for some\/ $0< \varepsilon <
  \frac{\sqrt{3}}{2}$.
\end{itemize}
\end{theorem}

\begin{proof}
  The implication $\text{(i)} \Longrightarrow \text{(ii)}$ follows
  from Proposition~\ref{T55t55} and $\text{(ii)} \Longrightarrow
  \text{(iii)}$ follows from Theorem~\ref{T66t66}, while $\text{(iii)}
  \Longrightarrow \text{(iv)}$ is obvious.

  We need to show that $\text{(iv)} \Longrightarrow \text{(i)}$ to
  complete the proof.  By Proposition \ref{T55t55}, for all $0 <
  \varepsilon' < 2$, the set $(\vL^\varepsilon)^{\varepsilon'}$ is
  relatively dense. Since $\vL^\varepsilon$ is a model set, hence
  relatively dense, it follows from Theorem~\ref{T66t66} that
  $(\vL^\varepsilon)^{\varepsilon'}$ is a model set for all $0<
  \varepsilon' < \frac{\sqrt{3}}{2}$. In particular,
  $(\vL^\varepsilon)^{\varepsilon}$ is a model set, which contains the
  set $\vL$.
\end{proof}

Later we will need to use the fact that, if $\vL$ is a set which
satisfies one, and hence all, of the conditions in Theorem~\ref{T1x},
and $t \in G$, then $t+\vL$ also satisfies the conditions of
Theorem~\ref{T1x}. We will give two different proofs for this fact,
one in Corollary~\ref{translMS} based on the properties of the
$\varepsilon$-dual sets, and another one in Proposition~\ref{T1xy},
using properties of model sets.

\begin{lemma}\label{junk1}
  Let\/ $\vL \subset G$, $\varepsilon>0$ and\/ $t \in G$. If\/
  $\vL^\varepsilon$ is relatively dense, then\/ $(t+\vL)^{3\varepsilon}$
  is relatively dense.
\end{lemma}

\begin{proof}
  Select an integer $n$ so that $\frac{2\pi}{n} < \varepsilon$.  For
  each $0\leq k \leq n-1$, let $C_k:=\bigl\{ \ee^{\ii \theta} \mid
  \frac{2k\pi}{n} \leq \theta < \frac{2(k+1)\pi}{n}\bigr\}$. Then,
  $\bigcup_{k=0}^{n-1} C_k$ is a partition of $U(1)$.  Define
\[
   \vG_k \,:=\, \{  \chi \in \vL^\varepsilon \mid \chi(t) \in C_k \} \ts .
\]
It is clear by construction that
\[
  \vL^\varepsilon \,=\, \bigcup_{k=0}^{n-1} \vG_k \ts .
\]
We now show that
\[
   \vG_k - \vG_k \,\subset\, (t+\vL)^{3\varepsilon} \ts .
\]
Let $\chi^{}_{1}, \chi^{}_{2} \in \vG^{}_{k}$, and $x \in \vL$ be
arbitrary. Then,
\[
\begin{split}
    \bigl| \chi^{}_{1}(t+x)\ts\overline{\chi^{}_{2}(t+x)} -1 \bigr| \,
    & = \,
     \bigl| \chi^{}_{1}(t)\ts\overline{\chi^{}_{2}(t)}\ts
            \chi^{}_{1}(x)\ts \overline{\chi^{}_{2}(x)} -1 \bigr| \\[1mm]
     &\leq\, \bigl| \chi^{}_{1}(t)\ts \overline{\chi^{}_{2}(t)}\ts
       \chi^{}_{1}(x)\ts\overline{\chi^{}_{2}(x)} - \chi^{}_{1}(x)\ts
       \overline{\chi^{}_2(x)}\ts\bigr|\\ &\qquad +
       \bigl|  \chi^{}_{1}(x)\ts\overline{\chi^{}_{2}(x)} -
       \overline{\chi^{}_{2}(x)}\ts\bigr| +
       \bigl| \overline{\chi^{}_{2}(x)} - 1 \bigr| \\[1mm]
    &=\,\bigl| \chi^{}_1(t)\overline{\chi^{}_2(t)}  - 1 \bigr| +
       \bigl| \chi^{}_1(x)- 1\bigr| +
       \bigl|\ts\ts \overline{\chi^{}_2(x)} - 1 \bigr| \\[1mm]
    &\leq\, \bigl| \chi^{}_{1}(t)\ts\overline{\chi^{}_{2}(t)}  - 1 \bigr|
     + 2 \varepsilon
\end{split}
\]
Since $\chi^{}_{1}, \chi^{}_{2} \in \vG^{}_{k}$, we have
$\chi^{}_{1}(t)\ts\overline{\chi^{}_{2}(t)} \in \bigl\{ \ee^{\ii\theta}
\mid -\frac{2\pi}{n} \leq \theta \leq \frac{2\pi}{n} \bigr\}$, and
hence
\[
   \bigl| \chi^{}_{1}(t)\ts\overline{\chi^{}_{2}(t)} -1 \bigr| \,\leq\,
    \myfrac{2 \pi}{n} < \varepsilon \ts .
\]
Consequently,
\[
   \bigl| \chi^{}_{1}(t+x)\ts\overline{\chi^{}_{2}(t+x)} -1 \bigr|
    \, < \, 3 \varepsilon \ts .
\]
Let now $K \subset \widehat{G}$ be any compact set so that
$\vL^\varepsilon +K =\widehat{G}$.  Define the set $J:= \{ k \mid 0
\leq k \leq n-1, \vG_k \neq \varnothing \}$. For each $j \in J$, pick
some $\chi^{}_{j} \in \vG^{}_{j}$. Then,
\[
\begin{split}
  \widehat{G} \, & \subset\, \bigcup_{k=0}^{n-1} (\vG^{}_{k}+K) \, = \,
      \bigcup_{j \in J} (\vG^{}_{j}+K) \, = \,
     \bigcup_{j \in J} (\chi^{}_{j}-\vG^{}_{j}+\chi^{}_{j}+K) \\
    &\subset\, \bigcup_{j \in J} (\vG^{}_{j}-\vG^{}_{j}+\chi^{}_{j}+K)\\
    & \subset\,
     \bigcup_{j \in J} \bigl[(t+\vL)^{3\varepsilon} +\chi^{}_{j}+K\bigr] \,=\,
      (t+\vL)^{3\varepsilon} +\bigcup_{j \in J} (\chi^{}_{j}+K)
\end{split}
\]
As $\bigcup_{j \in J} (\chi^{}_{j}+K)$ is compact we are done.
\end{proof}

\begin{remark}
  The result of Lemma~\ref{junk1} can be improved easily in the
  following way. Let $\varepsilon'>0$. If one picks $n$ so that
  $\frac{2\pi}{n} < \varepsilon$, then the same proof yields
  that $(t+\vL)^{2\varepsilon+\varepsilon'}$ is relatively dense.  \exend
\end{remark}

An immediate consequence of Lemma \ref{junk1} is the following result.

\begin{corollary}\label{translMS}
  Let\/ $\vL \subset G$ be relatively dense and\/ $t \in G$. Then, the
  following statements are equivalent.
\begin{itemize}\itemsep=1pt
\item[(i)] $\vL^\varepsilon$ is relatively dense for all\/ $0<
  \varepsilon <2$.
\item[(ii)] $(t+\vL)^\varepsilon$ is relatively dense for all\/ $0<
  \varepsilon <2$.\qed
\end{itemize}
\end{corollary}

Note that, for any $\varnothing \neq \vG \subset G$ and any $\varepsilon
\geq 2$, we have $\vG^\varepsilon=\widehat{G}$. Thus,
$\vG^\varepsilon$ is trivially relatively dense for all $\varepsilon
\geq 2$.

Corollary~\ref{translMS} is also a consequence of Theorem~\ref{T1x} and the following result.

\begin{proposition}\label{T1xy}
  Let\/ $\vL \subset G$ be a subset of a model set and let\/ $F \subset G$
  be a finite non-empty set. Then, $\vL+F$ is a subset of a model set.
\end{proposition}

\begin{proof}
  Let $\oplam(W)$ be any model set containing $\vL$, and let $(G \times
  H, \cL)$ be its cut and project scheme. The problem we
  have to face is that $F$ might not be in the projection of the
  lattice, thus we will need to enlarge the lattice in the cut and
  project scheme.

  Let $K \subset H$ be any compact set so that $ W \cup \{ 0 \}
  \subset K^\circ$. Let $g \in C_{\mathsf{c}}(G)$ be so that $g(x)=1$ for all $x
  \in K$. Let $\omega=\omega_g$. Then $\omega$ is norm-almost periodic
  and $\supp(\omega)$ is a model set.

  Let $L:= \langle \supp(\omega) \cup F \rangle$. Then, the pair
  $\bigl( \{ P^{\infty}_{\!\varepsilon} (\omega) \}^{}_{0 <
    \varepsilon < \| \omega\|^{}_\infty} , L\bigr)$ satisfies
  conditions (A1)--(A5). By Corollary~\ref{123123123}, we thus obtain a
  cut and project scheme in which $\{ P^{\infty}_{\!\varepsilon} (\omega)
  \}^{}_{0 < \varepsilon < \| \omega\|^{}_\infty}$ are model sets.

  Let $W_0$ be the window of one such $P^{\infty}_{\!\varepsilon}
  (\omega)$ and let $\vG:= \supp(\omega) \subset L$.  Then, $\oplam(W_0)
  \subset \vG \subset L$, and $\vG$ is a model set, hence has finite
  local complexity. Therefore, there exists a finite set $J$ so that
\[
    \vG \,\subset\, \oplam(W_0) + J \ts .
\]
Since $L$ is a group which contains both $\oplam(W_0)$ and $\vG$, it is
easy to see that the smallest such $J$ is also a subset of $L=
\pi^{}_{1}(\cL)$. We can thus define $J^\star:= \{ j^\star \mid j
\in J \}$. Then,
\[
    \vL \,\subset\, \vG \,\subset\, \oplam(W+J^\star) \,,
\]
and, since $F \subset L$,
\[
   \vL+F \,\subset\, \oplam(W+J^\star+F^\star) \,.
\]
Thus $\vL +F$ is relatively dense and is a subset of a model set.
\end{proof}

\section{Almost lattices}

In this section we complete the characterisation of Meyer sets by studying the almost lattice property. We start by reviewing a result from \cite{S-MEY}. As the result is not stated explicitly, and the details are lost along the work, we include it here with its proof.

\begin{lemma}\label{meyl1}
  Let\/ $\vL \subset G$ be relatively dense. Then, the following
  statements are equivalent.
  \begin{itemize}\itemsep=1pt
  \item[(i)] $\vL- \vL-\vL$ is locally finite.
  \item[(ii)] $\vL- \vL-\vL$ is uniformly discrete.
  \item[(iii)] $\vL$ is locally finite and there exists a finite set\/
    $F$ so that \/$\vL-\vL \subset \vL+F$.
  \end{itemize}
\end{lemma}

\begin{proof}
  The implication $\text{(ii)}\Longrightarrow \text{(i)}$ is obvious.

  Let us now show that $\text{(i)} \Longrightarrow \text{(iii)}$.

  Since
  $\vL$ is relatively dense, there exists a compact set $K$ so that
  $\vL+K = G$. Let
\[
    F \, :=\,  (\vL-\vL-\vL) \cap K \ts .
\]
Then $F$ is a finite set. We prove now that $\vL - \vL \subset \vL +F$, which completes this implication.

Let $y \in \vL- \vL$, then $y \in G=\vL+K$, and thus $y=z+k$ for some
$z \in \vL, k \in K$. Then, $k= y-z$ and therefore $k \in
\vL-\vL-\vL$.

This shows that $k\in F$, and since $z \in \vL$, we are done.

Finally, to show that $\text{(iii)} \Longrightarrow \text{(ii)}$, note that
\[
   \vL- \vL-\vL \,\subset\, (\vL+F)-\vL \,\subset\, \vL+F+F \ts.
\]
Hence
\[
\begin{split}
   ( \vL- \vL-\vL) - (\vL- \vL-\vL)\, & \subset\,
   (\vL+F+F)-(\vL+F+F) \\
    & =\,  (\vL-\vL)+F+F-F-F \\
   & \subset\, \vL+F+F+F-F-F \ts .
\end{split}
\]
Since $\vL$ is locally finite and $F$ is finite, the set
$\vL+F+F+F-F-F$ is locally finite, and hence there exists an open
neighbourhood $U$ of zero so that $(\vL+F+F+F-F-F) \cap U = \{ 0
\}$. Then, $\vL- \vL-\vL$ is $U$-uniformly discrete.
\end{proof}

Let us recall the definition of an almost lattice \cite{S-MEY,S-MOO}.

\begin{definition}
  Let $\vL \subset G$. We say that $\vL$ is an \emph{almost lattice} if
  $\vL$ is relatively dense, locally finite and there exists a finite
  set $F$ so that $\vL-\vL \subset \vL+F$.
\end{definition}

A very important property of almost lattices, which will carry on to Meyer sets, is that they are closed under taking the Minkovski difference. We show this in the next Lemma. This property is the key for Theorem~\ref{almostlatice}, which later will allow us complete the characterisation of Meyer sets.

\begin{lemma}
  Let\/ $\vL \subset G$ be an almost lattice. Then, $\vL-\vL$ is also
  almost lattice. In particular, any almost lattice has finite local
  complexity.
\end{lemma}

\begin{proof}
  If $\vL$ is almost lattice, then $\vL-\vL \subset \vL+F$ is locally
  finite. Also, $\vL- \vL$ contains a translate of $\vL$ and thus is
  relatively dense. Moreover,
\[
   (\vL-\vL)-(\vL-\vL) \,\subset\, (\vL-\vL)+F-F \ts ,
\]
which completes the proof.
\end{proof}

Next we prove that every almost lattice is a subset of a model set. The key for this Theorem is a result of \cite{S-LS2} which shows that the autocorrelation measure of a weakly almost periodic measure is strongly almost periodic. Therefore, if $\vL$ is an almost lattice with autocorrelation $\gamma$, and if $\eta$ is the autocorrelation measure of $\gamma$, $\eta$ is a strongly almost periodic measure supported on an almost lattice. Therefore, we can apply Theorem~\ref{T2}, and the claim follows from there.

\begin{theorem}\label{almostlatice}
  Let\/ $\vL \subset G$ be an almost lattice. Then, $\vL$ is a subset of
  a model set.
\end{theorem}

\begin{proof}
  Let $\gamma$ be an autocorrelation for $\vL$. Since $\vL$ is
  relatively dense, $\gamma \neq 0$. Also, since $\vL$ has finite
  local complexity, $\supp(\gamma) \subset \vL - \vL$.

  Let $\eta$ be an autocorrelation of $\gamma$. Here, $\gamma$ is a
  positive definite measure, thus weakly almost periodic
  (see \cite{S-ARMA} or \cite{S-MoSt} in this volume). We repeat now the argument of \cite{S-LS2} to show that $\eta$ is strongly almost periodic.

  For all $f,g \in C_{\mathsf{c}}(G)$, the functions
  $f*\gamma, g*\tilde{\gamma}$ are weakly almost periodic
  functions. Their Eberlein convolution $(f*\gamma) \star
  (g*\tilde{\gamma})$ is thus a strongly almost periodic function by
  \cite[Thm.~15.1]{S-EBE} (see also \cite{S-MoSt}). An easy computation shows that $(f*\gamma)
  \star (g*\tilde{\gamma})=f*g*\eta$. Indeed, for any fixed $t \in G$,
  we have
\[
\begin{split}
   (f*\gamma) \star (g*\widetilde{\gamma})(t)\, &=
    \lim_{n \to \infty} \frac{ \int_{A_n} (f*\mu)(s)\,
     (g*\widetilde{\gamma})(t-s) \dd s}{\theta_G(A_n)} \\
    &= \lim_{n \to \infty} \frac{ \int_{G} (f*\gamma)(s) \,
       (g*\widetilde{\gamma})(t-s)\, 1^{}_{\! A_n}(s) \dd s}{\theta_G(A_n)}
\end{split}
\]
By \cite{S-BL},
\[
   f*\eta*g \, = \, \lim_{n \to \infty}
    \frac{ f*\gamma^{}_{\! A_n}*\widetilde{\gamma}*g}{\theta_G(A_n)} \ts ,
\]
and thus we get
\[
\begin{split}
    \lefteqn{\!\!\!\bigl| (f*\gamma) \star (g*\widetilde{\gamma})(t)-
    f*g*\eta (t)\bigr|} \\
     &=\,\biggl|  \lim_{n \to \infty}
     \frac{\int_{G} (f*\gamma)(s)\, (g*\widetilde{\gamma})(t-s)\,
           1^{}_{\! A_n}(s) \dd s}{\theta_G(A_n)}\\
    &\qquad - \lim_{n \to \infty}
     \frac{\int_{G} (f*\gamma^{}_{\! A_n})(s)\, (g*\widetilde{\gamma})(t-s)\,
      \dd s}{\theta_G(A_n)} \biggr|\\
    &\leq\, \lim_{n \to \infty}
     \frac{\int_{G} \bigl|  (f*\gamma)(s)\, 1^{}_{\! A_n}(s)\,
     (g*\widetilde{\gamma})(t-s) -
     (f*\gamma^{}_{A_n})(s)\, (g*\widetilde{\gamma})(t-s)
     \dd s \bigr|}{\theta_G(A_n)} \\
    &=\, \lim_{n \to \infty}
      \frac{\int_{G} \bigl|  (f*\gamma)(s)\, 1^{}_{\! A_n}(s)  -
       (f*\gamma^{}_{A_n})(s) \bigr| \cdot
      \bigl| (g*\widetilde{\gamma})(t-s)  \bigr| \dd s}{\theta_G(A_n)} \\
    &\leq\, \lim_{n \to \infty}
      \|g*\widetilde{\gamma} \|^{}_\infty
      \frac{ \int_{G} \bigl|  (f*\gamma)(s)\, 1^{}_{\! A_n}(s)  -
       (f*\gamma^{}_{A_n})(s) \bigr| \dd s}{\theta_G(A_n)} \, .
\end{split}
\]
Let $K$ be any compact set containing $\pm\supp(f)$. We claim that
\[
    (f*\gamma)(s)\, 1^{}_{\! A_n}(s)  -
    (f*\gamma^{}_{\! A_n})(s) \, \neq\,  0
    \:\Longrightarrow\: s \in \partial^K(A_n) \ts .
\]

We need to consider two cases. Start with \emph{case 1} that $s \in
A_n$. Then,
\[
    (f*\gamma)(s) \, \neq\,
     (f*\gamma^{}_{\! A_n})(s) \:\Longrightarrow\:
     \int_G \! f(s-t) \dd \gamma(t) \, \neq\,
     \int_G \! f(s-t)\, 1^{}_{\! A_n}(t) \dd \gamma(t) \ts .
\]
Thus,
\[
  \int_G  f(s-t) \, (1-1^{}_{\! A_n}(t)) \dd \gamma(t) \,\neq\, 0 \ts .
\]
Then, there must be a $t$ such that $ f(s-t)\, (1-1^{}_{\! A_n}(t)) \neq
0$, which means that $t \notin A_n$ and $s-t \in \supp(f)$.  Therefore $s
\in (G \backslash A_n)+K$. Since $s \in A_n$, we get $s
\in \partial^K(A_n)$.

Now consider \emph{case 2} that $s \notin A_n$.
Then,
\[
   0 \,\neq\, (f*\gamma^{}_{A_n})(s) \, =\,
   \int_G f(s-t)\, 1^{}_{\! A_n}(t) \dd \gamma(t) \ts .
\]
Thus, there exists a $t$ so that $f(s-t)\, 1^{}_{\! A_n}(t) \neq
0$. Hence, $t \in A_n$ and $s-t \in \supp(f)$. This shows that $s \in
A_n+K$, and therefore $s \in \partial^K(A_n)$.

Now, since $(f*\mu)(s)\, 1^{}_{\! A_n}(s) - (f*\mu^{}_{A_n})(s)=0$ for
$s$ outside $\partial^K(A_n)$, we get
\[
\begin{split}
   \bigl| (f*\gamma) &\star (g*\widetilde{\gamma})(t)-
    f*g*\eta (t)\bigr| \\
   &\leq\, \lim_{n \to \infty}
     \|g*\widetilde{\gamma} \|^{}_\infty
      \frac{ \int_{G} \bigl|  (f*\gamma)(s)\, 1^{}_{\! A_n}(s)  -
      (f*\gamma^{}_{A_n})(s) \bigr| \dd s}{\theta_G(A_n)}\\
   &= \,\lim_{n \to \infty} \|g*\widetilde{\gamma} \|^{}_\infty
      \frac{ \int_{\partial^K(A_n)} \bigl|  (f*\gamma)(s)\, 1^{}_{\! A_n}(s)  -
      (f*\gamma^{}_{A_n})(s) \bigr| \dd s}{\theta_G(A_n)} \\
   &\leq\, \lim_{n \to \infty}
      \|g*\widetilde{\gamma} \|^{}_\infty
      \frac{ \int_{\partial^K(A_n)} 2 \|  f*\gamma \|^{}_\infty
      \dd s}{\theta_G(A_n)} \\
   &=\, \lim_{n \to \infty} 2 \ts \|  f*\gamma \|^{}_\infty \,
      \|g*\widetilde{\gamma} \|^{}_\infty \ts
      \frac{ \theta_G(\partial^K(A_n))  }{\theta_G(A_n)}
    \,  =\, 0
\end{split}
\]
Thus, for all functions $f,g \in C_{\mathsf{c}}(G)$, the function $f*g*\eta$ is
almost periodic. This implies that $\eta$ is an almost periodic measure \cite{S-ARMA,S-MoSt}.

Since $\vL$ is an almost lattice, we have $\supp(\eta) \subset
(\vL-\vL)- (\vL-\vL)=: \vG$ and $\vG$ is an almost lattice.  Therefore
$\eta$ is a strongly almost periodic measure and
$\supp(\eta)-\supp(\eta)$ is uniformly discrete. Thus, by
Theorem~\ref{T2}, $\supp(\eta)$ is a model set.

As $\supp(\eta) \subset \vG$ and $\vG$ has
finite local complexity, it follows from Lemma \ref{NS1x} that there exists a finite set $J$ so that
\[
  \vG \,\subset\, \supp(\eta) +J \ts .
\]
For some $x \in \vL$, we thus have
\[
   \vL-x+x-x \,\subset\,
    (\vL-\vL)-(\vL-\vL) \,\subset\, \supp(\eta)+J \ts .
\]
By Proposition~\ref{T1xy}, $\vL$ is then a subset of a model set.
\end{proof}

We can now complete our characterisation of Meyer sets  $\sigma$-compact LCAG $G$, similar to the one in
\cite[Thm.~9.1]{S-MOO}.

\begin{theorem}\label{MEYchar}
  Let\/ $\vL \subset G$ be relatively dense. Then, the following
  properties are equivalent.
  \begin{itemize}\itemsep=1pt
  \item[(i)] $\vL$ is a subset of a model set.
  \item[(ii)] $\vL$ is harmonious.
  \item[(iii)] $\vL^\varepsilon$ is relatively dense for all\/ $0<
    \varepsilon <2$.
  \item[(iv)] $\vL^\varepsilon$ is a model set for all\/ $0< \varepsilon
    < \frac{\sqrt{3}}{2}$.
  \item[(v)] $\vL^\varepsilon$ is a model set for some\/ $0< \varepsilon
    < \frac{\sqrt{3}}{2}$.
  \item[(vi)] $\vL^\varepsilon$ is relatively dense for some\/ $0<
    \varepsilon < \frac{\sqrt{3}}{6}$.
  \item[(vii)] $\vL$ is an almost lattice.
  \item[(viii)] $\vL- \vL- \vL$ is locally finite.
\end{itemize}
\end{theorem}

\begin{proof}
  The equivalences $\text{(i)} \Longleftrightarrow \text{(iii)}
  \Longleftrightarrow \text{(iv)} \Longleftrightarrow \text{(v)}$
  follow from Theorem~\ref{T1x}.  The equivalence $\text{(vii)}
  \Longleftrightarrow \text{(viii)}$ follows from Lemma~\ref{meyl1},
  while $\text{(ii)} \Longleftrightarrow \text{(iii)}$ is proven in
  Theorem~\ref{harm} in the Appendix.

  The implications $\text{(iv)} \Longrightarrow \text{(vi)}$ and
  $\text{(i)} \Longrightarrow \text{(viii)}$ are trivial, while
  $\text{(vii)} \Longrightarrow \text{(i)}$ follows from
  Theorem~\ref{almostlatice}.

  We now show $\text{(vi)} \Longrightarrow \text{(i)}$ to complete the
  proof.  Since $\vL^\varepsilon-\vL^\varepsilon -\vL^\varepsilon
  \subset \vL^{3\varepsilon}$ and $ 3\varepsilon <
  \frac{\sqrt{3}}{2}$, it follows from Proposition~\ref{EDProp} that
  $\vL^\varepsilon-\vL^\varepsilon -\vL^\varepsilon$ is uniformly
  discrete. Also, by (vi), $\vL^\varepsilon$ is relatively dense.

  Thus $\vL^\varepsilon$ is an almost lattice by Lemma~\ref{meyl1}, and
  hence, by Theorem~\ref{almostlatice}, $\vL^\varepsilon$ is a subset
  of a model set. Then, by Theorem~\ref{T1x},
  $\vL^{\varepsilon\varepsilon}$ is a model set, which contains $\vL$.
\end{proof}

\begin{definition}
  A set $\vL \subset G$ is called a \emph{Meyer set} if $\vL$ is
  relatively dense and $\vL-\vL-\vL$ is locally finite.
\end{definition}

For a relatively dense subset $\vL \subset G$, any of the eight
conditions in Theorem~\ref{MEYchar} is equivalent to $\vL$ being a
Meyer set.

The following result follows immediately from the properties of almost
lattices and model sets.

\begin{corollary}\label{CT1333}
  Let\/ $\vL \subset G$ be a Meyer set and\/ $F$ be a finite set. Then
  $\vL+F$, $\vL \cup F$ and $\vL \pm \vL \pm ...\pm \vL$ are Meyer
  sets. In particular, $\Delta:= \vL- \vL$ a Meyer set.\qed
\end{corollary}

In the case of compactly generated LCAGs, Baake, Lenz and Moody showed
\cite{S-BLM} that the condition \emph{$\vL$ is an almost lattice} is
equivalent to the weaker condition \emph{$\vL- \vL$ is weakly
  uniformly discrete}. Thus we obtain the following result.

\begin{theorem}\label{MeyCharCompGen}
  Let\/ $G$ be a compactly generated LCAG and let\/ $\vL \subset G$ be
  relatively dense. Then, the following statements are equivalent.
  \begin{itemize}\itemsep=2pt
  \item[(i)] $\vL$ is a subset of a model set,
  \item[(ii)] $\vL$ is harmonious,
  \item[(iii)] $\vL^\varepsilon$ is relatively dense for all\/ $0<
    \varepsilon <2$,
  \item[(iv)] $\vL^\varepsilon$ is a model set for all\/ $0< \varepsilon
    < \frac{\sqrt{3}}{2}$,
  \item[(v)] $\vL^\varepsilon$ is a model set for some\/ $0< \varepsilon
    < \frac{\sqrt{3}}{2}$,
  \item[(vi)] $\vL^\varepsilon$ is relatively dense for some\/ $0<
    \varepsilon < \frac{\sqrt{3}}{6}$,
  \item[(vii)] $\vL$ is an almost lattice,
  \item[(viii)] $\vL- \vL$ is weakly uniformly discrete.
\end{itemize}
\end{theorem}

\begin{proof}
  In this case, the equivalence $\text{(vii)} \Longleftrightarrow
  \text{(viii)}$ was proven in \cite{S-BLM}. The rest of the theorem
  now follows from Theorem~\ref{MEYchar}.
\end{proof}

\section[WAP measures with Meyer support]{Weakly almost periodic measures with Meyer set support}\label{WAPMS}

In this section we study the properties of weakly almost periodic measures supported inside Meyer sets. The key tool for this section is Corollary~\ref{WAPMEYdec} below, which says that the class of weakly almost periodic measures supported inside a Meyer set is closed under the canonical almost periodic decomposition. These results have important consequences to the diffraction of measures supported inside Meyer sets, and we will look at this in Section~\ref{diffMS}.

As the definition of weakly almost periodicity for measures, and the decomposition
\[
     {\mu^{}_{}} \, =\,
    {\mu^{}_{\mathsf{s}} +{\mu^{}_{0}}} \ts ,
\]
of a weakly almost periodic measure into the strongly and null-weakly almost periodic components are technical, we skip the definitions here and instead refer the reader to \cite{S-MoSt} in this volume.

\begin{theorem}\label{WAPModdec}\cite[Thm.~4.6]{S-NS9} Let\/ $(G \times H, \cL)$ be a cut and project scheme,\/ $W \subset H$ a compact set and let\/ $\mu$ be a weakly almost periodic measure such that\/ $\supp(\mu) \subset \oplam(W)$. Then\/ $\supp(\mu^{}_{\mathsf{s}}) \subset \oplam(W)$ and\/ $\supp(\mu^{}_{0}) \subset \oplam(W)$. \qed
\end{theorem}

An immediate consequence of Theorem~\ref{WAPModdec} is the following result.

\begin{corollary}\label{WAPMEYdec} Let\/ $\mu$ be a weakly almost periodic measure in some LCAG $G$ such that\/ $\supp(\mu)$ is a subset of a Meyer set. Then\/ $\supp(\mu^{}_{\mathsf{s}})$ and\/ $\supp(\mu^{}_{0})$ are also subsets of Meyer sets.
\end{corollary}

\begin{proof}
Let $\vL$ be a Meyer set such that $\supp(\mu) \subset \vL$. By Theorem~\ref{MEYchar}, there exists some model set $\oplam(W)$ such that $\vL \subset \oplam(W)$. As $\overline{W}$ is compact and
$\supp(\mu) \subset \oplam(\overline{W})$, it follows from Theorem\ref{WAPModdec} that $\supp(\mu^{}_{\mathsf{s}})$ and $\supp(\mu^{}_{0})$ are subsets of the Meyer set $\oplam(\overline{W})$.
\end{proof}

It follows from Corollary~\ref{WAPMEYdec} that given a weakly almost periodic measure $\mu$ whose support is a subset of a Meyer set, the strongly almost periodic component $\mu^{}_{\mathsf{s}}$ satisfies the condition (iii), hence all the conditions of Theorem~\ref{T2}. It was observed in \cite[Prop.~5.7]{S-NS1} that null weakly almost periodic measures are also easy to characterize, which leads to a simple characterization of the decomposition of almost periodic measures with Meyer set support.

We start by reviewing an important equivalence from \cite[Prop.~5.7]{S-NS1}.

\begin{lemma}\label{L5.7} Let\/ $\mu$ be a weakly almost periodic measure with uniformly discrete support in some $\sigma$-compact LCAG $G$, and let\/ $A_n$ be a van Hove sequence. Then\/ $\mu$ is null weakly almost periodic if and only if
  \[
   \lim_{n \to \infty}
   \frac{\left| \mu\right|(A_n)}{\theta_G(A_n)} \,= \,0 \ts .
\]
\end{lemma}
\begin{proof} This result is proved under the extra assumption that \/ $\mu$ is Fourier transformable in \cite[Prop.~5.7]{S-NS1}. As Fourier transformability is only needed for some of the other equivalent conditions in \cite[Prop.~5.7]{S-NS1}, we skip the proof and instead we let the reader check that the proof from \cite{S-NS1} works in this case.
\end{proof}

Combining Corollary~\ref{WAPMEYdec}, Theorem~\ref{T2} and Lemma~\ref{L5.7} we get the following result.

\begin{proposition}\label{P5.7} Let\/ $\mu$ be a weakly almost periodic measure in some $\sigma$-compact LCAG $G$ such that\/ $\supp(\mu)$ is a subset of a Meyer set, and let\/ $A_n$ be a van Hove sequence . Then,
\begin{itemize}\itemsep=2pt
  \item[(i)] $\mu^{}_{\mathsf{s}}$ is norm-almost periodic.
  \item[(ii)] $   \lim_{n \to \infty}
   \frac{\left| \mu^{}_{0}\right|(A_n)}{\theta_G(A_n)} \,= \,0 \ts .$ \qed
\end{itemize}
\end{proposition}

We complete this section by proving that for weakly almost periodic measures supported inside Meyer sets, the conditions (i) and (ii) in Proposition~\ref{P5.7} characterize the weakly almost periodic decomposition.

\begin{proposition}\label{P5.8}Let\/ $\mu$ be a weakly almost periodic measure in some $\sigma$-compact LCAG $G$, and let\/ $A_n$ be a van Hove sequence. Let $\nu^{}_{1}$ and $\nu^{}_{2}$ be two measures such that
\begin{itemize}\itemsep=2pt
  \item[(i)] $\mu=\nu^{}_{1}+\nu^{}_{2}$.
  \item[(ii)] $\nu^{}_{1}$ is norm-almost periodic.
  \item[(iii)] $ \lim_{n \to \infty}  \frac{\left| \nu^{}_{2} \right|(A_n)}{\theta_G(A_n)} \,= \, 0 \ts .$
\end{itemize}
Then,
\[
   \nu^{}_{1} \,=\,
   \mu^{}_{\mathsf{s}} \quad\text{and}\quad
    \nu^{}_{2} \,=\,
   \mu^{}_{0} \ts .
\]
\end{proposition}
\begin{proof}
As $\nu^{}_{1}$ is norm almost periodic, it follows from \cite[Lemma~7]{S-BM} that $\nu^{}_{1}$ is strongly almost periodic. Thus, $\nu^{}_{2}=\mu-\nu^{}_{1}$ is weakly almost periodic as the difference of weakly almost periodic measures. Hence, for each $f \in C_{\mathsf{c}}(G)$ the limit

\[
  M(|f* \nu^{}_{2}|) = \lim_{n \to \infty}
      \frac{ \int_{A_n} \bigl|  (f*\nu^{}_{2})(s)\bigr| \dd s}{\theta_G(A_n)}\ts ,
\]
exists. To complete the proof, we need to show that this limit is always equal to zero.

Let $K:= \supp(f)$.  An easy computation shows that

\[
\int_{A_n \backslash \partial^K(A_n)} \bigl|  (f*\nu^{}_{2})(s)\bigr| \dd s \leq \bigl( \int_G  \bigl|  f(s)\bigr| ds \bigr)  \bigl|  \nu^{}_{2}\bigr| (A_n) \,.
\]

Combining this inequality with $\text{(iii)}$ and the van Hove condition, we get
\[
  M(|f* \nu^{}_{2}|) =0 \ts,
\]
Since this holds for all $f \in C_{\mathsf{c}}(G)$, it follows that $\nu^{}_{2}$ is a null weakly almost periodic measure. The claim now follows from the uniqueness of the almost periodic decomposition (see \cite{S-ARMA}, or \cite{S-MoSt} in this volume).
\end{proof}

Proposition~\ref{P5.7} and Proposition~\ref{P5.8} give a simple characterization of the weakly almost periodic decomposition for measures supported inside Meyer sets.

\begin{corollary} Let\/ $\mu$ be a weakly almost periodic measure with Meyer set support in some $\sigma$-compact LCAG $G$, and let\/ $A_n$ be a van Hove sequence. Then\/ $\mu^{}_{\mathsf{s}}$ and \/ $\mu^{}_{0}$ are uniquely characterized by
\begin{itemize}\itemsep=2pt
  \item[(i)] $\mu=\mu^{}_{\mathsf{s}}+ \mu^{}_{0} $.
  \item[(ii)] $\mu^{}_{\mathsf{s}}$ is norm-almost periodic.
  \item[(iii)] $ \lim_{n \to \infty}  \frac{\left|  \mu^{}_{0}  \right|(A_n)}{\theta_G(A_n)} \,= \, 0 \ts .$
\end{itemize}
\end{corollary}

\section[Diffraction from Meyer sets]{Diffraction
of weighted Dirac combs with Meyer set support in $\sigma$-compact
LCAGs}\label{diffMS}

The diffraction of Meyer sets in $\RR^d$ was studied in \cite{S-NS1} and
\cite{S-NS2}. As the two papers use two characterisations of Meyer sets, which were known to be equivalent in $\RR^d$ but not in general, it was previously not clear if and how these results can be generalized outside the case $G=\RR^d$. Theorem~\ref{MEYchar} eliminates now this problem.

In this section we use the results from Section~\ref{WAPMS} to study the diffraction measures of weighted Dirac comb supported inside Meyer sets. This generalizes all the results of \cite{S-NS1}
and \cite{S-NS2} in two directions: from the case $G=\RR^d$ to a $\sigma$-compact LCAG and from Meyer sets to arbitrary complex weighted Dirac combs supported inside Meyer sets.

We start by reviewing a Lemma which we will use few times.

\begin{lemma}\cite[Prop.~3.5]{S-NS1}\label{lmksd1}
  Let\/ $\mu \neq 0$ be a strongly almost periodic measure in some
  LCAG $G$. Then, $\supp(\mu)$ is relatively dense in\/ $G$. \qed
\end{lemma}

Now we use Corollary~\ref{WAPMEYdec} to show that both the discrete and continuous diffraction spectra of a weighted Dirac comb with Meyer set support are almost periodic measures.

\begin{theorem}\label{S21}
  Let\/ $\omega$ be any weighted Dirac comb supported inside some Meyer set, and let\/ $\gamma$ be any autocorrelation for\/ $\omega$. Then\/ $\widehat{\gamma}^{}_{\mathsf{pp}}$ and\/
$\widehat{\gamma}^{}_{\mathsf{c}}$ are strongly almost periodic measures.

In particular, each of the pure point and continuous diffraction spectra of\/ $\omega$ is either trivial or has a relatively dense support.
\end{theorem}

\begin{proof}
Let $\vL$ be any Meyer set containing $\supp(\omega)$. Then, $\supp(\gamma) \subset \vL-\vL$. It follows from Corollary~\ref{CT1333} that $\supp(\gamma)$ is a subset of a Meyer set.
Thus, by Corollary~\ref{WAPMEYdec}, $\gamma^{}_{\mathsf{s}}$ and ${\gamma^{}_{0}}$ are supported inside a Meyer set, hence they are discrete measures.

As $\gamma$ is a positive definite measure, it is Fourier Transformable and by \cite{S-MoSt} or \cite{S-LS2} we have
\[
    \widehat{\gamma^{}_{\mathsf{s}}} \,=\,
    (\widehat{\gamma})^{}_{\mathsf{pp}} \quad\text{and}\quad
    \widehat{\ts\gamma^{}_{0}\ts} \,= \,\nts
    (\widehat{\gamma})^{}_{\mathsf{c}} \ts .
\]

Hence, by \cite[Cor.~11.1]{S-ARMA}, both measures
$(\widehat{\gamma})^{}_{\mathsf{pp}}$ and
$(\widehat{\gamma})^{}_{\mathsf{c}}$ are strongly almost periodic.

The rest of the proof follows from Lemma~\ref{lmksd1}.
\end{proof}

The following corollary is immediate consequence of
Proposition~\ref{S21}.

\begin{corollary}\label{cms}
  Let\/ $M \subset G$ be a Meyer set with autocorrelation \/ $\gamma$. Then,
  $(\widehat{\gamma})^{}_{\mathsf{pp}}$ and\/
  $(\widehat{\gamma})^{}_{\mathsf{c}}$ are strongly almost periodic, and\/
  $(\widehat{\gamma})^{}_{\mathsf{pp}}\neq 0$.

  In particular, $M$ has a relatively dense set of Bragg
  peaks. Moreover, either\/ $M$ is pure point diffractive or\/ $M$ has a
  relatively dense continuous spectrum.
\end{corollary}
\begin{proof}
$(\widehat{\gamma})^{}_{\mathsf{pp}}\neq 0$ follows immediately from the observation
\[
    \widehat{\gamma}( \{0 \}) \,\geq\,
     \bigl(\underline{\dens}(M)\bigr)^{2} \, >  \, 0 \ts ,
\]
compare \cite[Cor.~9.1]{S-TAO} for the case $G=\RR^{d}$.

The rest of the proof follows now from Proposition~\ref{S21}.
\end{proof}

One simple consequence of these results is the fact that, if a Meyer
set has no singularly continuous spectrum, and if its absolutely
continuous spectrum is given by some uniformly continuous bounded
function $f$, then $f$ is automatically an almost periodic function.

\begin{corollary}
  Let\/ $\omega$ be a weighted Dirac comb supported inside a Meyer set and let\/ $\gamma$ be an autocorrelation
  of\/ $\omega$.  Suppose that\/ $(\widehat{\gamma})^{}_{\mathsf{sc}}=0$ and\/
  $(\widehat{\gamma})^{}_{\mathsf{ac}}=f \theta_G$ for some\/ $f \in
  C_{\mathsf{u}}(\widehat{G})$. Then, $f$ is an almost periodic function on\/
  $\widehat{G}$.
\end{corollary}

\begin{proof}
  We know that the measure $(\widehat{\gamma})^{}_{\mathsf{c}}=f
  \theta_G$ is a strongly almost periodic measure.  It follows
  from \cite[Cor.~5.2]{S-ARMA} or \cite{S-MoSt} that $f$ is an almost periodic
  function. Since the proof of this result is relatively simple, we
  include it here.

  Let $f_\alpha \subset C_{\mathsf{c}}(G)$ be an approximate identity
  for the convolution algebra $(C_{\mathsf{u}}(G), *)$. Then, for all
  $g \in C_{\mathsf{u}}(G)$, the net $g*f_\alpha$ converges uniformly
  to $g$.

  Since $f \theta_G$ is a strongly almost periodic measure, for all $\alpha$ the function $f
  *f_\alpha$ is almost periodic. Therefore $f$ is the uniform limit of the net of almost periodic functions $f*f_\alpha$. The claim follows now from the fact that the space of almost periodic functions is closed in $(C_{\mathsf{u}}(G), \| \, \|_\infty)$ \cite{S-EBE,S-MoSt}.
\end{proof}

Corollary~\ref{cms} tells us that, for all $f \in K_2(\widehat{G})$, the
sets
\[
    P_{\mathsf{pp}}(\varepsilon; f)  \, := \,
    \bigl\{ t \in \widehat{G} \mid
     \bigl\|T_t( f*(\widehat{\gamma})^{}_{\mathsf{pp}}) -
     f* (\widehat{\gamma})^{}_{\mathsf{pp}} \bigr\|^{}_\infty < \varepsilon\bigr\}
\]
and
\[
   P_{\mathsf{c}}(\varepsilon; f) \, := \,
   \bigl\{ t \in \widehat{G} \mid
    \bigl\|T_t(f*(\widehat{\gamma})^{}_{\mathsf{c}}) -
     f* (\widehat{\gamma})^{}_{\mathsf{c}} \bigr\|^{}_\infty < \varepsilon \bigr\}
\]
are relatively dense.

As our next result, we will prove directly that these sets contain
some $\varepsilon$-dual sets, thus providing new proofs for
Corollary~\ref{cms} and Theorem~\ref{S21}.  More precisely, we
prove that there exists a $C := C(f; \gamma)$ so that, for any regular
model set $\vG$ containing $\vL$, we have $(\vG-\vG)^\varepsilon
\subset P_1(C\varepsilon; f) \cap P_2(C\varepsilon; f)$.

For a finite measure $\nu$, we will denote its total variation by $\|
\nu \|$, that is
\[
   \| \nu \| \, := \, \left| \nu \right| (G) \ts .
\]

\begin{lemma}\label{mu2}
  Let\/ $\mu$ be a twice Fourier transformable measure, and let\/
  $\Delta$ be a concentration set for\/ $\mu$. Let\/ $f \in
  K_2(\widehat{G})$. Then\/ $\| \widecheck{f} \mu \|< \infty$, and for all\/
  $\varepsilon >0$ and all\/ $\chi \in (\Delta)^\varepsilon$ we have
\[
    \bigl\| T_\chi (f *\widehat{\mu})- f*\widehat{\mu} \bigr\|^{}_\infty
     \,\leq\,   \bigl\| \widecheck{f}  \mu \bigr\| \ts\varepsilon \ts .
\]
\end{lemma}

\begin{proof}
  Since $f \in K_2(\widehat{G})$, and $\widehat{\mu}$ is Fourier
  transformable, $\widecheck{f} \in L^1(\mu) $. Hence $\widecheck{f}
  \mu$ is a finite measure, and thus
\[
   \bigl\| \widecheck{f}  \mu \bigr\|  \,<\, \infty \ts .
\]
Now, let $\chi \in \Delta^\varepsilon$, and let $\psi \in
\widehat{G}$. Then,
\[
\begin{split}
   \bigl| T_\chi( &f*\widehat{\mu})(\psi) -
             (f* \widehat{\mu}) ( \psi) \bigr|\,
   =\, \bigl| ( f*\widehat{T_\chi \mu})(\psi) -
     (f* \widehat{\mu}) ( \psi) \bigr|  \\[1mm]
   &=\, \bigl| f*\widehat{(T_\chi \mu-\mu)}(\psi) \bigr|
   \, = \, \left| \int_G \widecheck{f}(t)\,
       (\chi(t)-1)\, \psi(t) \dd \mu(t) \right|\\[1mm]
   &\leq \int_G \bigl|  \widecheck{f}(t)\, (\chi(t)-1)\,\psi(t) \bigr|
     \dd |\mu |(t)
   \, = \int_\Delta \bigl| \widecheck{f}(t)  \bigr|\, \bigl|\chi(t)-1 \bigr|\,
     \bigl|\psi(t) \bigr| \dd |\mu |(t) \\[1mm]
   &\leq \, \varepsilon \int_\Delta \bigl| \widecheck{f}\bigr| (t)
     \dd | \mu | (t)
    \,\leq\, \varepsilon \ts  \bigl\| \widecheck{f}  \mu \bigr\| \ts ,
\end{split}
\]
which completes the argument.
\end{proof}

Using Lemma~\ref{mu2} for $\gamma^{}_{\mathsf{s}}$ and $\gamma^{}_{0}$, we obtain
the following result.

\begin{proposition}\label{S22}
  Let\/ $\omega$ be any translation bounded measure, let\/ $\gamma$ be any autocorrelation for $\omega$ and
  let\/ $\Delta$ be any concentration set for both\/ $\gamma^{}_{\mathsf{s}}$
  and\/ $\gamma^{}_{0}$.  Let\/ $f \in K_2(\widehat{G})$. If\/ $\gamma$ is twice Fourier transformable, then, there
  exists a\/ $C$ which depends on\/ $f$, $\gamma^{}_{\mathsf{s}}$ and\/
  $\gamma^{}_{0}$, so that, for all $\varepsilon >0$, we have
\[
   \Delta^\varepsilon \,\subset\,
    P_{\mathsf{pp}}(C\varepsilon; f) \,=\,
    \bigl\{ t \in \widehat{G} \mid
     \bigl\|T_t( f*(\widehat{\gamma})^{}_{\mathsf{pp}}) -
     f* (\widehat{\gamma})^{}_{\mathsf{pp}} \bigr\|^{}_{\infty}
     < \varepsilon \bigr\}
\]
and
\[
    \Delta^\varepsilon \,\subset\,
     P_{\mathsf{c}}(C\varepsilon; f) \,=\,
     \bigl\{ t \in \widehat{G} \mid
    \bigl\|T_t( f*(\widehat{\gamma})^{}_{\mathsf{c}}) -
    f* (\widehat{\gamma})^{}_{\mathsf{c}} \bigr\|^{}_{\infty}
    < \varepsilon \bigr\}  \ts .
\]
\end{proposition}

\begin{proof}
First lets note that since $\gamma$ is Fourier transformable, so are $\gamma^{}_{\mathsf{s}}$
  and $\gamma^{}_{0}$ \cite{S-MoSt}. Moreover, as $\widehat{\gamma}$ is Fourier transformable, the measures $\widehat{\gamma})^{}_{\mathsf{pp}}$ and $\widehat{\gamma})^{}_{\mathsf{c}}$ are also Fourier transformable \cite[Thm.~11.1]{S-ARMA}. This implies that each of $\gamma^{}_{\mathsf{s}}$
  and\/ $\gamma^{}_{0}$ is twice Fourier transformable.

Using
\[
   C  \, := \, \max \Bigl\{ \bigl\| \widecheck{f}
            \gamma^{}_{\mathsf{s}} \bigr\| ,
            \bigl\|  \widecheck{f}  \gamma^{}_{0} \bigr\| \Bigr\} \ts ,
\]
this becomes an immediate consequence of Lemma~\ref{mu2}.
\end{proof}

Note that, unless $\Delta$ is uniformly discrete,
$(\Delta)^\varepsilon$ cannot be relatively dense. If $\vL$ is Meyer,
combining Theorem~\ref{S21} with Proposition~\ref{S22} establishes the
following result.

\begin{proposition}
  Let\/ $\omega$ be an weighted Dirac comb supported inside some Meyer set, let\/ $\oplam(W)$ be any regular
  model set containing\/ $\supp(\omega)$, with\/ $W$ compact and let\/ $\Delta:= \oplam(W) - \oplam(W)$. Let \/ $\gamma$ be any autocorrelation of\/ $\omega$ and let\/ $f \in K_2(\widehat{G})$.

  If\/ $\gamma$ is twice Fourier transformable, then there exists a\/ $C$ which depends
  only on\/ $f$, $\gamma^{}_{\mathsf{s}}$ and\/ $\gamma^{}_{0}$, so that, for all\/
  $\varepsilon >0$, we have
\[
    \Delta^\varepsilon \,\subset\,
    P_{\mathsf{pp}}(C\varepsilon; f) \,=\,
    \bigl\{ t \in \widehat{G} \mid
    \bigl\|T_t( f*(\widehat{\gamma})^{}_{\mathsf{pp}}) -
    f* (\widehat{\gamma})^{}_{\mathsf{pp}} \bigr\|^{}_{\infty} <
    \varepsilon \bigr\}
\]
and
\[
    \Delta^\varepsilon \,\subset\,
    P_{\mathsf{c}}(C\varepsilon; f) \,=\,
     \bigl\{ t \in \widehat{G} \mid
     \bigl\|T_t( f*(\widehat{\gamma})^{}_{\mathsf{c}}) -
     f* (\widehat{\gamma})^{}_{\mathsf{c}} \bigr\|^{}_{\infty}
    < \varepsilon \bigr\} .
\]
In particular, $f*(\widehat{\gamma})^{}_{\mathsf{pp}}$ and\/
$f*(\widehat{\gamma})^{}_{\mathsf{c}}$ are strongly almost periodic
functions.
\end{proposition}

The following results show that, for Meyer sets and, more generally,
for any measure whose autocorrelation is supported on a Meyer set, the
pure point spectrum is sup-almost periodic. Moreover, we show that,
up to multiplication by some constant, the $\varepsilon$-dual sets are
subsets of the $\varepsilon$ sup-almost periods of the pure point
spectrum.

\begin{proposition}\label{S31}
  Let\/ $\gamma$ be a Fourier transformable measure and let\/
  $\Delta$ be any concentration set for\/ $\gamma$.  Then, there exists some\/ $C >0$ such that for all\/
  $\varepsilon >0$ and all\/ $\psi \in \Delta^{\varepsilon}$, we have
\[
    \bigl| \widehat{\gamma}(\{ \psi+\chi\})-
     \widehat{\gamma}(\{\chi\}) \bigr|
     \,\leq\, C \ts\ts \varepsilon \ts .
\]
In particular, we have
\[
   \Delta^{\varepsilon/C}_{} \,\subset\,
    P_{\!\varepsilon}^{\infty} \bigl( (\widehat{\gamma})^{}_{\mathsf{pp}}\bigr) \ts .
\]
\end{proposition}

\begin{proof}

While this result is more general, the proof below is almost identical to the proofs of \cite[Th.~3.1]{S-NS2} and \cite[Thm.~5.1]{S-NS9}.

As $\gamma$ is translation bounded, there exists a constant $0 < C < \infty$ such that
\[
\limsup_{n \to \infty} \frac{ \bigl| \gamma \bigr|(A_n) }{\theta_G(A_n)} \, < \, C \ts .
\]

Let $\chi \in \widehat{G}$ and $\psi \in
  \Delta^{\varepsilon}$. Then, by \cite[Cor.5]{S-L1} or \cite{S-MoSt} we have

\[
\begin{split}
    \bigl| \widehat{\gamma}(\{& \psi+\chi\})-
     \widehat{\gamma}(\{\chi\}) \bigr| \,
    =\, \lim_{n \to \infty} \biggl|
     \frac{ \int_{A_n} \bigl(\psi(t)\chi(t)-\chi(t)\bigr)
     \dd \gamma (t)}{\theta_G(A_n)} \biggr| \\[1mm]
    &\leq\, \lim_{n \to \infty}
    \frac{ \int_{A_n}  \bigl| \psi(t)\chi(t)-\chi(t)\bigr|
     \dd \bigl| \gamma \bigr| (t)  }{\theta_G(A_n)}
    \,=\, \lim_{n \to \infty}
    \frac{ \int_{A_n}  | \psi(t)-1 |
    \dd \bigl| \gamma\bigr|(t)}{\theta_G(A_n)} \\[1mm]
    &=\, \lim_{n \to \infty} \frac{ \int_{A_n \cap \Delta} | \psi(t)-1 |
    \dd \bigl| \gamma\bigr| (t)  }{\theta_G(A_n)}
    \,\leq\, \lim_{n \to \infty}
    \frac{ \int_{A_n \cap \Delta}  \varepsilon
     \bigl|\gamma\bigr|(t)  }{\theta_G(A_n)}\\[1mm]
    &=\, \varepsilon
    \lim_{n \to \infty} \frac{\bigl| \gamma \bigr|(A_n) }{\theta_G(A_n)}
    \, < \, \varepsilon\ts C \ts ,
\end{split}
\]
which completes the proof.
\end{proof}

\begin{corollary}\label{cms33}
  Let\/ $\vL$ be a Meyer set in\/ $G$, and let\/ $\Delta=\vL-
  \vL$. Then, there exists a\/ $C$ so that, for all\/ $\varepsilon
  >0$, $\Delta^{\varepsilon/C} \subset
  P_{\!\varepsilon}^{\infty}\bigl((\widehat{\gamma})^{}_{\mathsf{pp}}\bigr)$. In
  particular, $(\widehat{\gamma})^{}_{\mathsf{pp}}$ is sup-almost
  periodic.\qed
\end{corollary}

Combining Corollary~\ref{cms33} and Corollary~\ref{123123123} yields
the following result.

\begin{corollary}\label{cms445}
  Let\/ $\vL$ be a Meyer set in\/ $G$, and let\/ $\Delta=\vL-
  \vL$. Then, for all\/ $0 < \varepsilon <1$, the set\/
  $P^{}_{\widehat{\gamma}(\{0\})
    \varepsilon}\bigl((\widehat{\gamma})^{}_{\mathsf{pp}}\bigr)$ is a
  model set containing\/ $\Delta^\varepsilon$.\qed
\end{corollary}

More generally, we obtain the following consequence.

\begin{corollary}
  Let\/ $\mu$ be any translation bounded measure on\/ $G$, let\/
  $\gamma$ be an autocorrelation for\/ $\mu$ and let\/ $\Delta=
  \supp(\gamma)$. Then there exists a\/ $C > 0$ so
  that, for all\/ $\varepsilon >0$, $\Delta^{\varepsilon/C} \subset
  P_{\!\varepsilon}^{\infty}\bigl(
  (\widehat{\gamma})^{}_{\mathsf{pp}}\bigr)$. In particular, if\/
  $\Delta$ is a Meyer set, then\/
  $(\widehat{\gamma})^{}_{\mathsf{pp}}$ is sup-almost periodic.\qed
\end{corollary}

In the rest of the section we will study the set of Bragg peaks of the
diffraction measure with intensity above a certain threshold. Most of
the results will hold more generally for (positive and) positive
definite measures. For a positive definite measure $\gamma$ and $a>0$,
we define
\[
  I(a) \, := \,  \{ \chi \in \widehat{G} \mid
   \widehat{\gamma}(\{ \chi \}) \geq a \} \ts .
\]
We will call $I(a)$ the set of \emph{$a$-visible Bragg peaks}.

Let $a_{\text{max}}:= \sup_{\chi \in \widehat{G}} \bigl\{
\widehat{\gamma}(\{ \chi \})\bigr\}$. If $\gamma$ is positive, it is
easy to see that $a_{\text{max}}= \widehat{\gamma}(\{0\})$.

For $a \leq 0$, we have $I(a)= \widehat{G}$, while we have $I(a)=
\varnothing$ for $a > a_{\text{max}}$. Also, Lemma~\ref{L13} shows that,
in general, for aperiodic positive $\gamma$ we have
$I(a_{\text{max}})=\{ 0\}$. Hence, we are usually interested in the
case $0< a < a_{\text{max}}$.

We first need to recall that, given any positive definite measure $\mu$,
we can use Krein's inequality for the discrete component $\mu^{}_{d}$.

\begin{lemma}\cite{S-LS1}\label{positive-definite-measure-function}
  Let\/ $\mu$ be a positive definite measure on\/ $G$. Then, the
  discrete part\/ $\mu^{}_{d}$ is a positive definite measure on\/ $G$, and the
  function\/ $f\! :\, G \longrightarrow \CC$, defined by\/ $f(x) = \mu
  (\{ x \})$, is a positive definite function. In particular, for
  all\/ $x, t \in G$, we have
\[
   \bigl| \mu(\{ x+t\}) - \mu(\{ x \}) \bigr|^2
   \,\leq\, 2\ts \mu(\{ 0 \})\ts
   \bigl[ \mu(\{ 0 \}) - \Real \bigl(\mu(\{ t\})\bigr)\bigr] \ts .
\]
\end{lemma}

\begin{lemma}\label{L13}
  Let\/ $\gamma$ be a positive and positive definite measure. If\/
  $\widehat{\gamma}(\{\chi\})=a_{\rm max}$ for some\/ $\chi \neq
  0$, then\/ $\chi$ is a period for\/ $\widehat{\gamma}$.
\end{lemma}

\begin{proof}
  Since $\widehat{\gamma}$ is a positive definite measure,
  Lemma~\ref{positive-definite-measure-function} implies that
\[
    \bigl|  \widehat{\gamma} (\{ \chi + \psi \}) -
         \widehat{\gamma}(\{ \psi\}) \bigr|^2
     \,\leq\, 2\, \widehat{\gamma}(\{ 0\}) \,
     \bigl[  \widehat{\gamma}(\{0\}) -
     \Real\bigl( \widehat{\gamma}(\{ \chi\})\bigr)\bigr]
\]
holds for all $\psi \in \widehat{G}$.  Since $\widehat{\gamma}(\{0\})
= a_{\text{max}} = \widehat{\gamma}(\{ \chi\}) \in [0, \infty)$, we
conclude that
\[
   \bigl|  \widehat{\gamma} (\{ \chi + \psi \}) -
    \widehat{\gamma}(\{ \psi\}) \bigr|^2  \,\leq\,  0\ts ,
\]
which completes the proof.
\end{proof}

Next, we extend \cite[Thm.~4.4]{S-NS2} to arbitrary $\sigma$-compact, LCAG $G$.

\begin{corollary}\label{CT1x}
  Let\/ $\gamma$ be a  positive definite measure with\/ $\widehat{\gamma}_{pp} \neq 0$ and let\/
  $\Delta= \supp(\gamma)$. If\/ $\Delta$ is a Meyer set, then, for all
  $0< a < a_{\text{max}}$, the set of $a$-visible Bragg
  peaks is a Meyer set.
\end{corollary}

\begin{proof}
  Pick some $C >1$ as in Proposition~\ref{S31}. Let $0
  < a <  a_{\text{max}}$. Fix some $0< \varepsilon <
  \max\bigl\{ \frac{ a_{\text{max}}-a}{C}, \frac{a}{C},
  \frac{\sqrt{3}}{2} \bigr\}$.

We start by showing that $I(a)$ contains a translate of $\Delta^\epsilon$.

As $a+ C \ts \varepsilon < a_{\text{max}}$, there exists some $ \phi$ such that
\[
\widehat{\gamma}(\{ \phi \}) > a+ C \ts \epsilon \ts.
\]

Let $\psi \in \Delta^\varepsilon$ be arbitrary. By Proposition~\ref{S31}, we have
\[
    \bigl| \widehat{\gamma}(\{ \psi +\phi \})-
     \widehat{\gamma}(\{\phi\}) \bigr|
     \,\leq\, C \ts\varepsilon \,< \, \widehat{\gamma}(\{ \phi \}) -a  \ts .
\]
Hence $\widehat{\gamma}(\{ \psi +\phi \}) \geq a$, and thus
\[
 \phi+ \Delta^\varepsilon \,\subset\, I(a) \ts .
\]

Next, we show that $I(a)$ is contained in finitely many translates of $\Delta^\varepsilon$.

By Proposition~\ref{S31}, for all $\psi \in \Delta^\varepsilon$
and $\chi \in I(a)$, we have
\[
   \bigl| \widehat{\gamma}(\{ \psi \pm \chi\}) -
     \widehat{\gamma}(\{\chi\}) \bigr| \,\leq\, C\ts\varepsilon  \ts ,
\]
and thus
\[
  \widehat{\gamma}(\{ \psi \pm \chi\}) \,\geq\,
  a- C\ts\varepsilon \,>\,0 \ts .
\]
and
\[
   I(a) \pm \Delta^\varepsilon \,\subset\, I(a-C\varepsilon) \ts .
\]
Since $\Delta$ is a Meyer set and $\varepsilon<\frac{\sqrt{3}}{2}$, the set $\Delta^\varepsilon$ is relatively dense. Therefore, there exists a compact set $K$ so that
$\Delta^\varepsilon+K =\widehat{G}$.

Since $\widehat{\gamma}$ is translation bounded, and
$a-C\varepsilon>0$, the set $I(a-C\varepsilon)$ is locally
finite.

Thus $F:= I(a-C\varepsilon) \cap K$ is finite.

Let now $\chi \in I(a)$ be arbitrary. Then $\chi \in G$, hence $\chi=
\psi+\varphi$ with $\psi \in \Delta^\varepsilon$ and  $\varphi \in K$.
Then,
\[
   \varphi\, =\, \chi -\psi \,\in\,
    I(a) - \Delta^\varepsilon \, \subset\,  I(a-C\varepsilon) \ts ,
\]
and hence $\varphi \in F$. This  proves that
\[
   I(a) \,\subset\, \Delta^\varepsilon +F \ts .
\]

Since $\vL$ is a Meyer set, so is $\Delta$ and hence, by
Theorem~\ref{MEYchar}, $\Delta^\varepsilon$ is a Meyer set. Therefore $\Delta^\varepsilon+F$ is a Meyer set.
\[
  \phi+ \Delta^\varepsilon \,\subset\, I(a) \,\subset\,
   \Delta^\varepsilon +F \ts ,
\]
we conclude that $I(a)$ is a Meyer set.
\end{proof}

\begin{corollary}
  Let\/ $\vL$ be a Meyer set in\/ $G$, and let\/ $\Delta=\vL-
  \vL$. Then, for all\/ $0< a < \widehat{\gamma}(\{ 0 \})$, the set
  of\/ $a$-visible Bragg peaks is a Meyer set.
\end{corollary}

Note that all results in this section depend on the autocorrelation we
may choose for the Meyer set. Frequently, a Meyer set can have multiple
autocorrelations. An interesting question posed by Michael Baake is:
Given two different autocorrelations of the same Meyer set, how
different can the two diffractions be?

We will conclude this section by showing that, given two
autocorrelations of the same Meyer set, there does not seem to be any
connection between their continuous spectra, yet their pure point
spectra are somehow related.

\begin{example}
  Let $\vL_1, \vL_2 \subset \RR$ be two Meyer subsets of the same
  model set. Let
\[
   \vL \,:= \, \bigl[ \vL_{1} \cap (-\infty, 0)\bigr] \cup
               \bigl[\vL_{2} \cap [0, \infty) \bigr]\ts .
\]
If $A_n=(-n^2,n)$ or $A_n=(-n^2, 0)$, then any autocorrelation of
$\vL$ with respect to $A_n$ is actually an autocorrelation of
$\vL_{1}$.  Similarly, if $B_n=(-n,n^2)$ or $B_n=(0, n)$, then any
autocorrelation of $\vL$ with respect to $B_n$ is actually an
autocorrelation of $\vL_2$.

Now, if we pick $\vL_1=\ZZ$ and $\vL_2$ to be a random Bernoulli
subset of $\ZZ$, and $C_{2n}=[0,n^2]$, $C_{2n+1}=[-n^2,0]$, then, with
respect to the van Hove sequence $C_n$, the Meyer set $\vL:= \bigl[
\vL_1 \cap (-\infty, 0)\bigr] \cup \bigl[\vL_2 \cap [0, \infty)\bigr]$
has both
\[
    \gamma^{}_{1} \,=\, \delta^{}_{\ZZ} \quad\text{and}\quad
    \gamma^{}_{2} \,=\,\myfrac{1}{4}\ts\delta^{}_{\ZZ}  +
                      \myfrac{1}{4}\ts\delta^{}_0
\]
as autocorrelations. However, while $\gamma^{}_{1}$ is pure point
diffractive, $\gamma^{}_{2}$ has mixed pure point and absolutely
continuous spectrum.

If we set $\vL_{1}=\ZZ$ and $\vL_{2}$ to be the set of $a$-type
Thue--Morse, and again consider $\vL:= \bigl[ \vL_1 \cap (-\infty,
0)\bigr] \cup \bigl[\vL_2 \cap [0, \infty) \bigr]$, then, with respect
to the same van Hove sequence, $\vL$ has two different
autocorrelations, one being pure point diffractive, while the other
has mixed pure point and singularly continuous spectrum.

One can also mix the Thue--Morse system with a Bernoulli system to
construct a Meyer set with two autocorrelations, both with mixed
spectrum, one pure point and absolutely continuous, the other pure
point and singularly continuous.
\exend
\end{example}

\begin{remark}
  While this example is not homogeneous, it is natural in the
  following sense. If $\vL$ has two different autocorrelation
  measures, it means that arbitrarily large finite parts of $\vL$ have
  different statistical properties, and thus are different. In some
  sense, two `halves' (parts) of $\vL$ are completely different.

  A more interesting example is the following one. Let $\vL_1$ be the
  $a$-type Thue--Morse set and $\vL_2$ be a Bernoulli subset of
  $\ZZ$. Let
\[
\begin{split}
   \vL \,:= \, & \Biggl[ \bigcup_{n \in \ZZ_+}
               \Bigl[ \ZZ \cap \pm \bigl[3^{3n}, 3^{3n+1}\bigr) \Bigr] \Biggr]
               \cup \Biggl[ \bigcup_{n \in \ZZ_+}
               \Bigl[ \vL_1 \cap \pm \bigl[3^{3n+1}, 3^{3n+2}\bigr)
               \Bigr] \Biggr] \\
              & \quad \cup \Biggl[ \bigcup_{n \in \ZZ_+}
              \Bigl[ \vL_2 \cap \pm \bigl[3^{3n+2}, 3^{3n+3}\bigr)
              \Bigr] \Biggr] \ts .
\end{split}
\]
Then, $\vL$ has each of the autocorrelation of $\ZZ$, $\vL_1$ and
$\vL_2$ as cluster points, and also their sums.
\exend
\end{remark}

This example shows that, given two different autocorrelations of the
same Meyer set, we should not expect any connection between their
continuous spectra.

What about the pure point spectrum?

Let $\vL$ be a Meyer set and let $\gamma_1$ and $\gamma_2$ be two
autocorrelations of $\vL$. Let $\vG$ be any regular model set
containing $\vL$. Then, as we saw in the proof of
Proposition~\ref{S21}, since the autocorrelation $\gamma$ of $\vG$
exists and is the same with respect to any van Hove sequence, there
exists some finite sets $F_1$ and $F_2$ so that
\[
   (\gamma^{}_{1})^{}_{\mathsf{s}} \,\leq\, \gamma
   \,\leq\, \delta^{}_{F_1} * (\gamma^{}_1)^{}_{\mathsf{s}}
\]
and
\[
   (\gamma^{}_{2})^{}_{\mathsf{s}} \,\leq\, \gamma \,\leq\,
    \delta^{}_{F_2} * (\gamma^{}_{2})^{}_{\mathsf{s}} \ts .
\]
Combining the two, we get that there exists a finite set $F$ so that
\[
 (\gamma^{}_{1})^{}_{\mathsf{s}} \,\leq\,
  \delta^{}_{F} * (\gamma^{}_{2})^{}_{\mathsf{s}}
\]
and
\[
 (\gamma^{}_{2})^{}_{\mathsf{s}} \,\leq\,
  \delta^{}_{F} * (\gamma^{}_{1})^{}_{\mathsf{s}} \ts .
\]
Thus, the strongly almost periodic components of the two
autocorrelations are related as shown, which implies a relation
between the two pure point spectra. However, it is not apparent how
precisely the pure point spectra are related.

We complete this section by using the results we just proved to add few extra conditions to
the characterisation of Meyer sets in Theorem~\ref{MEYchar}.

\begin{theorem}
  Let\/ $\vL \subset G$ be relatively dense. Then, the following
  statements are equivalent.
  \begin{itemize}\itemsep=1pt
  \item[(i)] $\vL$ is a subset of a model set.
  \item[(ii)] $\vL$ is harmonious.
  \item[(iii)] $\vL^\varepsilon$ is relatively dense for all\/ $0<
    \varepsilon <2$.
  \item[(iv)] $\vL^\varepsilon$ is a model set for all\/ $0< \varepsilon
    < \frac{\sqrt{3}}{2}$.
  \item[(v)] $\vL^\varepsilon$ is a model set for some\/ $0< \varepsilon
    < \frac{\sqrt{3}}{2}$.
  \item[(vi)] $\vL^\varepsilon$ is relatively dense for some\/ $0<
    \varepsilon < \frac{\sqrt{3}}{6}$.
  \item[(vii)] $\vL$ is an almost lattice.
  \item[(viii)] $\vL- \vL- \vL$ is locally finite.
  \item[(ix)] There exists a regular sup-almost periodic discrete measure\/
    $\eta$ and an\/ $0 < \varepsilon < \| \eta \|^{}_\infty$ so that\/ $\vL
    \subset P_{\!\varepsilon}^{\infty}(\eta)$.
  \item[(x)] There exists a strongly almost periodic measure\/ $\eta$
    with\/ $\supp(\eta)-\supp(\eta)$ uniformly discrete, so that\/ $\vL
    \subset \supp( \eta)$.
  \item[(xi)] There exists a regular sup-almost periodic discrete measure\/
    $\eta$ so that\/ $\supp(\eta)$ has finite local complexity and\/ $\vL
    \subset \supp( \eta)$.
  \item[(xii)] There exists a norm-almost periodic discrete measure\/ $\eta$ so
    that\/ $\supp(\eta)$ has finite local complexity and\/ $\vL
    \subset \supp( \eta)$.
  \item[(xiii)] There exists a regular sup-almost periodic measure\/
    $\eta$ and\/ $a>0$ so that\/ $\supp(\eta)$ has finite local complexity and\/ $\vL\subset
    \{ x \in G \mid \left| \eta(\{ x \}) \right| > a \}$.
  \item[(xiv)] There exists a
    regular sup-almost periodic measure\/ $\eta$ so that\/ $\supp(\eta)$ has finite local complexity and\/
    $\delta^{}_{\nts\vL} \leq \eta $.
  \item[(xv)] There exists a
    norm-almost periodic measure\/ $\eta$ so that\/ $\supp(\eta)$ has finite local complexity and \/ $\delta_{\nts\vL} \leq
    \eta$.
 \end{itemize}
\end{theorem}

\begin{proof}
  We already know that the first eight conditions are equivalent.
  Since any norm-almost periodic measure is sup-almost periodic, the
  implication $\text{(xv)} \Longrightarrow \text{(xiv)}$ is obvious. Also, by picking any $0< a< 1$, the implication
  $\text{(xiv)} \Longrightarrow \text{(xiii)}$ is trivial, while  $\text{(xiii)} \Longrightarrow \text{(xi)}$ is immediate. Further,
  $\text{(x)} \Longleftrightarrow \text{(xi)} \Longleftrightarrow
  \text{(xii)} \Longrightarrow \text{(i)}$ follows from
  Theorem~\ref{T2}, and $\text{(ix)} \Longrightarrow \text{(i)}$
  follows from Corollary~\ref{123123123}.

  We now show that $\text{(i)} \Longrightarrow \text{(ix)}$.  Pick $0
  < \varepsilon <\frac{1}{2}$. Then, $\vL^\varepsilon$ is a model set.
  Let $\gamma$ be the autocorrelation of $\vL^\varepsilon$.  Let
  $\vG:= \vL^\varepsilon - \vL^\varepsilon$. Since $0 \in
  \vL^\varepsilon$, we have that $\vL^\varepsilon \subset \vG$.  By
  Proposition~\ref{S31}, for all $x \in G$ and $y \in
  \vG^{\varepsilon'}$, we have
\[
    \bigl| \widehat{\gamma}(\{ x+y\})-
    \widehat{\gamma}(\{x\}) \bigr| \,\leq\,
    \widehat{\gamma}(\{0\})\ts\varepsilon' \ts .
\]
Thus, for all $\varepsilon' <1$, we obtain
$\widehat{\gamma}(\{0\})\ts\varepsilon' <
\|(\widehat{\gamma})^{}_{\mathsf{pp}}\|^{}_\infty=
\widehat{\gamma}(\{0\})$ and
\[
   \vG^{\varepsilon'} \,\subset\,
   P_{\widehat{\gamma}(\{0\})\varepsilon'}^{}
   \bigl((\widehat{\gamma})^{}_{\mathsf{pp}}\bigr)\ts .
\]
Thus $(\widehat{\gamma})^{}_{\mathsf{pp}}$ is sup-almost periodic.
Also, by Lemma~\ref{DMS3}, we have
\[
    \vL^{\varepsilon \varepsilon} \,\subset\, \vG^{2\varepsilon} \ts .
\]
Using $2\varepsilon <1$, we get
\[
   \vL \,\subset\, \vL^{\varepsilon \varepsilon}
    \,\subset\, \vG^{2\varepsilon}
     \,\subset\, P^{}_{2\widehat{\gamma}(\{0\})\varepsilon}
    \bigl((\widehat{\gamma})^{}_{\mathsf{pp}}\bigr) \ts .
\]
Since $(\widehat{\gamma})^{}_{\mathsf{pp}}$ is sup-almost periodic, we
are done.

Alternately, one can prove $\text{(i)} \Longrightarrow \text{(ix)}$ by picking a cut and project scheme $(G \times H, \cL)$ and some precompact open set $U \subset H$ such that $\vL \subset \oplam(U)$, and then construct a continuous compactly supported function $g$ on $H$ such that, for some $0< a< \|g \|_\infty$ we have
\[
\|g-T^tg \|_\infty \,>\, a  \mbox{ for all } x \in U \ts .
\]
This can be achieved for example, by constructing $g$ to be continuous, positive definite, compactly supported and such that $|g(x)|$ is bounded from below on $U$.

Then, the measure $\omega_g$ satisfies the requirements of $\text{ix)}$.

To complete the proof, we show $\text{(i)} \Longrightarrow \text{(xv)}$, by constructing a measure which satisfies $\text{(xv)}$. Our measure will also satisfy the condition $\text{(xii)}$, but this implication is not needed.

Let $(G \times H, \cL)$ be a cut and project scheme and $W
\subset H$ be a compact set so that $\vL \subset \oplam(W)$.  Let $g
\in C_{\mathsf{c}}(H)$ be so that $g \geq 1^{}_{W}$. Let $\eta :=
\omega_g$.  Since $\supp(\eta)$ is a model set, and $\|g\|^{}_\infty <
\infty$, it is easy to see that $\eta$ is a regular measure. Moreover,
$\delta^{}_{\nts\vL} \leq \eta$.  Then, by Theorem~\ref{T2}, $\eta$ is
sup-almost periodic, and $\supp(\eta)$ is a model set, thus has finite
local complexity. Moreover,
\[
    \vL  \,\subset\,
     \bigl\{ x \in G \mid \eta(\{ x \}) > \tfrac{1}{2} \bigr\}
    \,\subset\, \supp(\eta) \ts .
\]

\end{proof}

\section{Period stable model sets}

When one studies cut and project schemes, two similar issues sometimes
occur.

The first issue is when the window has some periods. If $c^\star \in \pi^{}_{2}
(\cL)$ is such that $c^\star+W=W$, then $c$ is a period for
$\oplam(W)$.  The second issue is that, sometimes, the $\star$-map fails
to be one to one. In this case, $\ker(\star)$ is a periodic group for
$\oplam(W)$.  In both cases, we can see that $\oplam(W)$ has periods.

Given a Meyer set $\vL$, we are often able to construct a cut and
project scheme and to try to identify the window using the star map,
but if the candidate window has periods, or if $\ker(\star) \neq 0$,
we need additional conditions to obtain $\vL$ from the window in that
particular cut and project scheme: multiple translates of points must
simultaneously be either in $\vL$ or outside $\vL$.

One way of getting around this problem is by studying the larger
category of inter model sets. In this section we try a different approach to get
around this issue, namely we restrict our attention to a smaller class
of model sets.

\begin{definition}
  A point set $\vL$ is called \emph{period-stable} if, for all groups $H
  \leq G$ so that $H \subset \vL-\vL$, we have $H+ \vL=\vL$.
\end{definition}

We obtain a simple characterisation for period-stable
model sets.

\begin{theorem}\label{psms}
  Let\/ $\vL \subset G$ be relatively dense and period-stable. Then,
  the following statements are equivalent.
  \begin{itemize}\itemsep=1pt
  \item[(i)] $\vL$ is a model set.
  \item[(ii)] $\vL$ has finite local complexity, and there exists a
    positive sup-almost periodic measure\/ $\eta \neq 0$ so that\/ $\eta \leq
    \delta^{}_{\nts\vL}$.
  \item[(iii)] $\vL$ has finite local complexity, and there exists a
    positive norm-almost periodic measure\/ $\eta \neq 0$ so that\/ $\eta
    \leq \delta^{}_{\nts\vL}$.
  \item[(iv)] $\vL-\vL$ is uniformly discrete, and there exists a
    positive strongly almost periodic measure\/ $\eta \neq 0$ so that\/ $\eta
    \leq \delta^{}_{\nts\vL}$.
 \end{itemize}
\end{theorem}

\begin{proof}
  The equivalences $\text{(ii)} \Longleftrightarrow \text{(iii)}
  \Longleftrightarrow \text{(iv)}$ follow immediately from
  Theorem~\ref{T2}. It remains to be shown that $\text{(i)}
  \Longleftrightarrow \text{(ii)}$.

  We first prove the implication $\text{(i)} \Longrightarrow
  \text{(ii)}$. Let $(G \times H, \cL)$ be a cut and project
  scheme and let $W \subset H$ be precompact with $W^\circ \neq
  \varnothing$ and
\[
    \vL \, =\,  \oplam(W) \ts .
\]
Pick a non-zero $g \in C_{\mathsf{c}}(G)$ so that $0 \leq g \leq
1^{}_{W^\circ}$. Then, choosing $\eta=\omega_g$ works.

Now consider $\text{(ii)} \Longrightarrow \text{(i)}$.  Let $L:=
\langle \vL \rangle$. Then, the pair
$(P_{\!\varepsilon}^{\infty}(\eta), L)$ satisfies the axioms
(A1)--(A5), thus the Baake--Moody construction yields a cut and project
scheme $(G \times H, \cL)$.  Moreover, by
Lemma~\ref{capschemelemma}, there exists a $g \in C_{\mathsf{c}}(H)$
so that $\eta= \omega_g$.

Let $W_1:= \{ y \in H \mid g(y) > 0 \}$. Then $W_1$ is open and
non-empty.
Let
\[
   W \, := \, W_1 \cup \{ x^\star \mid x \in \vL \} \ts .
\]

Then, $W_1 \subset W$, and hence $W^\circ \neq \varnothing$. We now show
that $\overline{W}$ is compact and that $\vL= \oplam(W)$.

Since $\eta$ is sup-almost periodic and non zero, $\supp(\eta)$ is
relatively dense, and thus, by Lemma~\ref{NS1x}, there exists a finite
set $F$ so that
\[
   \vL \,\subset\, \supp(\eta) +F \ts .
\]
Since both $\vL, \supp(\eta) \subset L$, by picking a
minimal $F$ and eventually adding $0 \in L$, we can assume, without
loss of generality, that $0 \in F \subset L$.

For any subset $A \subset L$, we denote $A^\star:= \{ x^\star \mid x \in
A \}$. Then, it is easy to see that
\[
   W \,\subset\, W_1 + F^\star \ts .
\]
Since $g \in C_{\mathsf{c}}(G)$, we have that $W_1$ is precompact, and
hence $W$ is precompact. Also, we know that $\vL^\star \subset W$, and
thus the inclusion $\vL \subset \oplam(W)$ is obvious. We now show
that $\oplam(W) \subset \vL$.

Let $x \in \oplam(W)$. Then $x^\star \in W$. Hence, either $x^\star \in
W_1$ or $x^\star \in \vL^\star$.

Consider the first case, $x^\star \in W_1 \Longrightarrow
g(x^\star)>0$. Then,
\[
   \delta^{}_{\nts\vL}(\{ x \}) \,\geq\, \eta(\{x\})
   \,=\, \omega_g(\{x\}) \,=\,
   g(x^\star) \,>\, 0 \ts ,
\]
and hence $x \in \vL$.

For the second case, $x^\star \in \vL^\star$, we first show that
$\ker(\star) \subset \vL- \vL$. Let $z \in \ker(\star)$, then $z \in
L$ and $z^\star=0$.

Since $W_1$ is open and $L^\star$ is dense in $H$, there exists a $w
\in L$ so that $w^\star \in W_1$. Then, $g(w^\star)=g(w^\star+z^\star)
\neq 0$, and hence
\[
  \eta(\{w\}) \,=\, \eta(\{ w+z \}) \,\neq\, 0 \ts .
\]
Using $\eta \leq \delta^{}_{\nts\vL}$, we get that $w, w+z \in \vL$,
which shows that $z \in \vL- \vL$.  Thus, $\ker(\star) \subset \vL-
\vL$, and, since $\vL$ is period-stable, we have $\vL +
\ker(\star)= \vL$.

Now, since $x^\star \in \vL^\star$, we have $x^\star=y^\star$ for some
$y \in \vL$. If $x=y$, we are done. Otherwise, $(x-y)^\star=0$, thus
$x-y \in \ker(\star)$. Hence
\[
  x \, =\, x-y+y \, \in\,  \ker(\star)+\vL \, =\, \vL \ts ,
\]
which completes the proof.
\end{proof}

If $\vL$ is not period-stable, the following result can be proved
exactly like Theorem~\ref{psms}.

\begin{proposition}
  Let\/ $\vL \subset G$ be relatively dense. Then, the following
  statements are equivalent.
  \begin{itemize}\itemsep=1pt
  \item[(i)] There is a cut and project scheme\/ $(G\times H,
    \cL)$ and two sets\/ $U \subset W$ with\/ $\varnothing
    \neq U$ open and\/ $W$ compact so that\/ $\oplam(U) \subset
    \vL\subset \oplam(W)$.
  \item[(ii)] $\vL$ has finite local complexity, and there exists a
    positive sup-almost periodic measure\/ $\eta$ so that\/ $\eta \leq
    \delta^{}_{\nts\vL}$.
  \item[(iii)] $\vL$ has finite local complexity, and there exists a
    positive norm-almost periodic measure\/ $\eta$ so that\/ $\eta
    \leq \delta^{}_{\nts\vL}$.
  \item[(iv)] $\vL-\vL$ is uniformly discrete, and there exists a
    positive strongly almost periodic measure\/ $\eta$ so that\/ $\eta
    \leq \delta^{}_{\nts\vL}$.
 \end{itemize}
\end{proposition}

The first condition seems equivalent to $\vL$ being a inter model
set, but we are not sure if one can get the right window.

\section[Relatively dense sets of Bragg peaks]{Positive
measures with a relatively dense set of visible Bragg peaks}

Given a Meyer set $\vL$ with autocorrelation $\gamma$, we proved earlier that the set of $a$-visible Bragg peaks in the diffraction of $\vL$ is a Meyer set for all $0 < a <\widehat{\gamma}(\{0 \})$.

The goal of this section is to extend this result to any measure $\mu$ with positive autocorrelation measure and relatively dense sets of $a$-visible Bragg peaks.

In $\RR^d$, given a measure $\mu$ with positive autocorrelation $\gamma$, it was proven in \cite{S-LS1} that if the set of $a$-visible Bragg
peaks is relatively dense for some $\widehat{\gamma}(\{ 0 \})\left(\sqrt{3}-1 \right) < a < \widehat{\gamma}(\{ 0 \})$, then it must be a Meyer set.

We start by proving that if all the sets of $a$-visible Bragg
peaks are relatively dense, then they are model sets in the same cut and project scheme. In particular this is the case for the diffraction of a Meyer set.

Then, we prove a result similar to the one in \cite{S-LS1} about a single set of $a$-visible Bragg peaks, with $a$ large enough.

\begin{proposition}\label{FTPMRDSVBP}
  Let\/ $\mu$ be a measure and let\/ $\gamma$ be its autocorrelation. If\/
  $\gamma$ is positive and\/ $\widehat{\gamma}(\{ 0 \}) >0$, then the
  following statements are equivalent.
  \begin{itemize}\itemsep=1pt
  \item[(i)] For all\/ $0 < a < \widehat{\gamma}(\{ 0 \})$, the set\/ $I(a)$ of\/ $a$-visible Bragg peaks is relatively dense.
  \item[(ii)] For all\/ $0 < a < \widehat{\gamma}(\{ 0 \})$, the set\/ $I(a)$ of\/ $a$-visible Bragg peaks are model sets in the same cut and
    project scheme.
  \item[(iii)] There exists a cut and project scheme\/ $(\widehat{G}
    \times H_{\mathsf{fd}}, \widetilde{L_{\mathsf{fd}}})$ and a
    continuous function\/ $g^{}_{\mathsf{fd}} \in C_0 (H_{\mathsf{fd}})$
    so that\/ $(\widehat{\gamma})^{}_{\mathsf{pp}} =
    \omega_{g^{}_{\mathsf{fd}}}$.
\end{itemize}
\end{proposition}

\begin{proof}
The implication $\text{(ii)} \Longrightarrow \text{(i)}$ is obvious.

We now prove $\text{(i)} \Longrightarrow \text{(iii)}$.

Let $0 < \varepsilon <\sqrt{2}\,\widehat{\gamma}(\{ 0 \}) $, let $\chi \in I\bigl( \widehat{\gamma}(\{ 0 \})
-\frac{\varepsilon^2}{ 2 \widehat{\gamma}(\{ 0 \})}\bigr)$, and let
$\psi \in \widehat{G}$.

Then, by
Lemma~\ref{positive-definite-measure-function}, we have
\[
   \bigl| \widehat{\gamma}(\psi+\chi) -
          \widehat{\gamma}(\psi) \bigr|^2
   \, <\, 2\, \widehat{\gamma}(\{ 0 \}) \ts
    \bigl[ \widehat{\gamma}(\{ 0 \}) -
    \widehat{\gamma}(\{ \chi \})\bigr]
   \,<\, \varepsilon^{2} ,
\]
Therefore
\[
  I\biggl( \widehat{\gamma}(\{ 0 \}) -
   \myfrac{\varepsilon^2}{ 2\, \widehat{\gamma}(\{ 0 \})}\biggr)
   \,\subset\, P^{\infty}_{\!\varepsilon}
    \bigl((\widehat{\gamma})^{}_{\mathsf{pp}} \bigr) \ts .
\]
As $ I\biggl( \widehat{\gamma}(\{ 0 \}) -
   \myfrac{\varepsilon^2}{ 2\, \widehat{\gamma}(\{ 0 \})}\biggr) $ is relatively dense, so is $P^{\infty}_{\!\varepsilon} \bigl(
(\widehat{\gamma})^{}_{\mathsf{pp}} \bigr)$.

This proves that $(\widehat{\gamma})^{}_{\mathsf{pp}}$ is a discrete sup-almost periodic measure. The conclusion follows now from Theorem~\ref{T1}.

We next prove $\text{(iii)} \Longrightarrow
\text{(ii)}$. For this implication, we will simply denote
\[
g:= g^{}_{\mathsf{fd}} \ts .
\]

Since $\gamma$ is positive and positive definite, we have
$\widehat{\gamma}(\{ 0 \}) \geq \left| \widehat{\gamma}(\{ \chi \})
\right|=\widehat{\gamma}(\{ \chi \})$ for all $\chi \in \widehat{G}$.

As $0 \leq g(\chi^\star) \leq g(0)$ for all $(\chi, \chi^\star) \in
\widetilde{L_{\mathsf{fd}}}$, and since $\pi^{}_{2}\bigl(
\widetilde{L_{\mathsf{fd}}}\bigr)$ is dense in $H$, it follows
immediately from the continuity of $g$, that $g \geq 0$ and $\|g\|^{}_{\infty} =g(0) =
\widehat{\gamma}(\{ 0 \})$.

Let $0< a < \widehat{\gamma}(\{ 0 \}) =g(0)$.

Let $W_a := \{ y \in H \mid g(y) \geq a \}$.  Since $g \in C_0(H)$,
the set $W_a$ is compact. Moreover, since $a < g(0)$, the set
$W_a$ contains the non-empty open set $\{ y \in H \mid g(y) > a \}$.
Therefore $W_a$ is a compact set with non-empty interior.

As $(\widehat{\gamma})^{}_{\mathsf{pp}} =
    \omega_{g^{}_{\mathsf{fd}}}$, it follows from the definition of $\omega^{}_{g}$ that
    \[
     \bigl\{ \chi \in \widehat{G} \mid \widehat{\gamma}(\{ \chi \}) \neq 0 \bigr\} \, \subset \, \pi^{}_{1}\bigl(\widetilde{L_{\mathsf{fd}}}\bigr) \ts .
    \]
Therefore
\[
   I(a)\, = \,
   \bigl\{ \chi \in \widehat{G} \mid
         \widehat{\gamma}(\{ \chi \}) \geq a \bigr\}
   \,= \, \bigl\{ \chi  \mid (\chi, \chi^\star) \in
      \widetilde{L_{\mathsf{fd}}},\, g( \chi^\star) \geq a \bigr\}
    \,=\,  \oplam(W_a) \ts ,
\]
which completes the proof.
\end{proof}

Now, let us look at a single $I(a)$. Our goal is to prove that if $a$ is close enough to $\widehat{\gamma}(\{ 0 \})$, then $I(a)$ is a Meyer set.
The key to the proof is Lemma~\ref{LI(ab)} below, which shows that if $\chi$ and $\psi$ are Bragg peaks of high intensity, then $\chi+\psi$ is also a Bragg peak whose intensity is bounded from below by a quantity which only depends on the intensities of the Bragg peaks at $\chi$ and $\psi$.

\begin{lemma}\label{LI(ab)}
  Let\/ $\mu$ be a measure with positive autororrelation\/ $\gamma$. Let \/ $0< a,b < \widehat{\gamma}(\{ 0 \})$. Then
\[
I(a) \pm I(b) \, \subset \, I(b- \sqrt{ 2\, \widehat{\gamma}(\{ 0 \}) \ts
    \bigl[ \widehat{\gamma}(\{ 0 \}) -
    a \bigr]} ) \ts .
\]
\end{lemma}
\begin{proof}
Let $\chi \in I(a)$ and $\psi \in I(b)$. Then, as $\widehat{\gamma}$ is positive and positive definite, by Lemma~\ref{positive-definite-measure-function}, we have
\[
   \bigl| \widehat{\gamma}(\psi \pm \chi) -
          \widehat{\gamma}(\psi) \bigr|^2
   \, \leq\, 2\, \widehat{\gamma}(\{ 0 \}) \ts
    \bigl[ \widehat{\gamma}(\{ 0 \}) -
    \widehat{\gamma}(\{ \chi \})\bigr]
\]
Therefore
\[
\widehat{\gamma}(\psi \pm \chi) \, \geq \, b- \sqrt{ 2\, \widehat{\gamma}(\{ 0 \}) \ts
    \bigl[ \widehat{\gamma}(\{ 0 \}) -
    \widehat{\gamma}(\{ \chi \})\bigr]}  \, \geq \, b- \sqrt{ 2\, \widehat{\gamma}(\{ 0 \}) \ts
    \bigl[ \widehat{\gamma}(\{ 0 \}) -
    a \bigr]}
\]
\end{proof}

A simple induction on $n$ combined with Lemma~\ref{LI(ab)} yields Corollary~\ref{CI(ab)} below.

\begin{corollary}\label{CI(ab)}
  Let\/ $\mu$ be a measure with positive autororrelation\/ $\gamma$ be its autocorrelation. Let \/ $0< a < \widehat{\gamma}(\{ 0 \})$. Then
\[
\underbrace{I(a) \pm I(a) \pm ... \pm I(a)}_{n \mbox{ times}} \, \subset \, I \biggl( a- \, (n-1) \sqrt{ 2\, \widehat{\gamma}(\{ 0 \} ) \,
    \bigl[ \widehat{\gamma}(\{ 0 \}) -
    a \bigr]} \biggr) \ts .
\]
\end{corollary}

We are now ready to prove that if $I(a)$ is relatively dense for $a$ close enough to $\widehat{\gamma}(\{ 0 \})$, then $I(a)$ is a Meyer set.
\begin{proposition}
  Let\/ $\mu$ be a measure with positive autororrelation\/ $\gamma$. Let \/ $ \widehat{\gamma}(\{ 0 \})\bigl( \sqrt{24}-4 \bigr) < a < \widehat{\gamma}(\{ 0 \})$. If the set\/ $I(a)$ of\/ $a$-visible Bragg peaks is relatively dense then\/ $I(a)$ is a Meyer set.
\end{proposition}
\begin{proof}
We start by observing that
\[
\begin{array}{lr}
\begin{split}
a \, &- \, \ts 2\sqrt{ 2\, \widehat{\gamma}(\{ 0 \}) \ts
    \bigl[ \widehat{\gamma}(\{ 0 \}) -
    a \bigr]} \,  > \, 0 \ts &\Longleftrightarrow  \ts \\
a^2 \, & > \,  8\, \widehat{\gamma}(\{ 0 \}) \ts
    \bigl[ \widehat{\gamma}(\{ 0 \}) -
    a \bigr] \ts &\Longleftrightarrow \\
a^2 \, & + \,  8\, \widehat{\gamma}(\{ 0 \}) a \, + \, 16  \bigl( \widehat{\gamma}(\{ 0 \}) \bigr)^2 \ts
    > 24 \bigl( \widehat{\gamma}(\{ 0 \}) \bigr)^2 \ts &\Longleftrightarrow\\
 \bigl[a \bigr.\, &+\, \bigl. 4\widehat{\gamma}(\{ 0 \}) \bigr]^2 \ts
    > 24 \bigl( \widehat{\gamma}(\{ 0 \}) \bigr)^2 \ts &\Longleftrightarrow \\
 a \,&>\,
    \widehat{\gamma}(\{ 0 \})\bigl( \sqrt{24}-4 \bigr) \ts & \,
    \end{split}
    \end{array}
    \]

Therefore, our hypothesis implies that $a- \ts 2\sqrt{ 2\, \widehat{\gamma}(\{ 0 \}) \ts
    \bigl[ \widehat{\gamma}(\{ 0 \}) -
    a \bigr]} >0$.

As $\widehat{\gamma}$ is a translation bounded measure, for all $b>0$ the set $I(b)$ must be locally finite. In particular, $I\biggl(a- \ts 2\sqrt{ 2\, \widehat{\gamma}(\{ 0 \}) \ts
    \bigl[ \widehat{\gamma}(\{ 0 \}) -
    a \bigr]} \biggr)$ is locally finite.

By Corollary~\ref{CI(ab)} we have
\[
I(a) \pm I(a) \pm  I(a) \, \subset \, I\biggl(a- \ts 2\sqrt{ 2\, \widehat{\gamma}(\{ 0 \}) \ts
    \bigl[ \widehat{\gamma}(\{ 0 \}) -
    a \bigr]} \biggr) \ts .
\]

Therefore $I(a)-I(a)-I(a)$ is a locally finite set, which proves our claim.
\end{proof}

We complete the section by showing that any point set having
relatively dense sets of visible Bragg peaks shares many of them with
a regular model set.

\begin{proposition}
  Let\/ $\vL$ be a Delone set, and let\/ $\gamma$ be its
  autocorrelation. If, for all\/ $0 < a < \widehat{\gamma}(\{0\})$, the
  set\/ $I(a)$ of visible Bragg peaks is relatively dense, then there
  exists a regular model set\/ $\vL'$ so that
  \begin{itemize}\itemsep=1pt
  \item[(i)] $\vL$ and\/ $\vL'$ share a relative dense set of Bragg
    peaks;
  \item[(ii)] the sets of visible Bragg peaks of\/ $\vL$ and\/ $\vL'$ are
    equivalent by finite translations, that is, each set can be covered
    by finitely many translates of the other set.
 \end{itemize}
\end{proposition}

\begin{proof}
  By Proposition~\ref{FTPMRDSVBP}, the set of
  $a$-visible Bragg peaks $I_\vL(a)$ of $\vL$ is a model set for all
  $0 < a < \widehat{\gamma}(\{0\})$. Pick some $0 < \varepsilon <
  \frac{\sqrt{3}}{4}$ and some $0< a < \widehat{\gamma}(\{0\})$ so that
  $\widehat{\gamma}(\{0\})-a < \varepsilon$. For the rest of the
  proof, $a$ and $\varepsilon$ are fixed.

  Let $\vL'' := I_\vL(a)^\varepsilon$. Then $\vL''$ is a model
  set. Let $\vL'$ be any regular model set contained in $\vL''$.

  We prove that $\vL'$ has the desired properties. We start by proving (i).

 Let $\eta$
  be any autocorrelation of $\vL'$.  As $\supp(\eta) \subset
  \vL'-\vL'$, for all $\varepsilon' >0$ and all $\psi \in
  (\vL'-\vL')^{\varepsilon'}$ we have, by Theorem~\ref{S31},
\[
   \bigl| \widehat{\eta}(\{ \psi\})-\widehat{\eta}(\{ 0 \}) \bigr|
    \,\leq\, \widehat{\eta}(\{0\})\,\varepsilon' \ts .
\]
Thus, for all $0 < \varepsilon' <1$, we obtain
\[
    \widehat{\eta}(\{ \psi\}) \,\geq\,
    \widehat{\eta}(\{0\})\, (1-\varepsilon') \,,
\]
This shows that $(\vL'-\vL')^{\varepsilon'} \subset I^{}_{\vL'}
\bigl(\widehat{\eta}(\{0\})\ts (1-\varepsilon')\bigr)$.

Now, by Lemma~\ref{DMS3}, we have $\vL'^{\varepsilon'} \subset
(\vL'-\vL')^{2\varepsilon'}$, and therefore
\[
   \vL'^{\varepsilon} \,\subset\,
    I^{}_{\vL'} \bigl(\widehat{\eta}(\{0\})\ts (1-2\varepsilon)\bigr) \ts .
\]
Also, since $\vL' \subset \vL''= I^{}_{\vL}(a)^\varepsilon$, we have
\[
   I^{}_{\vL}(a) \,\subset\,
   I^{}_{\vL}(a)^{\varepsilon \varepsilon} \,\subset\,
   \vL'^{\varepsilon} \,\subset\,
   I^{}_{\vL'} \bigl(\widehat{\eta}(\{0\})\ts (1-2\varepsilon)\bigr) \ts .
\]
This proves that every $a$-visible Bragg peak of $\vL$ is a
$\bigl(\widehat{\eta}(\{0\})\ts (1-2\varepsilon)\bigr)$-visible Bragg
peak of $\vL'$. The claim follows from the relatively denseness of $I^{}_{\vL}(a)$.

We now prove (ii). We showed in (i) that
\[
   I^{}_{\vL}(a) \,\subset\, I^{}_{\vL'}
  \bigl(\widehat{\eta}(\{0\})\ts (1-2\varepsilon)\bigr) \ts .
\]
By Proposition~\ref{FTPMRDSVBP}, we know that, for all $0< b <
\widehat{\gamma}(\{ 0\})$, the sets $I^{}_{\vL}(b)$ are model sets in
the same cut and project scheme.

Therefore, for all $0< b$ , $c < \widehat{\gamma}(\{ 0\})$, there exists a finite set $F^{}_{b,c}$ so
that
\[
   I^{}_{\vL}(b) \,\subset\, I^{}_{\vL}(c) +F^{}_{b,c}
   \quad \text{and}\quad
   I^{}_{\vL}(c)  \,\subset\, I^{}_{\vL}(b) +F^{}_{b,c} \ts .
\]

Since $\vL'$ is a regular model set, it is a pure point diffractive Meyer set. Hence, by Corollary~\ref{CT1x}, for all $0< b <
\widehat{\eta}(\{ 0\})$, the set $I^{}_{\vL'}(b)$ is relatively dense.
Therefore, it follows from Proposition~\ref{FTPMRDSVBP} that, for all $0< b <
\widehat{\eta}(\{ 0\})$, the sets $I^{}_{\vL'}(b)$ are model sets in
the same cut and project scheme. Hence, for all $0< b$ , $c <
\widehat{\eta}(\{ 0\})$, there exists a finite set $F'_{b,c}$ so that
\[
   I^{}_{\vL'}(b) \,\subset\,
    I^{}_{\vL'}(c) +F'_{b,c} \quad\text{and}\quad
    I^{}_{\vL'}(c) \subset I^{}_{\vL'}(b) +F'_{b,c} \ts .
\]

Moreover, we also know that $ I^{}_{\vL}(a) \,\subset\, I^{}_{\vL'}
  \bigl(\widehat{\eta}(\{0\})\ts (1-2\varepsilon)\bigr)$ ; that $I^{}_{\vL}(a)$ is relatively dense and that $I^{}_{\vL'}
\bigl(\widehat{\eta}(\{0\})\ts (1-2\varepsilon)\bigr)$ has finite local
complexity. Therefore, there exists a finite set $F$ so that
\[
  I^{}_{\vL'} \bigl(\widehat{\eta}(\{0\}) (1-2\varepsilon)\bigr)
  \,\subset\, I(a)+F \ts .
\]
Pick any $0< b < \widehat{\gamma}(\{ 0\})$ and $0< c <
\widehat{\eta}(\{ 0\})$. Setting $F^{}_1:=F^{}_{a,b} $ and
$F^{}_2:=F'_{c,\widehat{\eta}(\{0\}) (1-2\varepsilon)} $, we obtain
\[
    I^{}_{\vL}(b) \,\subset\,
    I^{}_{\vL}(a)+F^{}_{1} \,\subset\,
    I^{}_{\vL'} \bigl(\widehat{\eta}(\{0\})\ts (1-2\varepsilon)\bigr)
    + F^{}_{1} \,\subset\,
    I^{}_{\vL'}(c) + F^{}_{1} + F^{}_{2} \ts .
\]
and
\[
  I^{}_{\vL'}(c) \,\subset\,
  I_{\vL'} \bigl(\widehat{\eta}(\{0\})\ts (1-2\varepsilon)\bigr)
  + F^{}_{2} \,\subset\,
   I^{}_\vL(a) + F + F^{}_{2} \,\subset\,
   I^{}_{\vL}(b) + F^{}_{1} + F + F^{}_{2} \ts ,
\]
which completes the proof.
\end{proof}

\section[Pure point spectra of Meyer sets]{Some more comments
on the Bragg spectra of Meyer sets}\label{sectfewnotes}

Let $\omega$ be any weighted Dirac comb supported on some Meyer set $\vL$ and let $\gamma$ be any autocorrelation of
$\omega$. Then, by \cite{S-MoSt} there exists a decomposition
$\gamma= \gamma^{}_{\mathsf{s}} + \gamma^{}_{0} $ so that both
$\gamma^{}_{\mathsf{s}}$ and $\gamma^{}_{0}$ are Fourier
transformable and
\[
   \widehat{\gamma^{}_{\mathsf{s}}} \,=\,
   \bigl(\widehat{\gamma}\bigr)_{\mathsf{pp}}
   \quad\text{and}\quad
   \widehat{\gamma^{}_{0}} \,=\,
   \bigl(\widehat{\gamma}\bigr)_{\mathsf{c}} \ts.
\]
It follows from Theorem~\ref{WAPModdec} that, for any regular model set
with closed window $\vG$ containing $\vL$, we have
\[
   \supp(\gamma^{}_{\mathsf{s}}) \,\subset\, \Delta
   \quad\text{and}\quad
   \supp(\gamma^{}_{0}) \,\subset\, \Delta \ts ,
\]
where $\Delta = \vG - \vG$.

Since $\gamma^{}_{\mathsf{s}}$ is strongly almost periodic and
$\Delta$ is a Meyer set, by Theorem~\ref{T2} we obtain a cut and
project scheme $(G \times H, \cL)$ and a function $g \in
C_{\mathsf{c}}(H)$ so that $\gamma^{}_{\mathsf{s}} = \omega^{}_{g}$.
Moreover, in Corollary~\ref{cms33} we showed that there exists a real
number $C > 0$ so that, for all $\varepsilon >0$,
\[
  \Delta^{\varepsilon /C }_{} \,\subset\,
  P^\infty \bigl((\widehat{\gamma})^{}_{\mathsf{pp}}\bigr) \ts .
\]
Hence the pure point spectrum
$(\widehat{\gamma})^{}_{\mathsf{pp}}$ is a regular discrete
sup-almost periodic measure. We thus established the following result.

\begin{proposition}\label{MS21}
  Let\/ $\omega$ be any weighted Dirac comb with Meyer set support and let\/ $\gamma$ be an autocorrelation
  of\/ $\vL$. Then, there exists a cut and project scheme\/ $(G \times
  H, \cL)$ and a function\/ $g \in C_{\mathsf{c}}(H)$ so
  that\/ $\gamma^{}_{\mathsf{s}} = \omega^{}_{g}$. Furthermore,
  there exists a cut and project scheme\/ $(\widehat{G} \times
  H_{\mathsf{fd}}, \widetilde{L_{\mathsf{fd}}})$ and a continuous
  function\/ $g^{}_{\mathsf{fd}} \in C_0 (H)$ which gives the pure
  point spectra of $\vL$, that is
  $\bigl(\widehat{\gamma}\bigr)_{\mathsf{pp}} =
  \omega^{}_{g^{}_{\mathsf{fd}}}$.\qed
\end{proposition}

\begin{remark}
  Let $\vG$ be a regular model set with autocorrelation $\eta$. Let $( G \times H, \cL)$
  be the cut and project scheme and let $W$ be the window which
  produces it. Let $(\widehat{G} \times \widehat{H}, Z)$ be the dual
  cut and project scheme. Then $\eta^{}_{\mathsf{s}}= \eta$, and
  we have \cite{S-BL}
\[
   \eta (\{ x \})\, = \,\begin{cases}
               1^{}_{W}* \widetilde{\ts 1^{}_W\ts}\nts (x^\star)\ts , &
               \mbox{ if $x \in \pi^{}_1(\cL)$\ts ,} \\
              0\ts , & \mbox{ otherwise,}
               \end{cases}
\]
and
\[
   \bigl(\widehat{\eta}\bigr)(\{ \chi \})\, =\,
   \begin{cases}
   \widehat{1^{}_{W}* \widetilde{\ts 1^{}_W\ts}}\nts(\chi^\star)\ts , &
   \mbox{ if $\chi \in \pi^{}_{1}(Z)$}\ts , \\
   0\ts , & \mbox{ otherwise.}
   \end{cases}
\]
Note that $1^{}_{W}* \widetilde{\ts 1^{}_W\ts }\nts \in C_{\mathsf{c}}(H)$ and
$\widehat{ 1^{}_{W}* \widetilde{\ts 1^{}_W\ts}\nts} \in
C_0(\widehat{H})$.

In Prop.~\ref{MS21} we showed that something similar holds for the pure
point part of the diffraction of Meyer sets. Given a Meyer set $\vL$ with autocorrelation $\gamma$,
there exist two cut and project schemes $(G \times H, \cL)$
and $(\widehat{G} \times H_{\mathsf{fd}},
\widetilde{L_{\mathsf{fd}}})$ and two functions $g \in C_{\mathsf{c}}(H)$ and
$g^{}_{\mathsf{fd}} \in C_0(H_{\mathsf{fd}})$ so that
$\gamma^{}_{\mathsf{s}}$ and
$\widehat{\gamma^{}_{\mathsf{s}}}=(\widehat{\gamma})^{}_{\mathsf{pp}}$ are
discrete measures satisfying the relations
\[
   \gamma^{}_{\mathsf{s}} (\{ x \}) \, =\,
   \begin{cases}
     g(x^\star)\ts , & \mbox{ if $x \in \pi^{}_1(\cL)$}\ts , \\
     0\ts , & \mbox{ otherwise,}
   \end{cases}
\]
and
\[
   \bigl(\widehat{\gamma}\bigr)(\{ \chi \}) \,=\,
   \begin{cases}
    g^{}_{\mathsf{fd}}( \chi^\star)\ts , &
    \mbox{ if $\chi \in \pi^{}_{1}(\widetilde{L_{\mathsf{fd}}})$}\ts , \\
    0\ts , & \mbox{ otherwise.}
    \end{cases}
\]
\exend
\end{remark}

Later in this section, we will see that there is a connection between
the two cut and project schemes $(G \times H, \cL)$ and
$(\widehat{G} \times H_{\mathsf{fd}},
\widetilde{L_{\mathsf{fd}}})$. More precisely, in
Theorem~\ref{dualdiffraction}, we will show that we can get the
results from Proposition~\ref{MS21} with $(\widehat{G} \times
H_{\mathsf{fd}}, \widetilde{L_{\mathsf{fd}}})$ as the cut and
project scheme dual to $(G \times H, \cL)$, and that $g^{}_{\mathsf{fd}}=\widehat{g}$.  But first
we will see some nice consequences of Proposition~\ref{MS21}.

\begin{corollary}
  Let\/ $\vL$ be a Meyer set and\/ $\gamma$ be an autocorrelation of\/
  $\vL$.
  \begin{itemize}\itemsep=1pt
  \item[(i)] For all\/ $0< a < \widehat{\gamma}(\{0 \})$, the sets\/
    $I(a)$ of\/ $a$-visible Bragg peaks are model sets in the same cut
    and project scheme.
  \item[(ii)] The set\/ $\supp(\gamma^{}_{\mathsf{s}})$ is a model set.\qed
 \end{itemize}
\end{corollary}

If $\vL$ is a repetitive Meyer set with pure point diffraction, we
have that $\supp(\gamma)=\supp(\gamma^{}_{\mathsf{s}})=\vL- \vL$, establishing the
following result.

\begin{corollary}
  Let\/ $\vL$ be a repetitive Meyer set with pure point
  diffraction. Then\/ $\vL - \vL$ is a model set.\qed
\end{corollary}

In the remainder of the section, we show that for combs $\omega_g$ with $g$ admissible, if $\omega^{}_{g}$ is Fourier transformable, then its
Fourier transform is $\omega^{}_{\widehat{g}}$. We first need to recall
some results by Lenz and Richard \cite{S-LR}.

\begin{definition}
  Let $(G \times H, \cL)$ be a cut and project scheme.  A
  function $g\!:\, H \longrightarrow \CC$ is called \emph{admissible}
  if it is measurable, locally bounded and, for all $\varepsilon >0$
  and $\phi \in C_{\mathsf{c}}(G)$, there exists a compact set $K \subset H$ so
  that, for all $(s,t) \in G \times H$, we have
\[
    \sum_{(t,h) \in \cL}
    \bigl| \phi(t+s)\, g(h+k) \bigr|\, \bigl(1-1^{}_K(h+k)\bigr)
    \,\leq\, \varepsilon \ts.
\]
\end{definition}

It is obvious that any $g \in C_{\mathsf{c}}(H)$ is admissible.

\begin{theorem}\cite{S-LR}\label{LR1}
  Let\/ $(G \times H, \cL)$ be a cut and project scheme, and
  let\/ $g\!:\, H \longrightarrow \CC$ be admissible. Let further\/
  $(\widehat{G} \times \widehat{H}, Z)$ be the dual cut and project
  scheme.
  \begin{itemize}\itemsep=1pt
  \item[(i)] $\omega^{}_{g}$ is strongly almost periodic.
  \item[(ii)] Let\/ $(\chi, \chi^\star) \in Z$. Then, the limit
\[
    c^{}_{\chi} \, = \, \lim_{n \to \infty}
    \frac{ \sum_{(s,s^\star) \in \cL} \chi(s)\, g(s^\star)}
        { \theta_G( A_n)}
\]
exists and satisfies
\[
  c^{}_{\chi} \, = \, \int_H g(h)\, \overline{\chi^\star} (h) \dd h \ts .
\]
\item[(iii)] Let\/ $\gamma$ be the autocorrelation of\/
  $\omega^{}_{g}$. Then $\widehat{\gamma}$ is discrete, and satisfies
\[
   \widehat{\gamma}\, =
    \sum_{(\chi, \chi^\star) \in Z}
    \left| c^{}_{\chi} \right|^2  \delta_{\chi} \ts .
\]
\end{itemize}
\end{theorem}

An immediate consequence of Theorem~\ref{LR21} is that an Meyer set has exactly the same set of Bragg peaks as a weighted model set.

\begin{corollary}\label{C1.19}
  Let\/ $\vL$ be a Meyer set and\/ $\gamma$ be an autocorrelation
  of\/ $\vL$.  Then, there exists a cut and project scheme\/ $(G
  \times H, \cL)$, a regular model set\/ $\oplam(W)$
  and a function\/ $g \in C_{\mathsf{c}}(H)$ so that, for the weight\/
  $\omega$ on the model set\/ $\oplam(W)$ defined by
\[
   \omega \, := \sum_{x \in \oplam(W)} g(x^\star)\, \delta_x \ts ,
\]
with autocorrelation $\eta$ we have
\[
   \widehat{\eta} (\{ \chi \})
   \, =\, \Bigl[ \widehat{\gamma} (\{ \chi \}) \Bigr]^2 \ts .
\]
In particular, $\omega$ and\/ $\vL$ have the same set of Bragg peaks.
\end{corollary}

\begin{proof}
  We know that $\gamma^{}_{\mathsf{s}}$ is a continuous weighted model
  comb supported on a Meyer set. Thus, by Theorem~\ref{T2}, there
  exists a cut and project scheme $(G \times H, \cL)$ and a
  function $g \in C_{\mathsf{c}}(H)$ so that $ \gamma^{}_{\mathsf{s}} = \omega^{}_{g}$.

  Let $W$ be any compact set in $H$ so that $\supp(g) \subset
  W^\circ$. Then,
\[
   \omega^{}_{g} \,= \, \gamma^{}_{\mathsf{s}} \, =
   \sum_{x \in \oplam(W)} \! g(x^\star)\, \delta_x \ts .
\]
As $g \in C_{\mathsf{c}}(H)$, it follows that $g$ is admissible. Therefore, by Theorem~\ref{LR21} (iii) and \cite{S-MoSt} we have
\[
  \widehat{\eta} (\{ \chi \}) \, = \,
  \bigl[ \widehat{\omega^{}_{g}} (\{ \chi \}) \bigr]^2 \, = \,
  \Bigl[ \widehat{\gamma} (\{ \chi \}) \Bigr]^2 ,
\]
which completes the proof.
\end{proof}

Another immediate consequence of Theorem~\ref{LR21} is the following result.

\begin{proposition}\label{LR21}
  Let\/ $(G \times H, \cL)$ be a cut and project scheme, and
  let\/ $g \in L^1(H)$ be admissible. Let further\/ $(\widehat{G}
  \times \widehat{H}, Z)$ be the dual cut and project scheme. If\/
  $\omega^{}_{g}$ is Fourier transformable, then the following
  statements hold.
  \begin{itemize}\itemsep=1pt
  \item[(i)] $\widehat{\omega^{}_{g}} = \omega^{}_{\widehat{g}}$.
  \item[(ii)] $\widehat{\omega^{}_{g}}$ is sup-almost periodic.
  \item[(iii)] $\omega^{}_{g}$ is positive definite if and only if\/ $g$ is
    positive definite.
\end{itemize}
\end{proposition}

\begin{proof}
  We start by proving (i). By Theorem~\ref{LR1}, $\omega^{}_{g}$ is
  strongly almost periodic, thus $\widehat{\omega^{}_{g}}$ is
  discrete \cite{S-MoSt}.  Hence, by \cite{S-ARMA} or \cite{S-MoSt}, for all $\chi \in \widehat{G}$ we
  have
\[
   \widehat{\omega^{}_{g}}(\chi) \, =\,  \lim_{n \to \infty}
    \frac{ \sum_{(s,s^\star) \in \cL} \chi(s) \,
    \omega^{}_{g}(s)}{\theta_G( A_n)}\ts .
\]
Therefore, for all $(\chi, \chi^\star) \in Z$, we obtain
\[
  \widehat{\omega^{}_{g}}(\chi) \, =\,
   \int_H g(h)\, \overline{\chi^\star (h)} \dd h
  \, =\,  \widehat{g}(\chi^\star) \ts .
\]
To complete the proof of (i), we need to show that, for all $\chi
\notin \pi^{}_{1}(Z)$, we have $\widehat{\omega^{}_{g}}(\chi)=0$.

Let $\gamma$ be the autocorrelation of $\omega^{}_{g}$. Since
$\omega^{}_{g}$ is strongly almost periodic, the limit
\[
   c^{}_{\chi} \, =\,  \lim_{n \to \infty}
   \frac{ \sum_{(s,s^\star) \in \cL} \chi(s)\,
         \omega^{}_{g}(s)}{ \theta_G( B_n)}
\]
exists uniformly for every van Hove sequence $\{ B_n \}$ and for all $\chi \in \widehat{G}$. Then, by
\cite{S-LR,S-HOF}, we have
\[
  \widehat{\gamma}(\{ \chi \}) \, = \, \left| c^{}_{\chi} \right|^2 \, = \,  \widehat{\omega^{}_{g}}(\chi) \ts .
\]
Now by Theorem~\ref{LR1}(iii) we have
$\widehat{\gamma}(\chi)=0$ for all $\chi \notin \pi^{}_{1}(Z)$, and
thus $\left| \widehat{\omega^{}_{g}}(\chi)\right|^2=0$.

To prove (ii), note that $\omega^{}_{\widehat{g}}$ is the Fourier
transform of $\widehat{\omega^{}_{g}}$, and thus translation
bounded. Also, $\widehat{g} \in C_0(\widehat{H})$, and thus, by
Theorem~\ref{T1} is sup-almost periodic.

Finally, we show (iii).  We have that $\omega^{}_{g}$ is positive
definite if and only if $\widehat{\omega^{}_{g}}=\omega_{\widehat{g}}$
is positive. Also, $g$ is positive definite if and only if
$\widehat{g}$ is positive.

Since $\widehat{g}$ is continuous and $\pi^{}_{2}(Z)$ is dense in $H$,
we get that $\widehat{g}$ is positive if and only if $\omega^{}_{\widehat{g}}$
is positive, which completes the proof.
\end{proof}

We now are able to prove that the results in Proposition~\ref{MS21}
can be obtained in a dual pair of cut and project schemes.

\begin{theorem}\label{dualdiffraction}
  Let\/ $\vL \subset G$ be a Meyer set, and let\/ $\gamma$ be an
  autocorrelation of\/ $\vL$. Then, there exists a cut and project scheme\/
  $(G \times H, \cL)$ and a\/ $g \in C_{\mathsf{c}}(H)$ so that
\[
   \gamma^{}_{\mathsf{s}} \,=\,  \omega^{}_{g} \quad\text{and}\quad
  (\widehat{\gamma})_{\mathsf{pp}} \,=\, \omega^{}_{\widehat{g}} \ts .
\]
\end{theorem}

\begin{proof}
  The existence of the cut and project scheme and of the function $g$
  with $\gamma^{}_{\mathsf{s}}= \omega^{}_{g}$ follow from Proposition~\ref{MS21}.

  Since $g \in C_{\mathsf{c}}(H)$, $g$ is in $L^1(H)$ and admissible.  Thus, by
  Proposition~\ref{LR21}, we have in the dual cut and project scheme
  $(\widehat{\gamma})_{\mathsf{pp}} = \omega^{}_{\widehat{g}}$.
\end{proof}

We complete the section by introducing a recent result which complements Proposition~\ref{LR21}.

\begin{theorem}\label{RS11} \cite{S-RS}
  Let\/ $(G \times H, \cL)$ be a cut and project scheme, and
  let\/ $g \in  C_{\mathsf{c}}(H)$ be such that\/ $\widehat{g} \in L^1(\widehat{H})$. Let further\/ $(\widehat{G}
  \times \widehat{H}, Z)$ be the dual cut and project scheme. Then
  $\omega^{}_{g}$ is Fourier transformable and
\[
\widehat{\omega^{}_{g}} = \omega^{}_{\widehat{g}} \ts .
\]
\end{theorem}

Note that in Theorem~\ref{RS11} we get the Fourier transformability for free, as long as $g \in  C_{\mathsf{c}}(H)$ and $\widehat{g} \in L^1(\widehat{H})$.
Note that in general it is unclear which admissible functions $g$ on $H$ define Fourier Transformable measures $\omega_g$.

\section{Dense norm-almost periodic combs}

\begin{definition}
  Let $(G \times H, \cL)$ be a cut and project scheme, and
  let $g \in C_{\mathsf{u}}(H)$. $g$ is called \emph{concentrated} for $(G \times
  H, \cL)$ if, for all $\varepsilon >0$, there exists a
  compact set $W \subset H$ so that
\[
    \| \ts\omega^{}_g -\omega^{}_{g^{}_W}\nts \|^{}_{K} \, <\, \varepsilon \ts ,
\]
where $g^{}_{W} = g\ts 1^{}_{W}$.
\end{definition}

The concentrated functions are important because of the following
theorem.

\begin{theorem}\cite{S-BLRS}\label{BLRS1}
  If\/ $g$ is concentrated for\/ $(G \times H, \cL)$ and
  uniformly continuous, then $\omega^{}_{g}$ is norm-almost periodic.
\end{theorem}

We first give a more general definition for \emph{concentrated}, which
does not depend on the cut and project setting and coincides with the
original one when both make sense. For this, we first prove the
following lemma.

\begin{lemma}\label{BLRS2}
  Let\/ $(G \times H, \cL)$ be a cut and project scheme, and
  let\/ $g \in C_{\mathsf{u}}(H)$. Then, the following statements are equivalent.
  \begin{itemize}\itemsep=1pt
  \item[(i)] The function\/ $g$ is concentrated.
  \item[(ii)] For each\/ $\varepsilon >0$, there exists a set\/ $\vL
    \subset G$ with finite local complexity so that
\[
   \Bigl\| \, \omega^{}_{g} -{\omega^{}_{g}\big|}^{}_{\!\vL}\ts \Bigr\|^{}_{K}
   \,< \,\varepsilon \ts .
\]
\end{itemize}
\end{lemma}

\begin{proof}
  The implication $\text{(i)} \Longrightarrow \text{(ii)}$ is clear,
  since, for all compact sets $W \subset H$, the set $\oplam(W)$ has
  finite local complexity.

  Consider now $\text{(ii)} \Longrightarrow \text{(i})$.  Let
  $\varepsilon >0$. Note that, once the concentration condition holds
  for some value of $\varepsilon$, it holds for all larger values,
  thus, without loss of generality, we can assume that $\varepsilon <
  \| g \|_\infty$.

  Let $U := \{ x \in H \mid \left| g(x) \right| > \varepsilon
  \}$. Then, $U$ is a non empty open subset of $H$. Let $V$ be a
  precompact open subset of $U$.  It is easy to see that $\oplam(U)
  \subset \vL$ and hence $\oplam(V) \subset \vL$.

  Since $\oplam(V)$ is a model set, it is relatively dense. Hence, by
  Lemma~\ref{NS1x}, there exists a finite set $F$ so that $\vL \subset
  \oplam(V) +F$. As both $\vL$ and $\oplam(V)$ are subsets of $\pi^{}_{1}(\cL)$, by choosing a minimal such
  $F$, we have $F \subset \pi^{}_{1}(\cL)$. Let
\[
   W \, :=\, \overline {V + F^\star} \ts .
\]
Then, $W$ is compact, and since $\vL \subset \oplam(W)$, we have
\[
   \Bigl\| \,\omega^{}_{g} - \omega^{}_{g^{}_W} \Bigr\|^{}_{K} \,\leq\,
   \Bigl\| \, \omega^{}_{g} -\omega^{}_{g}\big|^{}_{\!\vL} \Bigr\|^{}_{K}
    \, <\, \varepsilon \ts ,
\]
completing the argument.
\end{proof}

\begin{definition}
  A weighted comb $\omega := \sum_{x \in \vG} \omega(x)\ts \delta_x$ is
  called \emph{concentrated} if, for each $\varepsilon >0$, there exists
  a set $\vL \subset G$ with finite local complexity so that
\[
  \bigl\| \, \omega -\omega|^{}_{\!\vL} \bigr\|^{}_{K}
    \, <\, \varepsilon \ts .
\]
\end{definition}

As we see in the following lemma, any concentrated comb with finite
sup-norm is automatically translation bounded.

\begin{lemma}\label{BLRS3}
  Let\/ $\omega$ be a concentrated weighted comb with\/ $\| \omega
  \|_\infty < \infty$. Then, $\omega$ is translation bounded.
\end{lemma}

\begin{proof}
  Since $\omega$ is concentrated, there exists a $\vL$ with finite
  local complexity so that $\bigl\|\ts \omega -\omega|^{}_{\vL}
  \bigr\|^{}_{K} < 1$.  Since $\vL$ is uniformly discrete, and $\|
  \omega \|_\infty < \infty$, we also have
\[
    \bigl\|\,\omega|^{}_{\vL}\ts \bigr\|^{}_{K} \,< \, \infty \ts .
\]
Therefore,
\[
   \bigl\| \ts\omega  \bigr\|^{}_{K} \,< \,
   \bigl\| \, \omega -\omega|^{}_{\vL} \ts\bigr\|^{}_{K}  +
   \bigl\| \,\omega|^{}_{\vL}\ts \bigr\|^{}_{K}
   \,< \, \infty \ts ,
\]
which completes the proof.
\end{proof}

If we have a measure with weakly uniformly discrete support, we saw
that norm- and sup-almost periodicity are equivalent. If a measure is
concentrated, the support of `most of the measure' has finite local
complexity, a stronger requirement than weak uniform discreteness. We
will see in the next lemma that, under the concentrated assumption,
norm- and sup-almost periodicity are still equivalent.

\begin{lemma}\label{BLRS4}
  Let\/ $\omega$ be a concentrated weighted comb with\/ $\| \omega
  \|_\infty < \infty$. Then, the following statements are equivalent.
  \begin{itemize}
  \item[(i)] $\omega$ is norm-almost periodic.
  \item[(ii)] $\omega$ is sup-almost periodic.
\end{itemize}
\end{lemma}

\begin{proof}
The implication $\text{(i)} \Longrightarrow \text{(ii)}$ is clear.

For $\text{(ii)} \Longrightarrow \text{(i)}$, we repeat the argument
from the proof of Theorem~\ref{BLRS1} in \cite{S-BLRS}.  Let
$\varepsilon >0$. Then, there exists a $\vL$ with finite local
complexity so that
\[
   \bigl\|\, \omega -\omega|^{}_{\vL}\ts \bigr\|^{}_{K}
  \, <\,  \myfrac{\varepsilon}{3} \ts .
\]
Since $\vL$ is uniformly discrete, there exists an $N>0$ so
that,
\[
   \sup_{x \in G} \# \bigl[\vL \cap (x+K)\bigr] \,\leq\, N \ts .
\]
Since $\omega$ is sup-almost periodic, the set
\[
   R_{\frac{\varepsilon}{3N}} \,:=\,  \Bigl\{ t \in G \;\Big|\; \sup_{x \in G}
  \bigl| \omega( \{ -t+x \}) -\omega(\{ x \}) \bigr| <
   \myfrac{\varepsilon}{3N} \Bigr\}
\]
is relatively dense. Let $t \in R_{\frac{\varepsilon}{3N}}$. Then,
\[
  \begin{split}
  \bigl\| \omega - T_t \omega \bigr\|^{}_{K} \,
   &\leq\, \bigl\| \omega - \omega^{}_{\!\vL} \bigr\|^{}_{K} +
      \bigl\| \omega^{}_{\!\vL} - (T_t \omega)^{}_{\!\vL} \bigr\|^{}_{K}  +
     \bigl\| T_t \omega - (T_t\omega)^{}_{\!\vL} \bigr\|^{}_{K} \\
   & <\, \myfrac{\varepsilon}{3} +
     \bigl\| \omega^{}_{\!\vL} - (T_t \omega)^{}_{\!\vL} \bigr\|^{}_{K} +
     \myfrac{\varepsilon}{3} \ts .
\end{split}
\]
Moreover,
\[
  \begin{split}
   \bigl\| \omega^{}_{\!\vL} - (T_t \omega)^{}_{\!\vL} \bigr\|^{}_{K} \,
   &=\, \sup_{x \in G}
    \sum_{y \in \vL \cap (x+K)} \bigl| \omega(y) -\omega(y-t)\bigr| \\
   &\leq\, \myfrac{\varepsilon}{3N}\,
      \sup_{x \in G} \# \bigl[ (x+K) \cap \vL \bigr]
    \, <\,  \myfrac{\varepsilon}{3} \ts .
\end{split}
\]
Combining these, we get $\| \omega - T_t \omega \|_K < \varepsilon$ for
all $t \in R_{\frac{\varepsilon}{3}}$.
\end{proof}

We conclude this section by characterising the concentrated
norm-almost periodic weighted combs.

\begin{theorem}
  Let\/ $\omega := \sum_{x \in \vG} \omega(x)\ts \delta_x$ be so that
  $\| \omega\|_\infty < \infty$.  Then, the following statements are
  equivalent.
  \begin{itemize}
  \item[(i)] $\omega$ is concentrated and norm-almost
    periodic.
  \item[(ii)] There exists a cut and project scheme\/ $(G \times H,
    \cL)$ and a uniformly continuous function\/ $g \in
    C_{\mathsf{u}}(H)$, which is concentrated with respect to\/ $(G \times H,
    \cL)$, so that $\omega = \omega^{}_{g}$.
\end{itemize}
\end{theorem}

\begin{proof}
  The implication $\text{(ii)} \Longrightarrow \text{(i)}$ follows
  immediately from Theorem~\ref{BLRS1} and Lemma~\ref{BLRS2}.

  Now consider $\text{(i)} \Longrightarrow \text{(ii)}$. Since
  $\omega$ is norm-almost periodic, it is sup-almost periodic by
  Lemma~\ref{BLRS4}. Thus, by Theorem~\ref{T1}, there exists a cut and
  project scheme $(G \times H, \cL)$ and a uniformly
  continuous function $g \in C_0(H)$ so that $\omega=\omega_{g}$.

  Since $\omega$ is concentrated, it follows from Lemma~\ref{BLRS2}
  that $g$ is concentrated.
\end{proof}

\section[Generic regular model sets]{A characterisation
of generic regular inter model sets}\label{char rms}

In this section, we try to characterise the generic regular inter model sets
as limits of continuous weighted model combs.

\begin{definition}
  Let $\omega$ and $\nu$ be two weighted Dirac combs.  We say that
  $\omega \preccurlyeq \nu$ if, for each $\varepsilon >0$, there
  exists an $\varepsilon' >0$ so that
\[
   P^{\infty}_{\!\varepsilon'}(\omega) \,\subset\,
   P^{\infty}_{\!\varepsilon}(\nu) \ts .
\]
\end{definition}

The meaning of this definition is that the support function of
$\omega$ is uniformly continuous in the uniformity defined by $\nu$.

\begin{lemma}\cite{S-LR}\label{LRLemma}
  Let\/ $(G \times H, \cL)$ be a cut and project scheme, and
  let\/ $\{ A_n \}$ be any van Hove sequence. Then, there exists a
  Haar measure on\/ $H$ so that, for all\/ $g \in C_{\mathsf{c}}(H)$,
  the following limit exists
\begin{equation}\label{EQ55x}
   M(\omega^{}_{g}) \,:=\,
   \lim_n \frac{\omega^{}_{g}(x+A_n)}{\theta_G(A_n)} \ts ,
\end{equation}
   and satisfies
\[
    M(\omega^{}_{g}) \, =\,  \int_{H} g \dd \theta_{\nts H} \ts .
\]
\end{lemma}
\begin{proof} This result is already proven in \cite{S-LR}. We include here a different proof, based on the theory of almost periodicity and means.

For every $g \in C_{\mathsf{c}}(H)$ the measure $\omega_g$ is strongly almost periodic, hence amenable \cite{S-ARMA,S-MoSt}. Therefore the limit
\[
   \lim_n \frac{\omega^{}_{g}(x+A_n)}{\theta_G(A_n)} \ts .
\]
exists and is equal to the mean  $M(\omega^{}_{g})$.

Now we define a function $L : C_{\mathsf{c}}(H) \to \CC$ by
\[
L(g) \, := \, M(\omega_g) \ts .
\]
It is obvious that $L$ is a linear function, we now prove that it is continuous in the inductive topology of $C_{\mathsf{c}}(H)$:

Let $K \subset H$ be any compact set, and let $K \subset W$ be any regular window. Let $g \in C_{\mathsf{c}}(H)$ with $\supp(g) \subset K$. Then
\[
  \begin{split}
   \bigl| L(g)  \bigr| \,
   &=\, \biggl| \lim_n \dfrac{\omega^{}_{g}(A_n)}{\theta_G(A_n)}  \biggr| \,=\, \biggl| \lim_n \dfrac{\!\sum_{ x \in \oplam(K) \cap A_n }\! g(x^*)}{\theta_G(A_n)} \biggr| \\
    &\leq\,  \limsup_n \dfrac{\sum_{ x \in \oplam(K) \cap A_n } \bigl|\, g(x^*)\bigr| }{\theta_G(A_n)} \\
    &\leq\,  \limsup_n \dfrac{\sum_{ x \in \oplam(K) \cap A_n } \bigl\| g\bigr\|_\infty }{\theta_G(A_n)} \\
     &\leq\,  \limsup_n \| g \|_\infty \dfrac{\# \bigl[  \oplam(K) \cap A_n \bigr] }{\theta_G(A_n)}   \\
     &\leq\,  \limsup_n \| g \|_\infty \dfrac{\# \bigl[  \oplam(W) \cap A_n \bigr] }{\theta_G(A_n)} \, \leq\,   C_K \| g \|_\infty  \ts ,
\end{split}
\]
where $C_K$ denotes the density of the regular model set $\oplam(W)$, and thus does not depend on the choice of $g$.

This shows that $L$ is a continuous linear functional on $C_{\mathsf{c}}(H)$. Therefore, by Riesz representation Theorem, there exists a measure $m$ on $H$ so that
\[
L(g) \,=\, \int_H g d m \, \forall g \in C_{\mathsf{c}}(H) \ts .
\]
Moreover, for all $(x,y) \in \cL$ we have $\omega^{}_{T_{-y}g}=T_{-x} \omega^{}_{g}$. Hence
\[
  \begin{split}
  \int_H g d T_{y} m \,
   &=\,  \int_H T_{-y} g d  m \, = \, M(\omega^{}_{T_{-y} g}) \, = \,  M(T_{-x} \omega^{}_{g})  \\
    &=\, M(\omega^{}_{g}) \,=\,  \int_H g d m
\end{split}
\]
This shows that $T_{y} m=m$ for all $y \in \pi^{}_2(\cL)$. Therefore, $m$ is invariant under translates by a dense subset of $H$, therefore a Haar measure.

This completes the proof.
\end{proof}

In the next theorem, we will need to make an extra assumption on the
cut and project scheme, namely that the internal group is metrisable.

In general, this happens. Indeed, if $\vL$ is a Delone set with contained between two model sets $\oplam(W^\circ) \subset \vL \subset \oplam(W)$, with
$W=\overline{W^\circ}$ and $\theta_H(\partial W)=0$, then the
canonical map $\beta\!:\, \XX(\vL) \longrightarrow \AAA(\vL)$ is
continuous and one-to-one almost everywhere \cite{S-BLM}. But then, by
\cite[Thm.~6]{S-BLM}, $\vL$ agrees with a regular model set up to a
set of density zero. Moreover, the cut and project scheme constructed in the proof of \cite[Thm.~6]{S-BLM} has metrisable internal group $H$.

This shows that any inter model set agrees with a regular model set
with metrisable internal group up to a set of density zero.

\begin{theorem}
  Let\/ $\vL \subset G$ be a Delone set. Then, the following
  statements are equivalent.
  \begin{itemize}
  \item[(i)]$\vL$ is a regular inter model set with metrisable
    internal group.
  \item[(ii)] There exist nets sequences
\[
    0 \,\leq\, \omega^{}_{1} \,\leq\, \omega^{}_{2}
    \,\leq\, \dots \,\leq\, \omega^{}_{n} \,\leq\,
    \dots \,\leq\, \delta^{}_{\!\vL}\,  \leq\, \dots \,
    \leq\, \nu^{}_{n}\, \leq\,\dots \,\leq\, \nu^{}_{1}
\]
   with the following properties:
     \begin{itemize}
     \item[(a)] All\/ $\omega_n$ and\/ $\nu_n$ are strongly almost
       periodic and have Meyer set support;
     \item[(b)] $\lim_n M(\nu^{}_{n} - \omega^{}_{n} ) = 0$;
     \item[(c)] For all\/ $n$ we have\/ $\nu^{}_{n} \preccurlyeq
       \omega^{}_{1}$ and $\omega^{}_{n} \preccurlyeq \omega^{}_{1}$.
\end{itemize}
\end{itemize}
\end{theorem}

\begin{proof}
  We start by proving the implication $\text{(i)} \Longrightarrow
  \text{(ii)}$.  Let $( G \times H , \cL)$ be a cut and
  project scheme, let $d$ be the metric which defines the topology on
  $H$ and let $W$ be a compact set with $W^\circ \neq \varnothing$,
  $\theta_H(\partial W)=0$ and $\oplam(W ^\circ) \subset \vL \subset
  \oplam(W)$.

  We start by constructing two sequences $f_n$ and $g_n$ of functions in $C_{\mathsf{c}}(H)$ such that $f_n$ is increasing to $1_{W^\circ}$, $g_n$ is decreasing to $1_{W}$ and $\int_{H} g_n -f_n < \frac{2}{n}$ for all $n \geq 2$. Moreover, we pick this sequence in such a way that all $f_n$ and $g_n$ are uniformly continuous in the weakest topology which makes $f_1$ continuous. We do this by making sure that the weakest topology which makes $f_1$ continuous coincides with the topology of $H$.

  By the regularity of the Haar measure, for each $n$ we can
  find an open set $V_n$ and a compact set $W_n$ so that $W_n \subset
  W \subset V_n$, $\theta_H( W^\circ \backslash W_n) < \frac{1}{n}$
  and $\theta_H( V_n \backslash W) < \frac{1}{n}$.

  Define $K_n := \bigcup_{j=1}^n W_j$ and $U_n:= \bigcap_{j=1}^n V_j$. Then,
\[
    K_1 \,\subset\, K_2 \,\subset\, \dots \,\subset\, K_n
    \,\subset\, \dots \,\subset\, W^\circ \,\subset\, W
    \,\subset\, \dots \,\subset\, U_n \,\subset\, \dots \,\subset\, U_1 \ts ,
\]
and
\[
  \theta_H( W^\circ \backslash K_n) \, <\,  \tfrac{1}{n} \quad\text{and}\quad
  \theta_H( U_n \backslash W) \, <\, \tfrac{1}{n} \ts .
\]
Let $x^\star \in \pi^{}_{2}(\cL) \cap W^\circ$. Pick an open
neighbourhood $X$ of $x^\star$ so that $\overline{X} \subset W^\circ$,
and let $r >0$ be so that $B_r(x^\star) \subset X$.  Let $f\!:\, H
\longrightarrow \CC$ be defined by
\[
  f(y)\, :=\,  \min \{ d(y, x^\star) , r \} \ts .
\]
Then, $f(y) \equiv r$ outside $X$. Let
\[
   f^{}_{1}(y) \, :=\, \frac{r- f(y)}{r} \ts .
\]
It is easy to check that $f^{}_{1}$ has the following properties,
\[
   f^{}_{1}(x^\star) \, =\,  1\ts,\quad
   0 \,\leq\, f^{}_{1}(y) \,\leq 1\quad\text{and}\quad
    f^{}_{1} \,\equiv\, 0 \,\text{ outside }  X \ts .
\]
In particular, $\supp(f^{}_{1}) \subset W^\circ$. Let now $f'_n, g'_n \in
C_{\mathsf{c}}(H)$ be so that
\[
   1^{}_{K_n} \,\leq\, f'_n \,\leq\, 1^{}_{W^\circ}
   \,\leq\, 1^{}_{\overline{W}} \leq g'_n \,\leq\, 1^{}_{U_n} \ts .
\]
For each $n \geq 2$, we define
\[
  f^{}_n(x) \, = \,
  \max \{ f^{}_{1}(x), f'_2(x), f'_3(x), \dots , f'_n(x) \} \ts ,
\]
while for all $n \geq 1$ we define
\[
   g^{}_n(x) \, =\,  \min \{ g'_1(x), \dots , g'_n(x) \} \ts .
\]
Then $f^{}_n , g^{}_n \in C_{\mathsf{c}}(H)$ with
\[ f^{}_1 \, \leq \,  f^{}_2 \, \leq \,
\dots \,\leq\, f^{}_n \,\leq\, \dots \,\leq\, 1^{}_{W^\circ}  \,\leq\, 1^{}_{W}  \, \leq \dots \,\leq \, g^{}_n \, \leq\, \dots \,\leq\, g^{}_{1} \ts , \]
and, for
all $n \geq 2$, we have
\[
    \int_H f_n \dd \theta_H \, \leq\,
     \theta_H( W^\circ) \, =\,
     \theta_H( \overline{W}) \, \leq\,  \int_H g_n \dd \theta_H
\]
and
\begin{equation}\label{EQ66}
    \int_H (g_n -f_n) \dd \theta_H \, <\,  \myfrac{2}{n} \ts .
\end{equation}
Let $\omega^{}_{n} = \omega^{}_{g^{}_n}$ and $\nu^{}_{n} =
\omega^{}_{f^{}_n}$ be the weighted continuous model combs defined by
these functions. Then property (a) is clear, while (b) follows
immediately from Lemma~\ref{LRLemma} and Eq.~\ref{EQ66}.

We now prove property (c). We show that, for all $h \in C_{\mathsf{c}}(H)$ and for
all $\varepsilon >0$, there exists an $\varepsilon' >0$ so that
\[
   \| T_tf_1-f_1 \|_\infty \,< \,
   \varepsilon' \;\Longrightarrow\;
   \| T_th-h\|_\infty \,<\, \varepsilon \ts.
\]
Then, (c) follows immediately from this and the density of
$\pi^{}_{2}(\cL)$ in $H$.

Let $U:= \{ t \in H \mid \| T_th - h \|_\infty < \varepsilon \}$. Then,
$U$ is an open neighbourhood of zero in $H$. Pick a $\delta$ so that
$B_\delta(0) \subset U$. Pick any $0< \varepsilon'< \min \{ 1,
\frac{\delta}{r} \} $.

We show that $\| T_tf_1-f_1\|_\infty < \varepsilon'$ implies $t \in U$ which proves our claim.

Let $t$ be so that  $\| T_tf_1-f_1\|_\infty < \varepsilon'$. Then we have
\[
  \| T_tf-f\|_\infty \,<\, r\varepsilon' \ts .
\]
Hence
\[
   \bigl| f(-t+x^\star) -f(x^\star) \bigr| \,<\, r\varepsilon' \ts .
\]
Since $f(x^\star)=0$, we have $f(-t+x^\star) \neq r$, and thus
\[
   f(-t+x^\star)\, =\,
    d(-t+x^\star, x^\star) \, <\,  r\varepsilon' \, <\,  \delta \ts .
\]
Therefore
\[
   d(t, 0) \, <\,  \delta \ts ,
\]
which shows that $t \in B_\delta(0) \subset U$. which completes the proof.\smallskip

We now proof that $\text{(ii)} \Longrightarrow \text{(i)}$.  The pair
$(\{ P_\varepsilon(\omega^{}_{1})\}^{}_{0 < \varepsilon <
  \|\omega_1\|_\infty}, G)$ satisfies (A1)--(A5), thus the Baake--Moody construction produces a cut
and project scheme.

Since $\nu^{}_{n} \preccurlyeq \omega^{}_{1}$ and $\omega^{}_{n} \preccurlyeq
\omega^{}_{1}$, by Lemma~\ref{capschemelemma} there exist functions
$f_n, g_n \in C_{\mathsf{u}}(H)$ so that $\omega^{}_n =
\omega^{}_{g^{}_n}$ and $\nu^{}_n=\omega^{}_{f^{}_n}$.  Moreover,
since $\supp(\omega_n)$ and $\supp(\nu_n)$ are Meyer sets, exactly
like in the proof of Theorem~\ref{T2} we get that $f_n, g_n \in
C_{\mathsf{c}}(H)$.

Let $U := \bigcup_n \{ x \in H \mid g_n(x) \neq 0
\}$ and let $W=\left(\overline{U}\right)^\circ$.

We prove that $\oplam(W^\circ) \subset \vL \subset \oplam(W)$ and
$\theta_H(\partial W)=0$.

\emph{Step 1:} $\oplam(U) \subset \vL$.

Let $x \in \oplam(U)$. Then, there
exists an $n$ so that $g_n(x^\star) \neq 0$, and thus $\omega_n(x)
\neq 0$.

Since $0 \leq  \omega_n \leq \delta^{}_{\!\vL}$, we have
$\delta^{}_{\!\vL}(x) \geq \omega_n(x) > 0$.  Hence, $x \in \vL$, and
thus $\vL \supset \oplam(U)$.

\emph{Step 2:} $U$ is an open precompact set.

By definition, $U$ is a union
of open sets, thus open in $H$. Also, it is precompact, since $U \subset
\supp(g^{}_{1})$ which is a compact set.

\emph{Step 3:} $g_n \leq 1^{}_{W} \leq  1^{}_{\overline{W}} \leq f_n$.

First let us note that since $U$ is open and $U \subset \overline{U}$ we have $U \subset W$. Moreover, by definition we have $W \subset \overline{U}$ and therefore $\overline{W} \subset \overline{U}$. Therefore, it suffices to show that $g_n \leq 1^{}_{U} \leq  1^{}_{\overline{U}} \leq f_n$.

Let us first show that $g_n \leq 1_U$. It is clear by the definition
of $U$ that $g_n =0$ outside $U$.  Suppose, by contradiction, that
$g_n \nleq 1^{}_{U}$, then there exists $x \in U$ so that $g_n(x)
>1$. By the continuity of $g_n$ and density of
$\pi^{}_{2}(\cL)$, there exists an $y^\star \in U \cap
\pi^{}_{2}(\cL)$ so that $g_n(y^\star)>1$. Then $\omega_n(y) >1$,
which contradicts $\omega_n < \delta^{}_{\!\vL}$.

We show now that $f_n \geq 1^{}_{\overline{U}}$.  Let $W_n := \{ y \in
H \mid f_n(y) \geq 1 \}$. Then $W_n$ is closed. Moreover, for all $m$ we have $\nu_n \geq
\delta^{}_{\!\vL} \geq \omega_m$, and therefore $W_n$ contains $U \cap
\pi^{}_{2}(\cL)$.

As $U$ is open, $U \cap
\pi^{}_{2}(\cL)$ is dense in $\overline{U}$, and thus
$\overline{U} \subset W_n$. Therefore $f \geq 1$ on $\overline{U}$. Our claim follows now from the fact that $f_n$ is non-negative.

\emph{Step 4:} $\theta_H(\partial(W))=0$.

Let us first observe that, by Step 3, we have $f_n -g_n \geq 1^{}_{\partial(W)}$. Then,
\[
\theta_H(\partial(W)) \, \leq \,  \int_H (f_n -g_n) \dd \theta_H \ts .
\]
Now, by Lemma~\ref{LRLemma}, there exists a constant $C$ given by the choice of the Haar measure on $H$ so that for all $n$ we have
\[
\int_H (f_n -g_n) \dd \theta_H
  \,=\, C M(\nu_n -\omega_n) \ts .
\]
Therefore
\[
  \theta_H(\partial(W)) \,\leq\,
   \lim_n \int_H (f_n -g_n) \dd \theta_H
  \,=\, C \lim_n M(\nu_n -\omega_n) \,=\, 0\ts .
\]

\emph{Step 5:} $\vL$ is a inter regular model set.

We know that
$\oplam(W) \subset \vL \subset \oplam(\overline{W})$ and that
$\theta_H(\partial(W))=0$.

Since $W$ is open, we also obtain $\overline{W} = \overline{W^\circ}$.
\end{proof}

\section{Uniform density in model sets}

In this section, we look at the lower and upper density for pointsets, and to their connection to uniform density. For model sets we get bounds for the  lower and upper density in terms of the Haar measure of the window, which implies the results of \cite{S-MO}. The techniques used are similar to the ones used in \cite[Thm.~1]{S-MO}.

\begin{definition}
  For a Delone set $\vL$ we define
\[
\begin{split}
    \underline{\dens}(\vL) & \, =\,\liminf_{n \to \infty} \inf_{x \in G}
   \frac{ \# [\vL \cap (x + A_n)]}{\theta_G(A_n)} \ts ,\\
   \overline{\dens}(\vL) & \,=\,\limsup_{n \to \infty} \sup_{x \in G}
   \frac{ \# [\vL \cap (x + A_n)]}{\theta_G(A_n)} \ts .
  \end{split}
\]
We call $\underline{\dens}$ and $\overline{\dens}$ the \emph{lower}
and \emph{upper density} of $\vL$; compare \cite[Sec.~2.1]{S-TAO}. The
set $\vL$ is said to have \emph{uniform density} if the limit
\[
  \dens(\vL)\, := \, \lim_{n \to \infty}
   \frac{ \# [\vL \cap (x + A_n)]}{\theta_G(A_n)}
\]
exists uniformly in $x \in G$.
\end{definition}

We start by showing that uniform density is equivalent to the upper and lower density being equal.

\begin{proposition}
  A Delone set $\vL$ has uniform density if and only if\/
  $\underline{\dens}(\vL) = \overline{\dens}(\vL)$.
\end{proposition}

\begin{proof}
  We first prove the $\Longrightarrow$ implication.

  Let $\varepsilon >0$. Since $\dens(\vL):= \lim_{n \to \infty} \frac{ \# [\vL \cap (x
    + A_n)]}{\theta_G(A_n)}$ uniformly in $x \in G$, there exists an
  $N$ so that, for all $n >N$ and all $x \in G$, we have
\[
  \left| \dens(\vL)- \frac{ \# [\vL
  \cap (x + A_n)]}{\theta_G(A_n)} \right| \,<\, \varepsilon \ts .
\]
Thus, for all $n >N$, we have
\[
   \dens(\vL)- \varepsilon \, <\,
    \inf_{x \in G} \frac{ \# [\vL \cap (x + A_n)]}{\theta_G(A_n)}
   \,\leq\,
   \sup_{x \in G} \frac{\# [\vL \cap (x + A_n)]}{\theta_G(A_n)}
   \,<\, \dens(\vL)+\varepsilon \ts ,
\]
and hence,
\[
   \dens(\vL)- \varepsilon \, <\,
   \underline{\dens}(\vL)  \,\leq\,
   \overline{\dens}(\vL) \, <\, \dens(\vL)+ \varepsilon \ts .
\]
As this is true for all $\varepsilon >0$, we get $\underline{\dens}(\vL)  =
   \overline{\dens}(\vL)$ as desired.

Next we prove the $\Longleftarrow$ implication.

Let $\varepsilon >0$ and let us denote $d:= \underline{\dens}(\vL) = \overline{\dens}(\vL)$.  Then, there
exists an $N$ so that, for all $n > N$, we have
\[
   d -\varepsilon \,< \,\inf_{x \in G}
  \frac{ \# [\vL \cap (x + A_n)]}{\theta_G(A_n)}
\]
and
\[
  \sup_{x \in G} \frac{ \# [\vL
   \cap (x + A_n)]}{\theta_G(A_n)} \,<\, d +\varepsilon \ts .
\]
Thus, for all $n >N$ and all $x \in G$, we have
\[
   d -\varepsilon \,<\,
   \frac{ \# [\vL \cap (x +A_n)]}{\theta_G(A_n)}
   \,<\, d +\varepsilon \ts ,
\]
and hence $\vL$ has uniform density.
\end{proof}

Now we will show that for arbitrary model sets, there exists a choice of Haar measure on $H$ for which, the lower and upper density of the model set is bounded from below by the measure of the interior of the window, respectively bounded from above by the measure of the closure of the window.

\begin{proposition}\label{sudinms}
  Let\/ $( G \times H , \cL)$ be a cut and project scheme and
  let\/ $\vL = \oplam(W)$ be a model set. Then, with the Haar measure on\/
  $H$ from Lemma~$\ref{LRLemma}$,
\[
    \theta_H( W^\circ ) \,\leq\, \underline{\dens}(\vL) \,\leq\,
   \overline{\dens}(\vL) \,\leq\, \theta_H( \overline{W} ) \ts .
\]
\end{proposition}

\begin{proof}
  Let $\varepsilon >0$.

  Pick $f, g \in C_{\mathsf{c}}(H)$ so that $f \leq
  1^{}_{W^\circ} \leq 1^{}_{\overline{W}} \leq g$ and such that $\int_H f \dd \theta_{\nts
    H} \geq \theta_H( W^\circ ) -\varepsilon$ and $\int_H g \dd
  \theta_{\nts H} \leq \theta_H( \overline{W} ) +\varepsilon$.  Then, by Lemma~\ref{LRLemma} we have
\[
    \lim_{n \to \infty}
    \frac{ \omega_{f} (x + A_n)}{\theta_G(A_n)}
    \,= \,\int_H f \dd \theta_{\nts H} \ts ,
\]
and
\[
   \lim_{n \to \infty}
   \frac{ \omega_{g} (x + A_n)}{\theta_G(A_n)}
    \,=\, \int_H g \dd \theta_{\nts H}  \ts .
\]
uniformly in $x$.

Since $\omega^{}_{f} \leq \delta^{}_{\!\vL} \leq \omega^{}_{g}$,  for all $n$
and all $x$ we have
\[
  \frac{ \omega^{}_{f} (x + A^{}_n)}{\theta^{}_G(A^{}_n)} \,\leq\,
  \frac{\delta^{}_{\!\vL} (x + A^{}_n)}{\theta^{}_G(A^{}_n)} \,\leq\,
   \frac{ \omega^{}_{g}(x + A^{}_n)}{\theta^{}_G(A^{}_n)} \ts ,
\]
and thus
\[
  \theta_{\nts H}( W^\circ ) -\varepsilon  \,\leq\,
   \int_H \! f \dd \theta_{\nts H}  \,\leq\,
   \underline{\dens}(\vL) \,\leq\,
    \overline{\dens}(\vL) \,\leq\,
    \int_H \! f d \theta_{\nts H}  \,\leq\,
     \theta_H( \overline{W} )  +\varepsilon \ts ,
\]
which completes the proof.
\end{proof}

\begin{corollary}
  If\/ $\vL$ is a regular model set, then\/ $\vL$ has uniform
  density.
\end{corollary}

Next let us look at an interesting example from \cite{S-BaMoPl,S-BM,S-PH}.

\begin{example}\label{nudnvlp}
   For each prime $p$, let $\ZZ_p$ be the ring of
  $p$-adic integers, and let $\QQ_p$ be the field of $p$-adic
  numbers. Let $\AAA_{\QQ}$ be the rationale adele ring.

  Then $\QQ$ is a   lattice in $\AAA_{\QQ}$, and we have a natural cut and project
  scheme
\begin{equation}
\begin{array}{ccccc}
   \RR^2 & \stackrel{\pi^{}_{1}}{\longleftarrow} &
  \RR^2 \times \prod_p \QQ_p^2 & \stackrel{\pi^{}_{2}}
  {\longrightarrow} & \prod_p \QQ_p^2  \\
 && \bigcup \\
 &&\QQ^2
\end{array}
\end{equation}
where the product $\prod_p \QQ_p^2$ is the restricted product. Let
$\Omega := \prod_p \ZZ_p^2$ and let $V:=\prod_p (\ZZ_p^2 \backslash
p\ZZ_p^2)$. Then, $V$ is an window which defines the visible points
of $\ZZ^2$. It is compact and has no interior.

Let $W := \Omega \backslash V$. Then $W$ is open and $\overline{W} =
\Omega$.  The set $\oplam(W)$ is the set of invisible points of
$\ZZ^2$. By construction, $\oplam(W)$ is a model set.

Since the visible points of $\ZZ^2$ form a subset of $\ZZ^2$ which has
arbitrary large holes, and since $\oplam(W)$ is its complement in
$\ZZ^2$, it is clear that
\[
   \overline{\dens}(\oplam(W)) \,=\, \dens (\ZZ^2) \,=\,
   \theta^{}_{\prod_p \!\QQ_p^2} ( \Omega) \,=\, \theta (\overline{W}) \ts .
\]
Also, Baake and Moody \cite{S-BM} showed that, with respect to some van
Hove sequence, the density of the visible lattice points is exactly
$\theta^{}_{\prod_p\! \QQ_p^2} (V)$. This shows that
\[
   \overline{\dens}(\oplam(V)) \,\geq\, \theta (V) \ts ,
\]
and hence, by the uniform density of $\ZZ^2 = \oplam(\Omega)$, we get
\[
  \underline{\dens}(\oplam(W)) \,\leq\, \theta (W) \, =\, \theta(W^\circ) \ts .
\]
This example shows that, in general, the inequalities in
Proposition~\ref{sudinms} cannot be improved. For some more recent
developments, we refer to \cite{S-PH,S-HB} and references
therein. \exend
\end{example}

An interesting question is whether, under some simple extra
assumptions, we can get $ \theta_{\nts H}( W^\circ )=
\underline{\dens}(\vL)$ and $\overline{\dens}(\vL) = \theta_{\nts H}(
\overline{W} )$. Any such result would allow one to characterise the
property $\theta(\partial W)=0$, which is usually the harder part in
any characterisation of regular model sets, in terms of uniform
density of $\oplam(W)$. If $\oplam(W)$ is uniquely ergodic, it looks
reasonable that we might get the desired relations.

\begin{conjecture}\label{CJ1}
  Let\/ $(G \times H , \cL)$ be a cut and project scheme and
  let\/ $W \subset H$ be a compact set. If\/ $\oplam(W)$ is uniquely ergodic,
  then
\[
   \theta_{\nts H}( W^\circ ) \,=\, \underline{\dens}(\vL) \quad\text{and}\quad
  \overline{\dens}(\vL) \,=\,\theta_{\nts H}(W ) \ts .
\]
\end{conjecture}

\begin{remark}
  Note that, if Conjecture~\ref{CJ1} holds, then the following are
  equivalent.
\begin{itemize}
\item[(i)] $\vL$ is a statistical regular model set.
\item[(ii)]
\begin{itemize}
\item[(a)]
 $\vL$ is Meyer set;
\item[(b)]
 $\vL$ is pure point diffractive;
 \item[(c)]
 $\vL$ is uniquely ergodic;
 \item[(d)]
 $\vL$ has uniform density. \exend
 \end{itemize}
 \end{itemize}
\end{remark}

\section{Concluding remarks}

The Baake--Moody construction of a cut and project scheme is becoming a powerful tool in the study of model sets and Meyer sets, with applications in many seemingly unrelated areas. Introduced in \cite{S-BM}, it was a key ingredient in the classification of regular model sets via dynamical systems \cite{S-BLM}. Also, in our current work, it was the key to both the classification of norm- and sup-almost periodic measures in terms of cut and project schemes, and the classification of Meyer sets in terms of $\epsilon$ dual characters. The appearance of this construction in so many difference situations is surprising, and one can only wonder what other applications it might have in the future, and if there might be a deeper reason why this construction can be used in so many different situations.

The classification of norm and sup almost periodic weighted Dirac combs turned out to be surprisingly simple, but it also leads to some new interesting questions. For example, given a CPS $(G \times H, \cL)$, it would be very helpful to find all the functions $g \in C_0(H)$ which define translation bounded measures $\omega_g$ on $G$.

Also, given an arbitrary norm almost periodic measure $\mu$ with locally finite sets of almost periods, we can apply the Baake--Moody construction on the sets $P^K_\varepsilon(\mu)$. In this CPS, the sets $P^K_\varepsilon(\mu)$ are model sets, but it is currently unclear if one can recover the measure $\mu$ from the CPS.

While strongly almost periodicity is very important for long range order because of the connection with pure point diffraction, this work suggests that norm- and sup-almost periodicity might be more suitable for sets and measures produced by cut and project schemes. Therefore, norm almost periodicity becomes a very interesting and useful concept: it is stronger than strongly almost periodicity, hence implies pure point diffractive, and it also leads to the construction of a CPS. Moreover, for measures supported inside Meyer sets, strongly and norm almost periodicity are equivalent.

A famous conjecture of Lagarias asked if given a point set $\vL$ which is repetitive, has FLC and pure point diffraction, does it follow that $\vL$ is a Meyer set? This conjecture was disproved recently by Kellendonk and Sadun \cite{S-KS}. Our above discussion suggests that to obtain the Meyer condition, norm-almost periodicity of the autocorrelation is a more natural assumption than pure point diffraction. This leads to the following question: given a point set $\vL$ which is repetitive, has FLC and whose autocorrelation is norm-almost periodic, does it follow that $\vL$ is a Meyer set? We don't know if the example of Kellendonk and Sadun answers to this new question.

Using the Baake--Moody construction and properties of almost periodic measures we were able to extend the classification of Meyer sets from $\RR^d$ to arbitrary $\sigma$-compact LCAG, excepting the Lagarias condition $\vL-\vL$ is uniformly discrete. We had to use instead the seemingly stronger condition $\vL-\vL-\vL$ is locally finite. It is currently unclear if this stronger condition is indeed needed for non-compactly generated LCAG groups, or if maybe the methods we used here are not strong enough to prove this equivalence in general.

Finally, the Baake--Moody construction allowed us characterize the regular inter model sets in Section~\ref{char rms} by approximations by almost periodic measures. While the characterization is based on the standard method of approximating the window from below and above by continuous functions, we needed to add an extra condition of all almost periodic measures being continuous in the topology defined by one of them. This condition was necessarily to be able to get all functions in the same cut and project scheme, but it looks artificial, and we expect that this characterisation can be improved in the future.

\section*{Acknowledgments}

I am grateful to Michael Baake, Daniel Lenz, Robert Moody, Uwe Grimm and Ian
Putnam for various remarks and suggestions which improved this
article. Part of this work was done while the author was a
Postdoctoral Fellow at the University of Victoria and was supported by
the Natural Sciences and Engineering Research Council of Canada, and I am grateful to NSERC for their support.
 The
manuscript was completed while the author visited University
of Bielefeld and was partially supported by the German Research
Foundation (DFG), within the CRC 701.

\section*{Appendix. Harmonious sets}

In this appendix, we review Meyer's proof \cite{S-MEY} of
Theorem~\ref{harm} in the context of $\sigma$-compact LCAGs.  Meyer's
original proof was in the context of separable metrisable groups, but
we will see that his proof works in the more general context of
$\sigma$-compact LCAGs. Our proof follows closely the proof from
\cite{S-MOO}.

For the entire section we will denote by $(\widehat{G}\ts )_{\text{b}}$ the dual group of $G_d$.
As the dual of a discrete group, this group is compact, and it is usually called the \emph{Bohr compactification} of $\widehat{G}$  \cite{S-MoSt}.
We should note here that Meyer \cite{S-MEY} and Moody \cite{S-MOO} use the notation $(\widehat{G}\ts )_{\text{alg}}$ for the Bohr compactification of $G$, as it represents the group of algebraic characters on $G$. We prefer the notation $(\widehat{G}\ts )_{\text{b}}$, as it is the standard notation in this book.

\begin{theorem}\label{harm}
  Let\/ $\vL \subset G$ and\/ $\varepsilon>0$. Then, $\vL^\varepsilon$ is
  relatively dense if and only if $\vL$ is\/ $\varepsilon$-harmonious.
\end{theorem}

\begin{proof}
  For each $\varepsilon >0$ we define
\[
  V^\varepsilon(\vL) \, :=\,
  \bigl\{ \chi \in (\widehat{G}\ts )_{\text{b}} \;\big|\;
    \left| \chi(x) -1 \right| \leq \varepsilon
    \text{ for all } x \in \vL\bigr\} .
\]
Then, $\vL$ is $\varepsilon$-harmonious if and only if
\[
   V^\varepsilon(\vL) + \widehat{G}\, =\, (\widehat{G}\ts )_{\text{b}} \ts .
\]

We first prove the $\Longrightarrow$ implication.

Since $G$ is a $\sigma$-compact LCAG, so is $\widehat{G}$.  Let $K \subset \widehat{G}$ be a compact set so
that $\vL^\varepsilon+K = \widehat{G}$.  Using $\vL^\varepsilon =
V^\varepsilon(\vL) \cap \widehat{G}$ we get that
\[
   \widehat{G} \,\subset\,
   V^\varepsilon(\vL) + K \,\subset\,
   (\widehat{G}\ts )_{\text{b}} \ts .
\]
Since $K$ is compact in $\widehat{G}$ and since the embedding
$\widehat{G} \hookrightarrow (\widehat{G}\ts )_{\text{b}}$ is
continuous, $K$ is also compact in $(\widehat{G}\ts )_{\text{b}}$.

By definition,
$V^\varepsilon(\vL)$ is closed hence compact in $(\widehat{G}\ts )_{\text{b}}$. Therefore $V^\varepsilon(\vL) + K$ is compact, hence closed in $(\widehat{G}\ts )_{\text{b}}$.

Therefore $V^\varepsilon(\vL) + K$ is closed in $(\widehat{G}\ts )_{\text{b}}$ and contains the dense
subset $\widehat{G}$. This implies that
\[
   V^\varepsilon(\vL) + K \, =\,  (\widehat{G}\ts )_{\text{b}} \ts ,
\]
which shows our implication.

We now prove the $\Longleftarrow$ implication.

As $\widehat{G}$ is $\sigma$-compact, we can find a sequence $K_n$ of compact sets so that $0 \in
K_1$, $K_n \subset K_{n+1}^\circ$ and $\widehat{G} =\bigcup_n K_n$. Then,
\[
   (\widehat{G}\ts )_{\text{b}}  \,=\, V^\varepsilon(\vL)+ \widehat{G} \,=\,
    \bigcup_n \bigl( V^\varepsilon(\vL)+ K_n\bigr)
     \ts .
\]
Again the continuity of $\widehat{G} \hookrightarrow
(\widehat{G}\ts )_{\text{b}}$ implies that $K_n$ is also compact in
$(\widehat{G}\ts )_{\text{b}}$. Then $V^\varepsilon(\vL)+ K_n$ is
compact, hence closed in $(\widehat{G}\ts )_{\text{b}}$.

Since $(\widehat{G}\ts )_{\text{b}}$ is a locally compact Haussdorff space, by
the Baire category theorem \cite{S-BOURB} there exists an $n$ so that
$V^\varepsilon(\vL)+ K_n$ has non-empty interior. Therefore, as the embedding $\widehat{G} \hookrightarrow (\widehat{G}\ts )_{\text{b}}$ has
dense image,
\[
 \widehat{G} +\bigl(V^\varepsilon(\vL) + K_n\bigr) \,=\,
 (\widehat{G}\ts )_{\text{b}} \ts .
\]
As the left-hand side is an open cover of the compact set
$(\widehat{G}\ts )_{\text{b}}$ by translates of the set open set $\left(V^\varepsilon(\vL)+ K_n\right)^\circ$  there exists a finite set $F \subset
\widehat{G}$ so that
\[
  F+ \bigl(V^\varepsilon(\vL) + K_n\bigr) \, =\,
  (\widehat{G}\ts )_{\text{b}} \ts .
\]

Let $K= F+ K_n$. Then $K$ is compact in $\widehat{G}$ and we will show
that
\[
   \vL^\varepsilon +K \, =\, \widehat{G} \ts .
\]
Indeed, let $\chi \in \widehat{G}$. Then we can write $\chi =\psi+\phi+\nu$ with
$\psi \in F$, $\phi \in V^\varepsilon(\vL)$ and $\nu \in K_n$.

By the definition of $K$ we have $\psi+\nu \in K$. Moreover, $\phi = \chi - \psi-\nu$, thus $\phi \in
\widehat{G}$.

This shows that
\[
  \chi \, =\, \phi+(\psi+\nu) \ts ,\quad
   \phi \,\in\,  \widehat{G}\cap V^\varepsilon(\vL)\, =\,
    \vL^\varepsilon \ts ,\quad \psi+\nu \,\in\, K \ts ,
\]
which completes the proof.
\end{proof}

\begin{corollary}
  Let\/ $\vL$ be relatively dense in\/ $G$. Then, $\vL$ is harmonious
  if and only if\/ $\vL^\varepsilon$ is relatively dense for all\/
  $\varepsilon>0$.\qed
\end{corollary}

\end{document}